\documentclass[journal]{IEEEtran}

\usepackage{setspace, dsfont}
\usepackage{amssymb, amsmath, mathrsfs, amsthm, stmaryrd, upgreek, mathtools}
\usepackage{pgfplots, caption, tikz, subfigure, graphicx}
\usepackage{url, hyperref}
\usepackage[compress]{cite}
\usepackage{bbm}
\usetikzlibrary{spy}
\usepgflibrary{arrows} 
\usepackage{varwidth}
\def\tablename{Table}

\usepackage{graphicx,wrapfig,tikz}
\usetikzlibrary{calc,hobby} 

\usepackage{scalerel,stackengine}
\stackMath
\newcommand\reallywidehat[1]{%
\savestack{\tmpbox}{\stretchto{%
  \scaleto{%
    \scalerel*[\widthof{\ensuremath{#1}}]{\kern-.6pt\bigwedge\kern-.6pt}%
    {\rule[-\textheight/2]{1ex}{\textheight}}
  }{\textheight}%
}{0.5ex}}%
\stackon[1pt]{#1}{\tmpbox}%
}

\makeatletter
\def\@IEEEsectpunct{.\ \,}
\makeatother

\usepackage[font=footnotesize]{caption}

\newcommand{\Rd}{\mathbb{R}^d}

\newcommand{\eps}{\varepsilon}

\newcommand{\N}{\mathbb{N}}

\newcommand{\R}{\mathbb{R}}

\newcommand{\QQ}{\mathcal{Q}}

\newcommand{\la}{\lambda}

\newcommand\mydots{\ifmmode\ldots\else\makebox[1em][c]{.\hfil.\hfil.}\fi}
\newcommand{\VV}{\Vert}

\newcommand{\Z}{\mathbb{Z}}

\DeclareMathOperator{\Real}{Re}
\DeclareMathOperator{\Imag}{Im}

\def\ba#1\ea{\begin{align*}#1\end{align*}}	
\def\ban#1\ean{\begin{align}#1\end{align}}	
\def\bac#1\eac{\vspace{\abovedisplayskip}{\par\centering$\begin{aligned}#1\end{aligned}$\par}\addvspace{\belowdisplayskip}}	

\newtheorem{theorem}{Theorem}

\newtheorem{lemma}{Lemma}
\newtheorem{assumption}{Assumption}
\newtheorem{definition}{Definition}
\newtheorem{proposition}{Proposition}
\newtheorem{corollary}{Corollary}

\newtheorem{remark}{Remark}

\usepackage{endnotes}

\begin{document}
\title{Energy Propagation in \\ Deep Convolutional Neural Networks}

\author{Thomas~Wiatowski,
        Philipp~Grohs,
        and~Helmut~B\"olcskei,~\IEEEmembership{Fellow,~IEEE}
\thanks{T. Wiatowski and H. B\"olcskei are with the Department of Information Technology and Electrical Engineering, ETH Zurich,  Switzerland. Email:~\{withomas,~boelcskei\}@nari.ee.ethz.ch}
\thanks{P. Grohs is with the Faculty of Mathematics, University of Vienna, Austria. Email: philipp.grohs@univie.ac.at}
\thanks{The material in this paper was presented in part at the 2017 IEEE International Symposium on Information Theory (ISIT), Aachen, Germany. }
\thanks{Copyright (c) 2017 IEEE. Personal use of this material is permitted.  However, permission to use this material for any other purposes must be obtained from the IEEE by sending a request to pubs-permissions@ieee.org.}
}

\maketitle

\begin{abstract}
Many practical machine learning tasks employ very deep convolutional neural networks. Such large depths pose  formidable computational challenges in training and operating the network. It is therefore important to understand how fast the energy contained in the propagated signals (a.k.a. feature maps) decays across layers. In addition, it is desirable that the feature extractor generated by the network be informative in the sense of the only signal mapping to the all-zeros feature vector being the zero input signal. This ``trivial null-set''   property  can be accomplished by asking for ``energy conservation'' in the sense of the energy in the  feature vector being proportional to  that of the corresponding input signal. This paper establishes conditions for energy conservation (and thus for a trivial null-set) for a wide class of deep convolutional neural network-based feature extractors and characterizes corresponding feature map energy decay rates. Specifically, we consider general scattering networks employing the modulus non-linearity and we find  that under mild analyticity and high-pass conditions on the filters (which encompass, inter alia, various constructions of Weyl-Heisenberg filters, wavelets, ridgelets, ($\alpha$)-curvelets, and shearlets) the feature map energy decays at least polynomially fast. For broad families of wavelets and Weyl-Heisenberg filters, the guaranteed decay rate is shown to be exponential.  Moreover,  we provide handy estimates of the number of layers needed to have at least $((1-\varepsilon)\cdot 100)\%$ of the input signal energy be contained in the  feature vector.
\end{abstract}


%
\IEEEpeerreviewmaketitle

\section{Introduction}
%
%
%
%
\IEEEPARstart{F}{eature} extraction based on deep convolutional neural networks (DCNNs) has been applied with significant success in a wide range of practical machine learning tasks \cite{Nature,Goodfellow-et-al-2016,Bengio,LeCunNIPS89,Rumelhart,LeCunProc}. Many of these applications, such as, e.g.,  the classification of images in the ImageNet data set, employ very deep networks with potentially hundreds of layers \cite{he2015deep}. Such network depths entail formidable computational challenges in the training phase due to the large number of parameters to be learned, and in operating the network due  to the large number of convolutions that need to be carried out. It is therefore paramount to understand  how fast the energy contained in the signals generated in the individual network layers, a.k.a. feature maps, decays across layers. In addition, it is important that the feature vector---obtained by aggregating filtered versions of the feature maps---be informative in the sense of the only signal mapping to the all-zeros feature vector being the zero input signal. This ``trivial null-set''   property for the feature extractor can be obtained by asking for the energy in the  feature vector being proportional to  that of the corresponding input signal, a property we shall refer to as ``energy conservation''.

Scattering networks as introduced in \cite{MallatS} and extended in \cite{Wiatowski_journal} constitute an important class of feature extractors based on nodes that implement convolutional transforms with pre-specified or learned filters in each network layer (e.g., wavelets \cite{MallatS,Waldspurger}, uniform covering filters \cite{czaja2016uniform}, or general filters \cite{Wiatowski_journal}), followed by a non-linearity (e.g., the modulus  \cite{MallatS,Waldspurger,czaja2016uniform}, or a general Lipschitz non-linearity \cite{Wiatowski_journal}), and a pooling ope\-ration (e.g., sub-sampling or average-pooling \cite{Wiatowski_journal}).  Scattering network-based feature extractors were shown to yield classification performance competitive with the state-of-the-art on various data sets \cite{Bruna,Anden,Sifre,wiatowski2016discrete,CInC,ICASSP2016}. Moreover, a  mathematical theory exists, which allows to establish formally that   such  feature extractors are---under certain technical conditions---horizontally \cite{MallatS} or vertically \cite{Wiatowski_journal} translation-invariant and deformation-stable in the sense of \cite{MallatS}, or exhibit limited sensitivity to deformations in the sense of \cite{Wiatowski_journal} on input signal classes such as band-limited functions \cite{Wiatowski_journal,balan2017lipschitz}, cartoon functions \cite{grohs_wiatowski}, and Lipschitz functions \cite{grohs_wiatowski}.

It was shown recently that the energy in the feature maps generated by scattering networks employing, in every network layer, the same set of (certain) Parseval wavelets \cite[Section 5]{Waldspurger} or ``uniform covering'' \cite{czaja2016uniform}  filters (both satisfying  analy\-ticity and vanishing moments  conditions), the modulus non-linearity, and no pooling, decays at least exponentially fast and ``strict'' energy conservation (which, in turn, implies a trivial null-set) for the infinite-depth feature vector holds. Specifically,  the feature map energy decay was shown to be at least of order $\mathcal{O}(a^{-N})$, for some \textit{unspecified} $a>1$, where $N$ denotes the network depth. 
  We note that $d$-dimensional uniform covering filters as introduced in \cite{czaja2016uniform} are  functions  whose   Fourier transforms' support sets can be covered by a union of finitely many balls. This covering condition is satisfied by, e.g., Weyl-Heisenberg filters  \cite{Groechenig} with a band-limited prototype function, but fails to hold for multi-scale filters such as wavelets \cite{Daubechies,MallatW}, $(\alpha)$-curvelets \cite{CandesNewTight,CandesDonoho2,Grohs_alpha}, shearlets \cite{OriginShearlets,Shearlets}, or ridgelets \cite{Ridgelet,DonohoCandesRidgelet,Grohs_transport}, see \cite[Remark 2.2 (b)]{czaja2016uniform}.

\paragraph*{Contributions} 
The first main contribution of this paper is a characterization of the feature map energy decay rate in DCNNs employing the modulus non-linearity, no pooling, and \textit{general} filters that constitute a frame \cite{Antoine,Christensen,Kaiser,Daubechies}, but not necessarily a Parseval frame, and are allowed to be different in different network layers. We find that, under mild analyticity and high-pass conditions on the filters, the energy decay rate is at least polynomial in the network depth, i.e., the decay is at least of order $\mathcal{O}(N^{-\alpha})$, and we \emph{explicitly} specify the decay exponent $\alpha>0$. This result encompasses, inter alia, various constructions of Weyl-Heisenberg filters, wavelets, ridgelets, ($\alpha$)-curvelets, shearlets, and learned filters (of course as long as the learning algorithm imposes the analyticity and high-pass conditions   we require). For broad families of wavelets and Weyl-Heisenberg filters, the guaranteed energy decay rate is shown to be exponential in the network depth, i.e., the decay is at least of order $\mathcal{O}(a^{-N})$ with the decay factor given as $a=\frac{5}{3}$  in the wavelet case and $a=\frac{3}{2}$ in the Weyl-Heisenberg case. We hasten to add that our results constitute \emph{guaranteed} decay rates and do not preclude the energy from decaying faster in practice. 

Our second main contribution shows that the energy decay results above are compatible with a trivial null-set for finite- and infinite-depth networks. Specifically, this is accomplished by establishing energy proportionality between the feature vector and the underlying input signal with the proportionality constant lower- and upper-bounded by the frame bounds of the filters employed in the different layers.  We show that this energy conservation result is a consequence of a demodulation effect induced by the modulus non-linearity in combination with the analyticity and high-pass properties of the filters. Specifically, in every network layer, the modulus non-linearity moves the spectral content of each individual feature map to base-band (i.e., to  low frequencies), where it is subsequently extracted (i.e., fed into the feature vector) by a low-pass output-generating filter. 

Finally, for input signals that belong to the class of Sobolev functions\footnote{A wide range of practically relevant signal classes are Sobolev functions, for example,  band-limited functions and---as established in the present paper---cartoon functions \cite{Cartoon}. We note that cartoon functions are widely used in the mathematical signal processing literature \cite{grohs_wiatowski,wiatowski2016discrete,Grohs_alpha,grohs2014parabolic,ShearletsIntro} as a model for natural images such as, e.g., images of handwritten digits \cite{MNIST}.}, our energy decay and conservation results are shown to yield handy estimates of the number of layers needed to have at least $((1-\varepsilon)\cdot 100)\%$ of the input signal energy be contained in the feature vector. For example, in the case of exponential energy decay with $a=\frac{5}{3}$ and for band-limited input signals, only $8$ layers are needed to absorb $95\%$ of the input signal's energy.

We emphasize that throughout  energy decay results pertain to the feature maps, whereas  energy conservation statements apply to the feature vector, obtained by aggregating filtered versions of the feature maps.


\paragraph*{Notation}
The complex conjugate of $z \in \mathbb{C}$ is denoted by $\overline{z}$. We write $\Real(z)$ for the real, and $\Imag(z)$ for the imaginary part of $z \in \mathbb{C}$. The Euclidean inner product of $x,y \in \mathbb{C}^d$ is $\langle x, y \rangle:=\sum_{i=1}^{d}x_i \overline{y_i}$, with associated norm $|x|:=\sqrt{\langle x, x \rangle}$. For $x\in \mathbb{R}$, $(x)_{+} := \max\{0,x\}$ and $\langle x\rangle := (1+|x|^2)^{1/2}$. We denote the open ball of radius $r>0$ centered at $x\in \Rd$ by $B_r(x)\subseteq \Rd$. The first canonical orthant is  $H:=\{ x\in \Rd \ | \ x_k\geq 0,  \ k=1,$\mydots$,d\}$, and we define the rotated orthant $H_A:=\{ Ax \ | \ x\in H \}$  for $A\in O(d)$, where $O(d)$ stands for the orthogonal group of dimension $d\in \mathbb{N}$. The Minkowski sum of sets $A,B\subseteq \R^d$ is $(A+B):=\{ a+b \, | \, a\in A, \ b \in B \}$, and $A\Delta B:=(A\backslash B)\cup (B\backslash A)$ denotes their symmetric difference.  A multi-index $\alpha=(\alpha_1,\dots,\alpha_d) \in \mathbb{N}_0^d$ is an ordered $d$-tuple of non-negative integers $\alpha_i \in \mathbb{N}_0$. 

 For functions $W:\mathbb{N}\to \mathbb{R}$ and $G:\mathbb{N}\to \mathbb{R}$, we say that $W(N)=\mathcal{O}(G(N))$ if  there exist $C>0$ and $N_0\in \mathbb{N}$ such that $W(N)\leq CG(N)$, for all $N\geq N_0$. The support $\text{supp}(f)$ of a function $f:\Rd \to \mathbb{C}$ is the closure of the set $\{ x\in \Rd \ | \ f(x)\neq 0 \}$ in the topology induced by the Euclidean norm $|\cdot|$. For a Lebesgue-measurable function $f:\Rd \to \mathbb{C}$, we write $\int_{\Rd} f(x) \mathrm dx$ for its integral w.r.t. Lebesgue measure. The indicator function of a set $B\subseteq \Rd$ is defined as $\mathds{1}_B(x)=1$, for $x\in B$, and $\mathds{1}_B(x)=0$, for $x\in \Rd\backslash B$.  For a measurable set $B\subseteq \R^d$, we let $\mbox{vol}^{d}(B):=\int_{\R^d}\mathds{1}_B(x)\mathrm dx=\int_{B}1\mathrm dx$, and we write  $\partial B$ for its boundary.  $L^p(\Rd)$, with $p \in [1,\infty)$, stands for the space of Lebesgue-measurable functions $f:\Rd \to \mathbb{C}$ satisfying $\| f \|_p:= (\int_{\Rd}|f(x)|^p \mathrm dx)^{1/p}<\infty.$ $L^\infty(\Rd)$ denotes the space of Lebesgue-measurable functions $f:\Rd \to \mathbb{C}$ such that  $\| f \|_\infty:=\inf \{\alpha>0 \ | \ |f(x)| \leq \alpha \text{ for a.e.}\footnote{Throughout the paper ``a.e.''  is w.r.t. Lebesgue measure.}\  x \in \Rd \}<\infty$.  For a countable set $\QQ$, $(L^2(\Rd))^\QQ$ stands for the space of sets $S:=\{ f_q \}_{q\,\in\, \QQ}$, with $f_q \in L^2(\Rd)$ for all $q \in \QQ$, satisfying $||| S |||:=(\sum_{q \in \QQ} \VV f_q\VV_2^2)^{1/2}<\infty$. We denote the Fourier transform of $f \in L^1(\Rd)$ by $\widehat{f}(\omega):=\int_{\Rd}f(x)e^{-2\pi i \langle x,  \omega \rangle }\mathrm dx$ and extend it in the usual way to $L^2(\Rd)$ \cite[Theorem 7.9]{Rudin}. 
    
$\text{Id}: L^p(\Rd) \to L^p(\Rd)$ stands for the identity operator on $L^p(\Rd)$. The convolution of $f\in L^2(\Rd)$ and $g\in L^1(\Rd)$ is $(f\ast g)(y):=\int_{\Rd}f(x)g(y-x)\mathrm dx$. We write $(T_tf)(x):=f(x-t)$, $t \in \Rd$, for the translation operator, and $(M_\omega f)(x):=e^{2\pi i \langle x , \omega\rangle }f(x)$, $\omega \in \Rd$, for the modulation operator. We set $\langle f,g \rangle :=\int_{\Rd}f(x)\overline{g(x)}\mathrm dx$, for $f,g \in L^2(\Rd)$. 

$H^s(\Rd)$, with $s>0$,  stands for the Sobolev space of functions $f\in L^2(\Rd)$ satisfying $ \| f\|_{H^s}:=(\int_{\Rd}|\widehat{f}(\omega)|^2(1+|\omega|^2)^{s}\mathrm d\omega)^{1/2}<\infty$,   see \cite[Section 6.2.1]{Grafakos2}. Here, the index $s$   reflects the “degree” of smoothness of $f\in H^s(\Rd)$, i.e., larger $s$ entails smoother $f$. For a multi-index $\alpha \in \mathbb{N}_0^d$, $D^\alpha$ denotes the differential operator $D^\alpha:=(\partial/\partial x_1)^{\alpha_{1}}\dots (\partial/\partial x_d)^{\alpha_{d}}$, with order $|\alpha|:=\sum_{i=1}^d \alpha_i$. The space of functions $f:\Rd \to \mathbb{C}$ whose derivatives $D^\alpha f$ of order at most $k\in \mathbb{N}_0$ are continuous is designated by $C^k(\Rd,\mathbb{C})$. Moreover, we denote the gradient of a function $f:\Rd \to \mathbb{C}$ as $\nabla f$.


\begin{figure*}[t!]
\centering
\begin{tikzpicture}[scale=2.5,level distance=10mm,>=angle 60]

  \tikzstyle{every node}=[rectangle, inner sep=1pt]
  \tikzstyle{level 1}=[sibling distance=30mm]
  \tikzstyle{level 2}=[sibling distance=10mm]
  \tikzstyle{level 3}=[sibling distance=4mm]
  \node {$f$}
	child[grow=90, level distance=.45cm] {[fill=gray!50!black] circle (0.5pt)
		child[grow=130,level distance=0.5cm] 
        		child[grow=90,level distance=0.5cm] 
        		child[grow=50,level distance=0.5cm]
		child[level distance=.25cm,grow=215, densely dashed, ->] {}  
	}
        child[grow=150] {node {$|f\ast g_{\lambda_1^{(j)}}|$}
	child[level distance=.75cm,grow=215, densely dashed, ->] {node {$|f\ast g_{\lambda_1^{(j)}}|\ast\chi_{_2}$}
	}
	child[grow=83, level distance=0.5cm] 
	child[grow=97, level distance=0.5cm] 
        child[grow=110] {node {$||f\ast g_{\lambda_1^{(j)}}|\ast g_{\lambda_2^{(l)}}|$}
	child[level distance=.9cm,grow=215, densely dashed, ->] {node {$||f\ast g_{\lambda_1^{(j)}}|\ast g_{\lambda_2^{(l)}}|\ast\chi_{_3}$}
	}
        child[grow=130] {node {$|||f\ast g_{\lambda_1^{(j)}}|\ast g_{\lambda_2^{(l)}}|\ast g_{\lambda_3^{(m)}}|$}
	child[level distance=0.75cm,grow=215, densely dashed, ->] {node[align=left]{\\ $\cdots$}}}
        child[grow=90,level distance=0.5cm] 
 	child[grow=50,level distance=0.5cm]
       }
       child[grow=63, level distance=1.05cm] {[fill=gray!50!black] circle (0.5pt)
	child[grow=130,level distance=0.5cm] 
       child[grow=90,level distance=0.5cm] 
       child[grow=50,level distance=0.5cm] 
       child[level distance=.25cm,grow=325, densely dashed, ->] {}     
       }
       }
       child[grow=30] {node {$|f\ast g_{\lambda_1^{(p)}}|$}
       child[level distance=0.75cm, grow=325, densely dashed, ->] {node {$|f\ast g_{\lambda_1^{(p)}}|\ast\chi_{_2}$}
	}
	child[grow=83, level distance=0.5cm] 
	child[grow=97, level distance=0.5cm] 
        child[grow=117, level distance=1.05cm] {[fill=gray!50!black] circle (0.5pt)
        child[grow=130,level distance=0.5cm] 
        child[grow=90,level distance=0.5cm] 
        child[grow=50,level distance=0.5cm] 
        child[level distance=.25cm,grow=215, densely dashed, ->] {}  
	 }
        child[grow=70] {node {$||f\ast g_{\lambda_1^{(p)}}|\ast g_{\lambda_2^{(r)}}|$}
	 child[level distance=0.9cm,grow=325, densely dashed, ->] {node {$||f\ast g_{\lambda_1^{(p)}}|\ast g_{\lambda_2^{(r)}}|\ast\chi_{_3}$}}
	child[grow=130,level distance=0.5cm] 
         child[grow=90,level distance=0.5cm] 
             child[grow=50] {node {$|||f\ast g_{\lambda_1^{(p)}}|\ast g_{\lambda_2^{(r)}}|\ast g_{\lambda_3^{(s)}}|$}
             child[level distance=0.75cm,grow=325, densely dashed, ->] {node[align=left]{\\ $\cdots$}}}
	}
     }
	child[level distance=0.75cm, grow=215, densely dashed, ->] {node {$f\ast \chi_{_1}$}};
\end{tikzpicture}
\caption{Network architecture underlying the feature extractor  \eqref{ST}. The index $\lambda_n^{(k)}$ corresponds to the $k$-th filter $g_{\lambda_n^{(k)}}$ of the collection $\Psi_n$ associated with the $n$-th network layer. The function $\chi_{n+1}$ is the output-generating filter of the $n$-th network layer. The root of the network corresponds to $n=0$.} 
\label{fig:gsn}
\end{figure*}
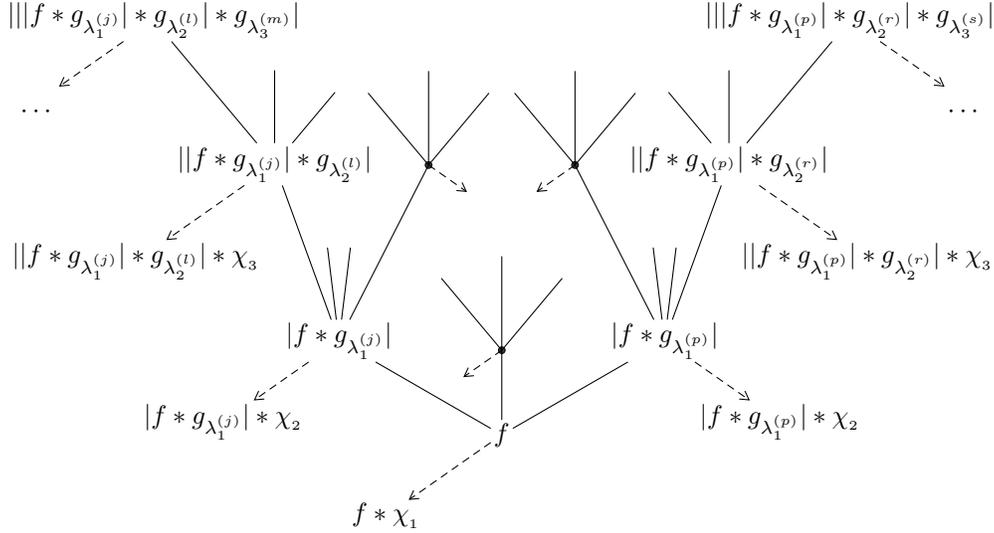

\section{dcnn-based feature extractors}\label{architecture}
Throughout the paper, we use the terminology of \cite{Wiatowski_journal}, consider (unless explicitly stated otherwise)  input signals $f\in L^2(\Rd)$, and employ the module-sequence
\begin{equation}\label{mods}
\Omega:=\big((\Psi_n,|\cdot|,\text{Id})\big)_{n\in \mathbb{N}},
\end{equation}
i.e., each network layer is associated with (i) a collection of filters  $\Psi_n:= \{\chi_n \}\cup\{ g_{\lambda_n}\}_{\lambda_n \in \Lambda_n}\subseteq L^1(\Rd) \cap L^2(\Rd)$, where $\chi_n$,  referred to as  output-generating filter, and the $g_{\lambda_n}$, indexed by a countable  set $\Lambda_n$,  satisfy the  frame condition \cite{Daubechies,Antoine,Kaiser}
\begin{equation}\label{PFP}
\vspace{0.1cm}
A_n\| f\|^2_2\leq \| f\ast\chi_n\|_2^2 +  \sum_{\lambda_n \in \Lambda_n}\| f\ast g_{\lambda_n}\|^2\leq B_n\| f\|_2^2, 
\end{equation} for all $f\in L^2(\Rd)$,  for some $A_n,B_n>0$, (ii) the modulus non-linearity $|\cdot|:L^2(\Rd)\to L^2(\Rd)$, $|f|(x):=|f(x)|$, and (iii) no pooling, which, in the terminology of \cite{Wiatowski_journal}, corresponds to pooling through the identity operator with pooling factor equal to one. Associated with the module $(\Psi_n,|\cdot|,\text{Id})$, the operator $U_n[\lambda_n]$ defined in \cite[Eq. 12]{Wiatowski_journal} particularizes to
\vspace{0.13cm}
 \begin{equation}\label{eq:1}
\vspace{0.2cm}
U_n[\lambda_n]f=\big|f\ast g_{\la_n}\hspace{-0.05cm} \big|.
\end{equation}   
We extend \eqref{eq:1} to paths on index sets $$q=(\lambda_1,\lambda_2,\dots, \lambda_n) \in \Lambda_1\times \Lambda_2\times\dots\times \Lambda_n=:\Lambda^n,\quad n \in \mathbb{N},$$ according to
\begin{align}
U[q]f&\,=\, U[(\lambda_1,\lambda_2,\text{\mydots}, \lambda_n)]f\nonumber\\
&:=\, U_n[\lambda_n] \cdots U_2[\lambda_2]U_1[\lambda_1]f\label{aaaab},
\end{align}
where, for the empty path $e:=\emptyset$, we set $\Lambda^0:=\{ e \}$ and $U[e]f:=f$. The signals $U[q]f$, $q\,\in\, \Lambda^n$, associated with the $n$-th network layer, are often referred to  as  feature maps in the deep learning literature. The feature vector $\Phi_\Omega(f)$ is obtained by aggregating filtered versions of the feature maps. More formally, $\Phi_\Omega(f)$ is defined as \cite[Definition 3]{Wiatowski_journal}
\begin{equation}\label{ST}
\Phi_\Omega (f):=\bigcup_{n=0}^\infty\Phi^n_\Omega(f),
\end{equation}
where \vspace{0.1cm}$$\Phi^n_\Omega(f):=\{ (U[q]f) \ast \chi_{n+1} \}_{q \in \Lambda^n}\vspace{0.1cm}$$ are the features generated in the $n$-th network layer, see Figure \ref{fig:gsn}. Here, $n = 0$ corresponds to the root of the network. The function $\chi_{n+1}$ is the output-generating filter of the $n$-th network layer. The feature extractor 
$$
\Phi_\Omega:L^2(\Rd) \to \big(L^2(\Rd)\big)^{\bigcup_{n=0}^\infty\Lambda^n}
$$ 
was shown in \cite[Theorem 1]{Wiatowski_journal} to be vertically translation-invariant, provided although that pooling is  employed, with pooling factors $S_n\geq1$, $n\in \mathbb{N}$, (see \cite[Eq. 6]{Wiatowski_journal} for the definition of the general pooling operator) such that $\lim\limits_{N\to \infty} \prod_{n=1}^N S_n=\infty$. Moreover, $\Phi_\Omega$ exhibits limited sensitivi\-ty to certain non-linear deformations on (input) signal classes such as band-limited functions \cite[Theorem 2]{Wiatowski_journal}, cartoon functions \cite[Theorem 1]{grohs_wiatowski}, and Lipschitz functions \cite[Corollary 1]{grohs_wiatowski}.

\section{Energy decay and trivial null-set}\label{Sec:prop}
The first central goal of this paper is to understand how fast the energy contained in the feature maps decays across layers. Specifically, we shall study the decay of
 \begin{equation}\label{ene}
 W_N(f):=\sum_{q\,\in\, \Lambda^N}\| U[q]f\|_2^2, \quad f\in L^2(\Rd),
 \end{equation}
as a function of network depth $N$.  Moreover, it is desirable that the infinite-depth feature vector $\Phi_\Omega(f)$ be informative in the sense of the only signal mapping to the all-zeros feature vector being the zero input signal, i.e., $\Phi_\Omega$ has a trivial null-set 
\begin{equation}\label{kdjkdjkdj}
\mathcal{N}(\Phi_\Omega):=\{ f \in L^2(\Rd) \ | \ \Phi_\Omega(f)=0  \} \stackrel{!}{=} \{ 0\}.
\end{equation}
Figure \ref{fig:binary} illustrates the practical ramifications of a non-trivial null-set in a binary classification task. $\mathcal{N}(\Phi_\Omega)=\{ 0\}$ can be  guaranteed by asking for ``energy conservation'' in the sense of 
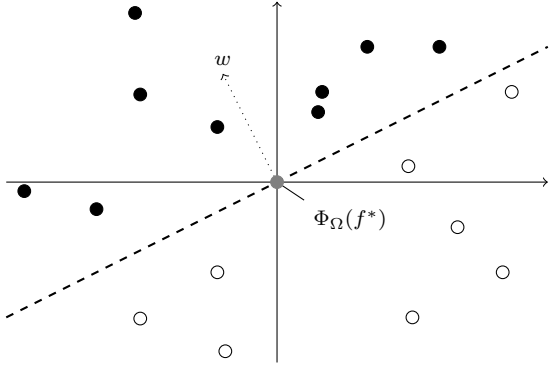
\begin{figure}[t]
\begin{center}\vspace{0.1cm}
    \begin{tikzpicture}[scale=1.2] 
	\tikzstyle{inti}=[draw=none,fill=none]
	\draw[->] (-3,0) -- (3,0);
	\draw[->] (0,-2) -- (0,2);

\draw[dashed,thick] (-3,-1.5) -- (3,1.5);			

\draw[dotted,->] (0,0) -- (-.6,1.2)  node[anchor=south] {\footnotesize$w$};	
\filldraw [black]    (-1.5154 ,   0.9713) circle (2pt)
			 (1.7993   ,1.5) circle (2pt)
			
			 (-1.5728  ,  1.8754) circle (2pt)
			 
			 (.5,1) circle (2pt)
			 (-2,-0.3) circle (2pt)
			 (1,1.5) circle (2pt)
			 (-2.8,-0.1) circle (2pt)
			(-0.6611   , 0.6107) circle (2pt)
			(0.4554   , 0.7763) circle (2pt)
			 ;	
\draw    
(-1.5154 ,   -1.513) circle (2pt)
			 (2 ,-.5) circle (2pt)
			
			 (-.5728  ,  -1.8754) circle (2pt)
			 
			 (2.5,-1) circle (2pt)
			  (2.6,1.) circle (2pt)
			 (1.5,-1.5) circle (2pt)
	
			(-0.6611   , -1) circle (2pt)
			(1.4554   , 0.1763) circle (2pt);
	
\draw (0cm, 0) -- (.3cm,-.2cm) node[anchor=north west] {\footnotesize$\Phi_\Omega(f^\ast)$};
\filldraw [gray]    
(0 ,0) circle (2pt);
                		\end{tikzpicture}
\end{center}
\vspace{0.3cm}
\caption{Impact of a non-trivial null-set $\mathcal{N}(\Phi_\Omega)$ in a binary classification task. The feature vector $\Phi_\Omega(f)$ is fed into a linear classifier \cite{Bishop}, which determines  set membership based on the sign of the inner product $\langle w,\Phi_\Omega(f)\rangle$. The (learned) weight vector $w$ is perpendicular to the separating hyperplane (dashed line). If the null-set of the feature extractor $\Phi_\Omega$ is non-trivial, there exist input signals $f^\ast \neq 0$ that are mapped to the origin in feature space, i.e., $\Phi_\Omega(f^\ast)=0$ (gray circle), and therefore lie---independently of the weight vector $w$---on the separating hyperplane. These input signals $f^\ast \neq 0$ are therefore unclassifiable.}
		\label{fig:binary}
\end{figure}
 \begin{equation}\label{eq:lbeee}
A_{\Omega} \| f\|^2_2\leq  |||\Phi_\Omega(f)|||^2\leq B_\Omega \| f\|_2^2, \quad \forall f\in L^2(\Rd),
 \end{equation}
 for some constants $A_\Omega,B_\Omega>0$ (possibly depending on the module-sequence $\Omega$) and with the feature space norm $||| \Phi_\Omega(f)|||:=\big(\sum_{n=0}^{\infty}|||\Phi^n_\Omega(f)|||^2\big)^{1/2}$, where $|||\Phi^n_\Omega(f)|||:= \big(\sum_{q\,\in\, \Lambda^n}\| (U[q]f)\ast \chi_{n+1}\|_2^2\big)^{1/2}$. Indeed,  \eqref{kdjkdjkdj} follows from \eqref{eq:lbeee} as the upper bound in \eqref{eq:lbeee} yields $\{ 0\}\subseteq \mathcal{N}(\Phi_\Omega)$, and   the lower bound  implies $\{ 0\} \supseteq \mathcal{N}(\Phi_\Omega)$. We emphasize that, as $\Phi_\Omega$ is a non-linear operator (owing to the modulus non-linearities), characterizing its null-set is non-trivial in general.  The  upper bound in \eqref{eq:lbeee} was established in  \cite[Appendix E]{Wiatowski_journal}. While the existence of this upper bound  is implied by the filters $\Psi_n$,  $n\in \mathbb{N}$, satisfying the frame property \eqref{PFP} \cite[Appendix E]{Wiatowski_journal}, perhaps surprisingly, this is not enough to guarantee $A_\Omega>0$ (see Appendix \ref{counterexmp} for an example). We refer the reader to Section \ref{sec:nlay} for results on the null-set of the \emph{finite-depth} feature extractor $\bigcup_{n=0}^N \Phi_\Omega^n$.

\begin{figure*}[t]
\begin{center}
    \begin{tikzpicture}[scale=1.3] 
	\tikzstyle{inti}=[draw=none,fill=none]

	\draw[->] (-4,0) -- (4,0) node[right] {\footnotesize $\omega$};
	\draw[->] (0,0) -- (0,1.7);	
	
	\draw[thick,rounded corners=4pt]	
	(-4,1.2)--(-3,1.4)--(-2,1.1)--(-1.,1.3)--(-.5,1.2)--(0,1.4)--(.5,1.)--(1.,1.3)--(2,1.1)--(3,.6)--(4,1.4);
	
	\draw[thick,dotted, domain=0:4,samples=100] plot (\x,{1-((exp(-(\x^(2))/.4)))});
	\draw[thick,dotted, domain=-4:0,samples=100] plot (\x,{1-((exp(-(-\x^(2))/.4)))});
	
	\draw (-2pt,1cm) -- (2pt,1 cm) node[anchor=east] {\small$1 \ $};
	\draw (1 cm,2pt) -- (1 cm,-2pt) node[anchor=north] {\small$ a^{N-1} $};
	\draw (-1 cm,2pt) -- (-1 cm,-2pt) node[anchor=north] {\small$- a^{N-1}$};
	\draw (3 cm,2pt) -- (3 cm,-2pt) node[anchor=north] {\small$ a^N$};
	\draw (-3 cm,2pt) -- (-3 cm,-2pt) node[anchor=north] {\small$- a^N$};		
	
	\draw  (-3 cm,.95cm)--(-2.7 cm,.7cm)  node[anchor=north] {\small$ \ {h_{N}}(\omega)$};
	\draw  (-2.2 cm,1.2cm)--(-1.9 cm,1.4cm)  node[anchor=south west] {\small$|\widehat{f}(\omega)|^2$};
\end{tikzpicture}
\begin{tikzpicture}[scale=1.3] 
	\draw  (-3.2 cm,1.3cm)--(-3 cm,1cm)  node[anchor=north] {\small$ \ |\widehat{f}(\omega)|^2\cdot {h_{N}}(\omega)$};
	\tikzstyle{inti}=[draw=none,fill=none]
	\draw[->] (-4,0) -- (4,0) node[right] {\footnotesize $\omega$};
	\draw[->] (0,0) -- (0,1.7);		
	 \draw[thick,dashed, rounded corners=8pt]	
	(-4,1.2)--(-3,1.4)--(-2,1.1)--(-1.,1.3)--(-.5,.4)--(0,-0.09)--(.5,.4)--(1.,1.3)--(2,1.1)--(3,.6)--(4,1.4);

	\draw (-2pt,1cm) -- (2pt,1 cm) node[anchor=east] {\small$1 \ $};
	\draw (1 cm,2pt) -- (1 cm,-2pt) node[anchor=north] {\small$ a^{N-1} $};
	\draw (-1 cm,2pt) -- (-1 cm,-2pt) node[anchor=north] {\small$- a^{N-1}$};
	\draw (3 cm,2pt) -- (3 cm,-2pt) node[anchor=north] {\small$ a^N$};
	\draw (-3 cm,2pt) -- (-3 cm,-2pt) node[anchor=north] {\small$- a^N$};		
\end{tikzpicture}
\begin{tikzpicture}[scale=1.3] 
	\tikzstyle{inti}=[draw=none,fill=none]

	\draw[->] (-4,0) -- (4,0) node[right] {\footnotesize $\omega$};
	\draw[->] (0,0) -- (0,1.7);	
	
	\draw[thick,rounded corners=4pt]	
	(-4,1.2)--(-3,1.4)--(-2,1.1)--(-1.,1.3)--(-.5,1.2)--(0,1.4)--(.5,1.)--(1.,1.3)--(2,1.1)--(3,.6)--(4,1.4);
	\draw  (-2.2 cm,1.2cm)--(-1.9 cm,1.4cm)  node[anchor=south west] {\small$|\widehat{f}(\omega)|^2$};
	\draw  (-3 cm,.9cm)--(-2.7 cm,.7cm)  node[anchor=north] {\small$ \ {h_{N+1}(\omega)}$};
	\draw[thick,dotted, domain=0:4,samples=100] plot (\x,{1-((exp(-(\x^(2))/3)))});
	\draw[thick,dotted, domain=-4:0,samples=100] plot (\x,{1-((exp(-(-\x^(2))/3)))});
	
	\draw (-2pt,1cm) -- (2pt,1 cm) node[anchor=east] {\small$1 \ $};
	\draw (1 cm,2pt) -- (1 cm,-2pt) node[anchor=north] {\small$ a^{N-1} $};
	\draw (-1 cm,2pt) -- (-1 cm,-2pt) node[anchor=north] {\small$- a^{N-1}$};
	\draw (3 cm,2pt) -- (3 cm,-2pt) node[anchor=north] {\small$ a^N$};
	\draw (-3 cm,2pt) -- (-3 cm,-2pt) node[anchor=north] {\small$- a^N$};		
	
\end{tikzpicture}
\begin{tikzpicture}[scale=1.3] 
	\draw (-2.1 cm,.65cm)--(-1.8 cm,.95cm)  node[anchor=south] {\small$\ \  \ |\widehat{f}(\omega)|^2\cdot {h_{N+1}}(\omega)$};
	\tikzstyle{inti}=[draw=none,fill=none]
	\draw[->] (-4,0) -- (4,0) node[right] {\footnotesize $\omega$};
	\draw[->] (0,0) -- (0,1.7);		
	 \draw[thick,dashed,rounded corners=6pt]	
	(-4,1.2)--(-3,1.4)--(-2.4,.7)--(-1.4,.3)--(-1.,.25)--(-.5,.1)--(0,-.01)--(.5,.05)--(1.,.2)--(2,.3)--(3,.6)--(4,1.4);

	\draw (-2pt,1cm) -- (2pt,1 cm) node[anchor=east] {\small$1 \ $};
	\draw (1 cm,2pt) -- (1 cm,-2pt) node[anchor=north] {\small$ a^{N-1} $};
	\draw (-1 cm,2pt) -- (-1 cm,-2pt) node[anchor=north] {\small$- a^{N-1}$};
	\draw (3 cm,2pt) -- (3 cm,-2pt) node[anchor=north] {\small$ a^N$};
	\draw (-3 cm,2pt) -- (-3 cm,-2pt) node[anchor=north] {\small$- a^N$};	;	
\end{tikzpicture}
\end{center}
\caption{Illustration of the impact of network depth $N$ on the upper bound on $W_N(f)$ in \eqref{erghekj}, for $\eps=1$ and $a>1$. The function ${h_{N}}(\omega):=\big(1-\widehat{r_g}\big(\frac{\omega}{\eps a^{N-1}}\big)\big)$, where $\widehat{r_g}(\omega)=e^{-\omega^2}$, is of increasing high-pass nature as $N$ increases, which makes  the upper bound in \eqref{erghekj}  decay in $N$. }
		\label{fig:low}
\end{figure*}
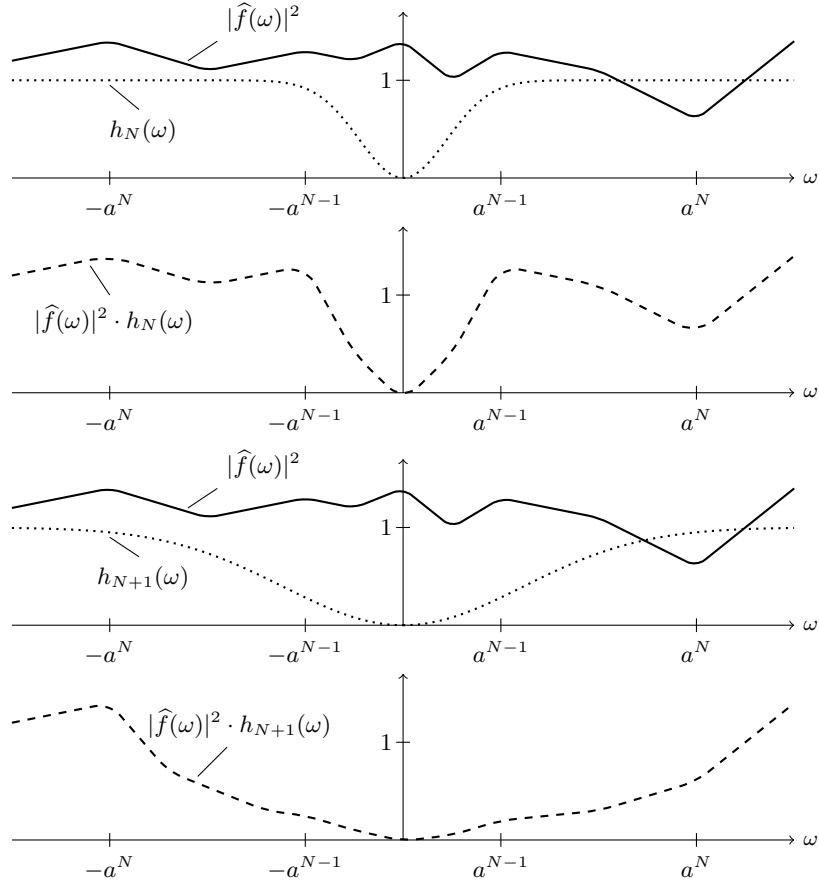

Previous work on the decay rate of $W_N(f)$ in \cite[Section 5]{Waldspurger} shows that for wavelet-based networks (i.e., in every network layer the filters $\Psi=\{\chi\}\cup\{g_{\lambda}\}_{\lambda\in \Lambda}$ in \eqref{mods} are taken to be (specific) $1$-D wavelets that constitute a Parseval frame, with  $\chi$ a low-pass filter) there exist $\varepsilon >0$ and $a>1$ (both constants unspecified)  such that 
\begin{equation}\label{erghekj}
W_N(f)\leq \int_{\mathbb{R}}|\widehat{f}(\omega)|^2\Big(1-\Big| \reallywidehat{r_g}\Big(\frac{\omega}{\varepsilon a^{N-1}}\Big)\Big|^2\Big)\mathrm d\omega,
\end{equation}
for real-valued $1$-D signals $f\in L^2(\R)$ and  $N\geq 2$, where $\reallywidehat{r_g}(\omega):=e^{-\omega^2}$\hspace{-0.1cm}. To see that this result indicates energy decay, Figure \ref{fig:low} illustrates the influence of  network depth $N$ on the upper bound in \eqref{erghekj}. Specifically, we can see that increasing the network depth results in cutting out increasing amounts of energy of $f$ and thereby making the upper bound in \eqref{erghekj} decay as a function of $N$. Moreover, it is interesting to note that the upper bound on $W_N(f)=\sum_{q\,\in\,\Lambda^N}\| U[q]f\|_2^2$  is independent of the wavelets generating the feature maps $U[q]f$, $q\in \Lambda^N$.  For scattering networks that  employ, in every network layer, uniform covering filters $\Psi=\{\chi\}\cup\{g_{\lambda}\}_{\lambda\in \Lambda}\subseteq L^1(\Rd)\cap L^2(\Rd)$ forming a Parseval frame (where $\chi$, again, is a low-pass filter), exponential energy decay according to \begin{equation}\label{czaja}
W_N(f)= \mathcal{O}(a^{-N}), \quad \forall f\in L^2(\Rd),
\end{equation}
for an unspecified $a>1$, was established in \cite[Proposition 3.3]{czaja2016uniform}. Moreover, \cite[Section 5]{Waldspurger} and \cite[Theorem 3.6 (a)]{czaja2016uniform} state---for the respective module-sequences---that \eqref{eq:lbeee} holds with $A_\Omega=B_{\Omega}=1$ and hence   
\begin{equation}\label{eco}
|||\Phi_\Omega(f)|||^2= \| f \|_2^2.
\end{equation}
The first main goal of the present paper is to establish i) for $d$-dimensional complex-valued input signals that $W_N(f)$ decays polynomially according to 
\begin{equation}\label{main_h:up1_man}
W_N(f)\leq B_\Omega^N \int_{\Rd}\big|\widehat{f}(\omega)\big|^2\Big(1-\Big|\widehat{r_l}\Big(\frac{\omega}{N^{\alpha}}\Big)\Big|^2\Big)\mathrm d\omega,
\end{equation}
for   $f\in L^2(\Rd)$ and  $N\geq 1$, where $\alpha=1$, for $d=1$, and $\alpha=\log_2(\sqrt{d/(d-1/2)}),$ for $d\geq 2$, $B_\Omega^N=\prod_{k=1}^N \max\{1,B_k\}$, and $\widehat{r_l}:\Rd\to \R$, $\widehat{r_l}(\omega)=(1-|\omega|)^{l}_{+}$, with $l>\lfloor d/2 \rfloor+1$, for networks based on general filters $\{\chi_n\}\cup\{ g_{\lambda_n}\}_{\lambda_n\in \Lambda_n}$ that satisfy mild analyticity and high-pass conditions and  are allowed to be different in different network layers (with the proviso that $\chi_n$,  $n\in \mathbb{N}$, is of low-pass nature in a sense to be made precise),  and ii) for $1$-D complex-valued  input signals that \eqref{ene} decays exponentially according to  
\begin{equation}\label{main_h:wave_man}
W_N(f)\leq \int_{\R}\big|\widehat{f}(\omega)\big|^2\Big(1-\Big|\widehat{r_l}\Big(\frac{\omega}{a^{N-1}}\Big)\Big|^2\Big)\mathrm d\omega,
\end{equation}
for   $f\in L^2(\R)$ and $N\geq 1$, for networks that are based, in every network layer, on a broad family of wavelets, with the decay factor given explicitly as $a=\frac{5}{3}$, or on a broad family of Weyl-Heisenberg filters \cite[Appendix B]{Wiatowski_journal}, with decay factor $a=\frac{3}{2}$. Thanks to the right-hand side (RHS) of \eqref{main_h:up1_man} and \eqref{main_h:wave_man} not depending on the specific filters $\{ \chi_n\} \cup \{ g_{\lambda_n}\}_{\lambda_n\in \Lambda_n}$, we will be able to establish---under smoothness assumptions on the input signal $f$---universal energy decay results.  Specifically,  particularizing the RHS expressions in \eqref{main_h:up1_man} and \eqref{main_h:wave_man} to Sobolev-class input signals $f\in H^s(\Rd)$, $s> 0$, where 
\begin{equation*}\label{eq:soooo}
H^s(\Rd)=\Big\{f\in L^2(\Rd) \ \Big| \ \| f\|_{H^s}<\infty \Big\},
\end{equation*}
with $\| f\|_{H^s}:= \big(\int_{\Rd}(1+|\omega|^2)^{s}|\widehat{f}(\omega)|^2\mathrm d\omega\big)^{1/2}$,
we show that \eqref{main_h:up1_man} yields  polynomial energy decay according to 
\begin{equation}\label{som3}
W_N(f)= \mathcal{O}\big(N^{-\gamma \alpha}\big),
\end{equation}
and \eqref{main_h:wave_man}  exponential energy decay  
\begin{equation}\label{som2}
W_N(f) = \mathcal{O}\big(a^{-\gamma N}\big),
\end{equation}
where $\gamma:=\min\{1,2s\}$ in both cases.
Sobolev spaces $H^s(\Rd)$ contain a wide range of practically relevant signal classes such as, e.g., 
\begin{itemize}
\item[--]{the space $L^2_L(\Rd):=\{ f\in L^2(\Rd) \ | \ \text{supp}(\widehat{f}\,)\subseteq B_L(0)\}$, $L\geq0$, of $L$-band-limited functions according to $L^2_L(\Rd)\subseteq H^s(\Rd)$, for  $L\geq0$ and   $s>0$. This follows from 
\begin{align*}\label{sdjhekjgfh}
\int_{\Rd}(1+|\omega|^2)^s|\widehat{f}(\omega)|^2\mathrm d\omega&=\int_{B_L(0)}(1+|\omega|^2)^s|\widehat{f}(\omega)|^2\mathrm d\omega\\
&\leq (1+|L|^2)^s\| f\|_2^2<\infty,
\end{align*}
for  $f\in L^2_L(\Rd)$, $L\geq0$,  and  $s>0$, where we used Parseval's formula and the fact that $\omega \mapsto (1+|\omega|^2)^s$,  $\omega \in \Rd$,  is  monotonically increasing in $|\omega|$, for  $s>0$,
 }
\item[--]{the space $\mathcal{C}^{K}_{\mathrm{CART}}$ of cartoon functions of size $K,$ introduced in \cite{Cartoon}, and widely used in the mathematical signal processing literature \cite{grohs_wiatowski,wiatowski2016discrete,Grohs_alpha,grohs2014parabolic,ShearletsIntro} as a model for natural images such as, e.g., images of handwritten digits \cite{MNIST} (see Figure \ref{fig:data}). For a formal definition of $\mathcal{C}^{K}_{\mathrm{CART}}$, we refer the reader to Appendix \ref{sec:proofcart}, where we also show that $\mathcal{C}^{K}_{\mathrm{CART}}\subseteq H^s(\Rd)$, for  $K>0$ and  $s\in (0,1/2)$. }
\end{itemize}

Moreover, Sobolev functions are contained in the space of $k$-times continuously differentiable functions $C^k(\Rd,\mathbb{C})$ according to $H^s(\Rd)\subseteq C^k(\Rd,\mathbb{C})$, for  $s>k+\frac{d}{2}$ \cite[Section 4]{adams1975sobolev}. 

Our second central goal is to prove energy conservation according to \eqref{eq:lbeee} (which, as explained above, implies $\mathcal{N}(\Phi_\Omega)=\{ 0\}$) for the network confi\-gurations corresponding to the energy decay results \eqref{main_h:up1_man} and \eqref{main_h:wave_man}. Finally, we provide handy estimates of the number of layers needed to have at least $((1-\varepsilon)\cdot 100)\%$ of the input signal energy be contained in the feature vector.

\section{Main results}\label{sec:mainresults}
Throughout the paper, we make the following  assumptions on the filters $\{ g_{\lambda_n}\}_{\lambda_n \in \Lambda_n}$.
\begin{assumption}\label{ass}
The $\{ g_{\lambda_n}\}_{\lambda_n \in \Lambda_n}$, $n \in \mathbb{N}$, are analytic in the following sense:  For every layer index $n\in \mathbb{N}$, for every $\lambda_n\in \Lambda_n$, there exists an orthant $H_{A_{\lambda_n}}\subseteq \Rd$, with  $A_{\lambda_n}\in O(d)$, such that  
\begin{equation}\label{eq:ass1}
\text{supp}(\widehat{g_{\lambda_n}})\subseteq H_{A_{\lambda_n}}. 
\end{equation}
Moreover, there exists $\delta>0$  so that 
\begin{equation}\label{eq:ass2}
\sum_{\lambda_n\in \Lambda_n}|\widehat{g_{\lambda_n}}(\omega)|^2=0, \quad \text{a.e. } \omega \in B_{\delta}(0).
\end{equation}
\end{assumption}
In the $1$-D case, i.e., for $d=1$, Assumption \ref{ass} simply amounts to every filter $g_{\lambda_n}$ satisfying either $$\text{supp}(\widehat{g_{\lambda_n}})\subseteq (-\infty,-\delta] \quad \text{or} \quad \text{supp}(\widehat{g_{\lambda_n}})\subseteq [\delta,\infty),$$ which constitutes an ``analyticity'' and ``high-pass''  condition. For dimensions $d\geq2$, Assumption \ref{ass}  requires that every filter $g_{\lambda_n}$ be of high-pass nature and have a Fourier transform supported in a (not necessarily canonical) orthant. Since the frame condition \eqref{PFP} is equivalent to the Littlewood-Paley condition \cite{Paley} 
\begin{equation}\label{eq:tight}
A_n\leq |\widehat{\chi_n}(\omega)|^2+\sum_{\lambda_n \in \Lambda_n}|\widehat{g_{\lambda_n}}(\omega)|^2\leq B_n,\hspace{0.25cm} \text{a.e. } \omega \in \Rd,
\end{equation} 
\eqref{eq:ass2} implies low-pass characteristics for $\chi_n$ to fill the spectral gap $B_{\delta}(0)$ left by the filters $\{ g_{\lambda_n}\}_{\lambda_n \in \Lambda_n}$.

\begin{figure}
\vspace{-0.5cm}
\centering
	\includegraphics[width = .22\textwidth]{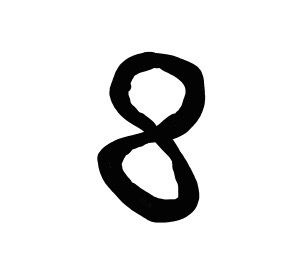}
	\caption{An image of a handwritten digit is modeled by a $2$-D cartoon function. }
	\label{fig:data}
\end{figure}
The conditions \eqref{eq:ass1} and \eqref{eq:ass2} we impose on the $\Psi_n$, $n\in \mathbb{N}$, are not overly restrictive as they encompass, inter alia, various constructions of  Weyl-Heisenberg filters (e.g., a $1$-D $B$-spline as prototype function  \cite[Section 1]{Gaborexmpl}), wavelets (e.g., analytic Meyer wavelets \cite[Section 3.3.5]{Daubechies}  in $1$-D, and Cauchy wavelets \cite{Vander} in $2$-D), and specific constructions of ridgelets \cite[Section 2.2]{Grohs_transport}, curvelets \cite[Section 4.1]{CandesDonoho2}, $\alpha$-curvelets \cite[Section 3]{Grohs_alpha}, and shearlets (e.g., cone-adapted shearlets \cite[Section 4.3]{ShearletsIntro}). We refer the reader to \cite[Appendices B and C]{Wiatowski_journal} for a brief review of some of these filter structures. 

We are now ready to state our  main result on energy decay and energy conservation.
\begin{figure*}[t]
\begin{center}

      \begin{tikzpicture}[scale=1.3] 

	\filldraw[black!10](-1,0) rectangle (1,1);
	\filldraw[black!10](-3.8,0) rectangle (-2.4,1);
	\draw (-2pt,1cm) -- (2pt,1 cm) node[anchor=south east] {\small$1 \ $};
	\draw[->] (-4,0) -- (4,0) node[right] {\footnotesize $\omega$};
	\draw[->] (0,0) -- (0,1.7);	
	
	\draw (1 cm,2pt) -- (1 cm,-2pt) node[anchor=north] {\small$\delta $};
	\draw (-1 cm,2pt) -- (-1 cm,-2pt) node[anchor=north] {\small$-\delta$};

	 \draw[thick,rounded corners=6pt]
		(-4,.1) -- (-3.8,.1)--(-3.4,0.05)--(-2.4,.8) -- (-2,0.8)--(-1.5,0.) -- (-1.,-0.7)--
		(-.7,-.8) --( -0.5,-0.25) -- (0,.5)--(0.5,0.2) --
		(1,0) --(1.5,0.7) --(1.7,1) -- (2.2,0.4) -- (4,-0.2);

	\draw[dashed] (-2.4,0) rectangle (-1,1);
	\draw[dashed] (-3.8,0) rectangle (-2.4,1);

	\node at (-4.2,0.5) {$\cdots$};
	
	\draw[dashed] (2.4,0) rectangle (1,1);
	\draw[dashed] (3.8,0) rectangle (2.4,1);

	\node at (4.3,0.5) {$\cdots$};

	\draw[dashed] (1,0.) rectangle (-1,1);
	
	\draw  (-3. cm,1.05cm)--(-3.3 cm,1.2cm)  node[anchor=south] {\small$ \ \widehat{g_{\lambda_n}}(\omega)$};
	
	\draw  (.5 cm,1.05cm)--(0.7 cm,1.2cm)  node[anchor=south] {\small$ \ \ \widehat{\chi_{n+1}}(\omega)$};
	\draw  (3cm,.175cm)--(3.2cm,.375cm)  node[anchor=south] {\small$ \ \ \widehat{f}(\omega)$};
	
	\end{tikzpicture}

      \begin{tikzpicture}[scale=1.3] 

	\draw (-2pt,1cm) -- (2pt,1 cm) node[anchor=south east] {\small$1 \ $};
	\draw[->] (-4,0) -- (4,0) node[right] {\footnotesize $\omega$};
	\draw[->] (0,0) -- (0,1.7);	
	\draw (1 cm,2pt) -- (1 cm,-2pt) node[anchor=north] {\small$\delta $};
	\draw (-1 cm,2pt) -- (-1 cm,-2pt) node[anchor=north] {\small$-\delta$};

	\draw  (-2.95 cm,.55cm)--(-3.25 cm,.7cm)  node[anchor=south] {\small$ \ \widehat{f}(\omega) \cdot \widehat{g_{\lambda_n}}(\omega)$};
		
	 \draw[thick,rounded corners=4pt]	
	(-3.8,.1)--(-3.4,0.05)--(-2.6,.7)--(-2.4,.8);

			 \draw[thick]	
	(-3.8,.11)--(-3.8,0);
				 \draw[thick]	
	(-4,0)--(-3.8,0);
				 \draw[thick]	
	(-2.4,0)--(4,0);

	 \draw[thick]	
	(-2.4,.815)--(-2.4,0);

	\end{tikzpicture}

\vspace{0.3cm}
    \begin{tikzpicture}[scale=1.3] 
	\tikzstyle{inti}=[draw=none,fill=none]

	\filldraw[black!10](-1,0) rectangle (1,1);
	\draw[dashed] (1,0.) rectangle (-1,1);
	\draw (-2pt,1cm) -- (2pt,1 cm) node[anchor=south east] {\small$1 \ $};
	\draw[->] (-4,0) -- (4,0) node[right] {\footnotesize $\omega$};
	\draw[->] (0,0) -- (0,1.7);	
	\draw (1 cm,2pt) -- (1 cm,-2pt) node[anchor=north] {\small$\delta $};
	\draw (-1 cm,2pt) -- (-1 cm,-2pt) node[anchor=north] {\small$-\delta$};

	\draw  (-1.6 cm,.2cm)--(-1.9 cm,.35cm)  node[anchor=south] {\small$ \reallywidehat{|f\ast g_{\lambda_n}|}(\omega)$};

		 \draw[thick,rounded corners=4pt]	
	(-4,.0)--(-2,.0)--(-1.4,0.2)--(-1.,0.3)--(-.4,0.6)--(0,.7)--(.4,0.6)--(1.,0.3)--(1.4,0.2)--(2,.0)--(4,.0);
	
	\draw  (.5 cm,1.05cm)--(0.7 cm,1.2cm)  node[anchor=south] {\small$ \ \ \widehat{\chi_{n+1}}(\omega)$};

	\end{tikzpicture}
	
\end{center}
\caption{Illustration of the demodulation effect of the modulus non-linearity.  The $\{ g_{\lambda_n}\}_{\lambda_n\in \Lambda_n}$ are taken as perfect band-pass filters (e.g., band-limited analytic Weyl-Heisenberg filters) and hence trivially satisfy the conditions in Assumption \ref{ass}. The modulus operation in combination with the analyticity and  the high-pass nature of the filters $\{ g_{\lambda_n}\}_{\lambda_n \in \Lambda_n}$ ensures that---in every network layer---the spectral content of each individual feature map  is moved to base-band (i.e., to low frequencies), where it is extracted by the (low-pass) output-generating filter $\chi_{n+1}$.}
		\label{fig:low3}
\end{figure*}
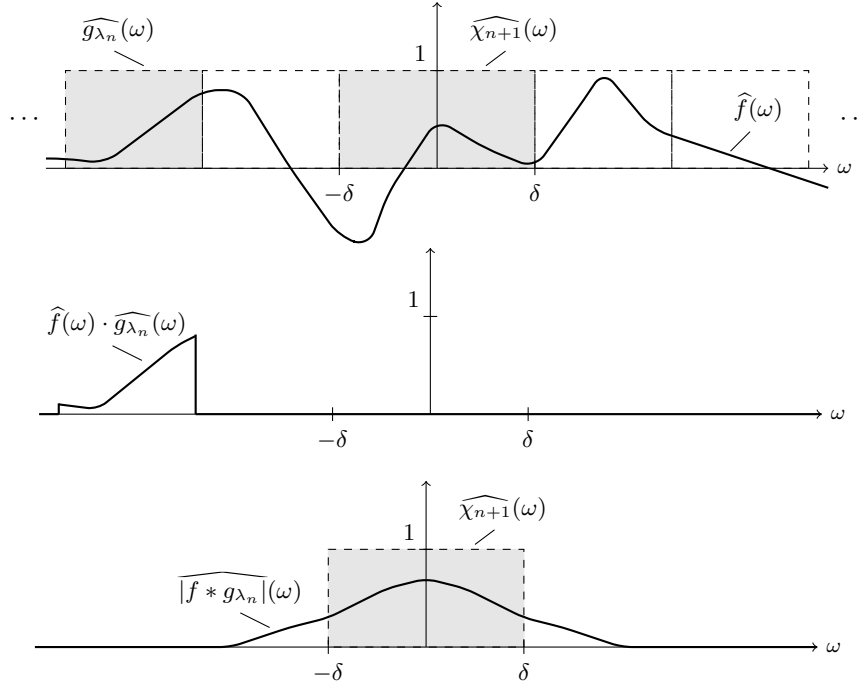

\begin{theorem}\label{thm3}
Let $\Omega$ be the module-sequence \eqref{mods} with filters $\{ g_{\lambda_n}\}_{\lambda_n \in \Lambda_n}$ satisfying the conditions  in Assumption \ref{ass}, and let $\delta>0$ be the radius of the spectral gap $B_{\delta}(0)$ left by the filters $\{ g_{\lambda_n}\}_{\lambda_n \in \Lambda_n}$ according to \eqref{eq:ass2}. Furthermore, let  $s\geq0$,\\[.5ex] $A_\Omega^N:=\prod_{k=1}^N \min\{1,A_k\}$,  $B_\Omega^N:=\prod_{k=1}^N \max\{1,B_k\}$, and 
\begin{equation}\label{defalp}
\alpha:=
\begin{cases}
1,& d=1,\\
\log_2(\sqrt{d/(d-1/2)}), &d\geq2.
\end{cases}
\end{equation}
 \begin{itemize}
 \item[\textit{i)}]{We have 
\begin{equation}\label{main_h:up1}
W_N(f)\leq B_\Omega^N \int_{\Rd}\big|\widehat{f}(\omega)\big|^2\Big(1-\Big|\widehat{r_l}\Big(\frac{\omega}{N^{\alpha}\delta}\Big)\Big|^2\Big)\mathrm d\omega,
\end{equation}
for  $f\in L^2(\Rd)$ and  $N\geq 1$, where $\widehat{r_l}:\Rd\to \R$, $\widehat{r_l}(\omega):=(1-|\omega|)^{l}_{+}$, with $l>\lfloor d/2 \rfloor+1$.}
\item[\textit{ii)}]{For every Sobolev function $f\in H^s(\Rd)$, $s>0$, we have 
\begin{equation}\label{main_so0}
W_N(f)= \mathcal{O}\big(B_{\Omega} ^N N^{-\gamma \alpha}\big),
\end{equation}
where $\gamma:= \min\{ 1,2s\}$.
}
\item[\textit{iii)}]{If, in addition to Assumption \ref{ass}, 
\begin{equation}\label{conditfb}0< A_\Omega:=\lim\limits_{N\to \infty} A_\Omega^N\leq B_\Omega:=\lim\limits_{N\to \infty} B_\Omega^N <\infty,
\end{equation} then we have energy conservation according to 
\begin{equation}\label{main}
A_\Omega\| f\|_2^2\leq |||\Phi_\Omega(f)|||^2\leq B_\Omega \|f\|_2^2,
\end{equation}
for all $f\in L^2(\Rd)$.
 }
 \end{itemize}
\end{theorem}
\vspace{-.5cm}
\begin{proof}[Proof]
For the proofs of i) and ii), we refer to Appendices \ref{app:prop2}  and \ref{Sob1Ap}, respectively. The proof of statement iii) is based on two key ingredients. First, we establish---in Proposition \ref{prop1} in Appendix \ref{app:prop1}---that the feature extractor $\Phi_\Omega$ satisfies the energy decomposition identity
\begin{equation}\label{eq:decomp}
A_\Omega^N\|f\|_2^2 \leq\sum_{n=0}^{N-1}||| \Phi^n_\Omega(f)|||^2+W_N(f)\leq B_\Omega^N\|f\|_2^2,
\end{equation}
for all $f\in L^2(\Rd)$ and all $N\geq 1$. Second, we show---in Proposition \ref{prop2} in Appendix \ref{bum}---that the integral on the RHS of \eqref{main_h:up1} goes to zero as $N\to \infty$ which, thanks to $\lim\limits_{N\to \infty} B_\Omega^N=B_\Omega<\infty$,  implies  that $W_N(f)\to 0$ as $N\to \infty$. We note that while the decomposition \eqref{eq:decomp} holds for general filters $\Psi_n$ satisfying the  frame property \eqref{PFP}, it is the upper bound \eqref{main_h:up1} that makes use of the analyticity and high-pass conditions in Assumption \ref{ass}. The final energy conservation result \eqref{main} is obtained by letting $N\to \infty$ in \eqref{eq:decomp}. 
\end{proof}
The strength of the results in Theorem \ref{thm3}  derives itself from the fact that the only condition we need to impose on the filters $\Psi_n$ is Assumption \ref{ass}, which, as already mentioned, is met by a wide array of filters. Moreover, condition \eqref{conditfb} is easily satisfied by normalizing the filters $\Psi_n$,  $n\in \mathbb{N}$, appropriately (see, e.g., \cite[Proposition 3]{Wiatowski_journal}). We note that this normalization, when applied to filters that satisfy Assumption \ref{ass}, yields filters that still meet Assumption \ref{ass}.

The identity \eqref{main_so0} establishes, upon normalization \cite[Proposition 3]{Wiatowski_journal} of the $\Psi_n$ to get $B_n\leq 1$,  $n\in \mathbb{N}$, that the energy decay rate, i.e., the decay rate of $W_N(f)$, is at least polynomial in $N$. We hasten to add that \eqref{main_h:up1} does not preclude the energy from decaying faster in practice. 

Underlying the energy conservation result \eqref{main} is the following demodulation effect induced by the modulus non-linearity in combination with the analyticity and high-pass properties of the filters $\{ g_{\lambda_n}\}_{\lambda_n \in \Lambda_n}$. In every network layer, the spectral content of each individual feature map  is moved to base-band (i.e., to low frequencies), where it is extracted by the low-pass output-generating atom $\chi_{n+1}$, see Figure \ref{fig:low3}. The components not collected by  $\chi_{n+1}$ (see Figure \ref{fig:low3}, bottom row) are  captured by the analytic high-pass filters $\{ g_{\lambda_{n+1}}\}_{\lambda_{n+1} \in \Lambda_{n+1}}$ in the next layer and, thanks to the modulus non-linearity, again moved to low frequencies and extracted by $\chi_{n+2}$. Iterating this process ensures that the null-set of the feature vector (be it for the infinite-depth network or, as established in  Section \ref{sec:nlay}, for finite network depths) is trivial. It is interesting to observe that the sigmoid, the rectified linear unit, and the hyperbolic tangent non-linearities---all widely used in the deep learning literature---exhibit very different behavior in this regard, namely, they  do not  demodulate in the way the modulus non-linearity does \cite[Figure 6]{SPIEWiatowski}.  It is therefore unclear whether the proof machinery for energy conservation developed in this paper extends to these non-linearities or, for that matter, whether one gets energy decay and conservation at all. 

The feature map energy decay result \eqref{main_so0} relates to the feature vector energy conservation result \eqref{main} via the energy decomposition identity \eqref{eq:decomp}. Specifically, particularizing \eqref{eq:decomp} for Parseval frames, i.e.,  $A_n=B_n=1$, for all $n\in \mathbb{N}$, we get 
\begin{equation}\label{fdljslkkjjs}
\sum_{n=0}^{N-1}||| \Phi^n_\Omega(f)|||^2+W_N(f)= \|f\|_2^2.
\end{equation}
This shows that the input signal energy contained in the network layers $n\geq N$ is precisely given by $W_N(f)$. Thanks to $W_N(f)\to 0$ as $N\to \infty$ (established in Proposition \ref{prop2} in Appendix \ref{bum}) this residual energy  will eventually be collected in the infinite-depth feature vector $\Phi_\Omega(f)$ so that no input signal energy is ``lost''  in the network. In Section \ref{sec:nlay}, we shall answer the question of how many layers are needed to absorb $((1-\varepsilon)\cdot 100)\%$ of the input signal energy.

The next result shows that, under additional structural assumptions on the filters $\{ g_{\lambda_n}\}_{\lambda_n\in \Lambda}$, the guaranteed energy decay rate can be improved from polynomial to exponential. Specifically, we can get exponential energy decay for  broad families of wavelets and  Weyl-Heisenberg filters. For conceptual reasons, we consider the $1$-D case and, for simplicity of exposition, we employ filters that constitute Parseval frames and are identical across network layers.

\begin{theorem}\label{thm2}
Let $\widehat{r_l}:\R\to \R$, $\widehat{r_l}(\omega):=(1-|\omega|)^{l}_{+}$, with $l>1$. 
\begin{itemize}
\item[\textit{i)}]{Wavelets: Let the mother and father wavelets $\psi,\phi \in L^1(\R)\cap L^2(\R)$ satisfy $\text{supp}(\widehat{\psi})\subseteq [1/2,2]$ and 
\begin{equation}\label{eq:Waveletcondi}
 |\widehat{\phi}(\omega)|^2+\sum_{j=1}^{\infty} |\widehat{\psi}(2^{-j}\omega)|^2=1,\hspace{0.5cm} a.e. \ \omega\geq 0.\vspace{-0.1cm}
 \end{equation}
Moreover, let $g_j(x):=2^{j}\psi(2^jx)$, for $x\in \R$, $j\geq1$, and $g_j(x):=2^{|j|}\psi(-2^{|j|}x)$, for  $x\in \R$, $j\leq-1$, and set $\chi(x):=\phi(x)$, for $x \in \R$. Let $\Omega$ be the module-sequence \eqref{mods} with filters $\Psi=\{ \chi\}\cup\{g_j\}_{j\in \mathbb{Z}\backslash\{ 0\}}$ in every network layer. Then, 
\begin{equation}\label{main_h:wave}
W_N(f)\leq \int_{\R}\big|\widehat{f}(\omega)\big|^2\Big(1-\Big|\widehat{r_l}\Big(\frac{\omega}{(5/3)^{N-1}}\Big)\Big|^2\Big)\mathrm d\omega,
\end{equation}
for  $f\in L^2(\R)$ and  $N\geq 1$. Moreover, for every Sobolev function $f\in H^s(\R)$, $s>0$, we have 
\begin{equation}\label{main_sow1}
W_N(f)=\mathcal{O}\big((5/3)^{-\gamma N}\big),
\end{equation}
where $\gamma:= \min\{1,2s \}$.}
\item[\textit{ii)}]{Weyl-Heisenberg filters: For $R\in \mathbb{R}$, let the functions $g,\phi \in L^1(\R)\cap L^2(\R)$ satisfy $\text{supp}(\widehat{g})\subseteq [-R,R]$, $\widehat{g}(-\omega)=\widehat{g}(\omega)$, for $\omega \in \R$, and 
\begin{equation}\label{eq:Gaborcondi}
 |\widehat{\phi}(\omega)|^2+\sum_{k=1}^{\infty} |\widehat{g}(\omega-R(k+1))|^2=1,
 \end{equation}
a.e. $\omega\geq 0$. Moreover, let $g_k(x):=e^{2\pi i (k+1)Rx}g(x)$, for $x\in \R$, $k\geq1$, and $g_k(x):=e^{-2\pi i (|k|+1)Rx}g(x)$, for $x\in \R$, $k\leq-1$, and set $\chi(x):=\phi(x)$, for $x \in \R$. Let $\Omega$ be the module-sequence \eqref{mods} with filters $\Psi=\{ \chi\}\cup\{g_k\}_{k\in \mathbb{Z}\backslash\{ 0\}}$ in every network layer. Then, 
\begin{equation}\label{main_g:wave}
\hspace{-0.16cm}W_N(f)\leq \int_{\R}\big|\widehat{f}(\omega)\big|^2\Big(1-\Big|\widehat{r_l}\Big(\frac{\omega}{(3/2)^{N-1}R}\Big)\Big|^2\Big)\mathrm d\omega,
\end{equation}
for  $f\in L^2(\R)$ and  $N\geq 1$. Moreover, for every Sobolev function $f\in H^s(\R)$, $s>0$, we have
\begin{equation}\label{main_sog1}W_N(f) = \mathcal{O}\big((3/2)^{-\gamma N}\big),
\end{equation}
where $\gamma := \min \{1,2s\}$.}
\end{itemize}
\end{theorem}
\begin{proof}
See Appendix \ref{majabum}.
\end{proof}
The conditions we impose on the mother and father wavelet $\psi,\phi$ in i)  are satisfied, e.g., by analytic Meyer wavelets \cite[Section 3.3.5]{Daubechies}, and those on the prototype function $g$ and low-pass filter $\phi$ in ii) by B-splines \cite[Section 1]{Gaborexmpl}. Moreover, as shown in \cite[Theorem 3.1]{SPIEWiatowski},  the exponential energy decay results in \eqref{main_sow1} and \eqref{main_sog1} can be generalized to $\mathcal{O}(a^{-N})$ with \emph{arbitrary} decay factor $a>1$ realized through suitable choice of the mother wavelet or the Weyl-Heisenberg prototype function. 

We note that in the presence of pooling by sub-sampling (as defined in \cite[Eq. 9]{Wiatowski_journal}), say with pooling factors $S_n:=S\in [1,a)$, for all $n\in \mathbb{N}$, (where $a=\frac{5}{3}$ in the wavelet case and $a=\frac{3}{2}$ in the Weyl-Heisenberg case) the effective decay factor in \eqref{main_sow1} and \eqref{main_sog1} becomes $\frac{5}{3S}$ and $\frac{3}{2S}$, respectively. Exponential energy decay is hence compatible with vertical translation invariance according to \cite[Theorem 1]{Wiatowski_journal}, albeit at the cost of a slower (exponential) decay rate.  The proof of this statement is structurally very similar to that of Theorem \ref{thm2} and will therefore not be  given  here. Finally, we note that the energy decay and conservation results in Theorems \ref{thm3} and \ref{thm2} are compatible with the feature extractor $\Phi_\Omega$ being deformation-insensitive according to \cite[Theorem 2]{Wiatowski_journal}, simply by noting that \cite[Theorem 2]{Wiatowski_journal} applies to general semi-discrete frames and general Lipschitz-continuous non-linearities.

We next put the results in Theorems \ref{thm3} and \ref{thm2}  into perspective with respect to the literature.   

\textit{Relation to} \cite[Section 5]{Waldspurger}: The basic philosophy of our proof technique for \eqref{main_h:up1}, \eqref{main}, \eqref{main_h:wave}, and \eqref{main_g:wave} is inspired by the proof in \cite[Section 5]{Waldspurger}, which establishes \eqref{erghekj} and \eqref{eco} for scattering networks based on certain wavelet filters and with $1$-D real-valued input signals $f\in L^2(\R)$. Specifically, in \cite[Section 5]{Waldspurger}, in every network layer,  the filters $\Psi_W=\{\chi\}\cup\{g_{j}\}_{j\in \Z}$ (where $g_j(\omega):=2^{j}\psi(2^j\omega)$, $j\in \Z$, for some mother wavelet $\psi\in L^1(\R)\cap L^2(\R)$) are $1$-D functions  satisfying the frame property \eqref{PFP} with $A_n=B_n=1$,  $n\in \mathbb{N}$, a mild analyticity condition\footnote{At the time of completion of the present paper, I. Waldspurger kindly sent us a preprint \cite{SamptaIrene} which shows that the analyticity condition \cite[Eq. 5.5]{Waldspurger} on the mother wavelet is not needed for \eqref{erghekj} to hold.} \cite[Eq. 5.5]{Waldspurger} in the sense of $|\widehat{g_{j}}(\omega)|$,  $j\in \mathbb{Z}$, being larger for positive frequencies $\omega$ than for the corresponding negative ones, and a  vanishing moments condition \cite[Eq. 5.6]{Waldspurger} which controls the behavior of $\widehat{\psi}(\omega)$ around the origin according to $|\widehat{\psi}(\omega)|\leq C|\omega|^{1+\varepsilon}$, for $\omega \in \R$, for some $C,\varepsilon>0$. Similarly to the proof of \eqref{eco}  as given in  \cite[Section 5]{Waldspurger}, we base our proof of \eqref{main} on the energy decomposition identity \eqref{eq:decomp} and on an upper bound on $W_N(f)$ (see \eqref{erghekj} for the corresponding upper bound established in \cite[Section 5]{Waldspurger}) shown to go to zero as $N\to \infty$. The exponential energy decay results \eqref{main_so0}, \eqref{main_sow1}, and \eqref{main_sog1} for Sobolev functions $f\in H^s(\Rd)$ are entirely new. The major differences between \cite[Section 5]{Waldspurger} and our results are (i) that \eqref{erghekj} (reported in \cite[Section 5]{Waldspurger}) depends on an unspecified $a>1$, whereas our results in \eqref{main_h:up1}, \eqref{main_so0}, \eqref{main_h:wave},  \eqref{main_sow1}, \eqref{main_g:wave}, and \eqref{main_sog1} make the decay factor $a$ and the decay exponent $\alpha$ explicit, (ii) the technical elements employed to arrive at the upper bounds on $W_N(f)$; specifically, while the proof in \cite[Section 5]{Waldspurger} makes explicit use of the algebraic structure of the filters, namely, the multi-scale structure of wavelets, our proof of \eqref{main_h:up1} is oblivious to the algebraic structure of the filters, which is why it applies to general (possibly unstructured) filters that, in addition, can be different in different network layers, (iii) the assumptions imposed on  the filters, namely the analyticity and vanishing moments conditions in \cite[Eq. 5.5--5.6]{Waldspurger}, in contrast to our Assumption \ref{ass}, and (iv) the class of input signals $f$ the results apply  to, namely $1$-D real-valued signals in \cite[Section 5]{Waldspurger}, and $d$-dimensional complex-valued signals in our Theorem \ref{thm3}.

\textit{Relation to}  \cite{czaja2016uniform}: For scattering networks that are based on so-called uniform covering filters \cite{czaja2016uniform},   \eqref{czaja} and \eqref{eco} are established in \cite{czaja2016uniform} for $d$-dimensional complex-valued signals $f\in L^2(\Rd)$. Specifically, in \cite{czaja2016uniform}, in every network layer,  the $d$-dimensional filters $\{\chi\}\cup\{g_{\lambda}\}_{\lambda\in \Lambda}$ are taken to satisfy i) the frame property \eqref{PFP} with $A=B=1$ and hence $A_n=B_n=1$,  $n\in \mathbb{N}$, see \cite[Definition 2.1 (c)]{czaja2016uniform}, ii) a vanishing moments condition \cite[Definition 2.1 (a)]{czaja2016uniform} according to $\widehat{g_\lambda}(0)=0$,  for  $\lambda\in \Lambda$, and iii) a uniform covering condition \cite[Definition 2.1 (b)]{czaja2016uniform} which says that  the filters' Fourier transform support sets can be covered by a union of finitely many balls. The major differences between \cite{czaja2016uniform} and our results are as follows:  (i) the results in \cite{czaja2016uniform} apply exclusively to filters satisfying the uniform covering condition such as, e.g., Weyl-Heisenberg filters with a band-limited prototype function \cite[Proposition 2.3]{czaja2016uniform}, but do not apply to multi-scale filters such as wavelets, $(\alpha)$-curvelets, shearlets, and ridgelets (see \cite[Remark 2.2 (b)]{czaja2016uniform}), (ii) \eqref{czaja} as established in \cite{czaja2016uniform} leaves the decay factor $a>1$ unspecified, whereas our results in \eqref{main_sow1} and \eqref{main_sog1} make  the decay factor $a$ explicit (namely, $a=5/3$ in the wavelet case and  $a=3/2$ in the Weyl-Heisenberg case), (iii) the exponential energy decay result in \eqref{czaja} as established in \cite{czaja2016uniform} applies to all $f\in L^2(\Rd)$ and thus, in particular, to  Sobolev input signals (owing to $H^s(\Rd)\subseteq L^2(\Rd)$, for all $s>0$), whereas our decay results in \eqref{main_so0}, \eqref{main_sow1}, and \eqref{main_sog1} pertain to Sobolev input signals $f\in H^s(\Rd)$, $s>0$, only, (iv) the technical elements employed to arrive at the upper bounds on $W_N(f)$, specifically, while the proof in  \cite{czaja2016uniform}  makes explicit use of the uniform covering property of the filters, our proof of \eqref{main_h:up1} is completely oblivious to the (algebraic) structure of the filters, (v) the assumptions imposed on  the filters, i.e., the vanishing moments and uniform covering condition in \cite[Definition 2.1 (a)-(b)]{czaja2016uniform}, in contrast to our Assumption \ref{ass}, which is less restrictive, and thereby makes our results in Theorem \ref{thm3} apply to general (possibly unstructured) filters that, in addition, can be different in different network layers.

\section{Number of layers needed}\label{sec:nlay}
DCNNs  used in practice employ potentially hundreds of layers \cite{he2015deep}. Such network depths entail formidable computational challenges both in training and in operating the network. It is therefore important to understand how many layers are needed to have most of the input signal energy be contained in the feature vector. This will be done by considering Parseval frames in all layers, i.e., frames with frame bounds $A_n=B_n=1$, $n\in \mathbb{N}$. The energy conservation result \eqref{main} then implies that the infinite-depth feature vector $$\Phi_\Omega(f) = \bigcup_{n=0}^{\infty}\Phi^n_\Omega(f)$$ contains the entire input signal energy according to $
|||\Phi_\Omega(f)|||^2 = \sum_{n=0}^\infty |||\Phi^n_\Omega(f)|||^2 = \|f\|_2^2. 
$ 
Now, the decomposition \eqref{fdljslkkjjs} reveals that thanks to $\lim\limits_{N\to \infty} W_N(f)\to 0$, increasing the network depth $N$ implies that the feature vector $\bigcup_{n=0}^{N}\Phi^n_\Omega(f)$ progressively contains a larger fraction of the input signal energy. We formalize the question on the number of layers needed by asking for bounds of the form 
\begin{equation}\label{bubu_main}
(1-\varepsilon)\leq \frac{ \sum_{n=0}^{N} |||\Phi_\Omega^n(f)|||^2}{\|f\|^2_2}\leq1,
\end{equation}
i.e., by determining the network depth $N$  guaranteeing  that at least $((1-\varepsilon)\cdot 100)\%$ of the input signal energy are captured by the corresponding depth-$N$ feature vector  $\bigcup_{n=0}^{N}\Phi^n_\Omega(f)$. Moreover, \eqref{bubu_main} ensures that the depth-$N$ feature extractor $\bigcup_{n=0}^{N}\Phi^n_\Omega$ exhibits a trivial null-set.

The following results establish handy estimates of the number $N$ of layers  needed to guarantee \eqref{bubu_main}. For pedagogical reasons, we start with the case of band-limited input signals and then proceed to a more general statement  pertaining to Sobolev functions.  
 
\begin{corollary}\label{cor1}
\begin{itemize}
\item[]
\item[\textit{i)}]{Let $\Omega$ be the module-sequence \eqref{mods} with filters $\{ g_{\lambda_n}\}_{\lambda_n \in \Lambda_n}$ satisfying the conditions in Assumption \ref{ass}, and let the corresponding frame bounds be $A_n=B_n=1$,  $n\in \mathbb{N}$.  Let $\delta>0$ be the radius of the spectral gap $B_{\delta}(0)$ left by the filters $\{ g_{\lambda_n}\}_{\lambda_n \in \Lambda_n}$ according to \eqref{eq:ass2}. Furthermore, let $l>\lfloor d/2 \rfloor+1$, $\eps \in (0,1)$, $\alpha$ as defined in \eqref{defalp}, and  $f\in L^2(\Rd)$  $L$-band-limited. If
\begin{equation}\label{hihih}
N\geq \Bigg\lceil \Bigg(\frac{L}{(1-(1-\varepsilon)^{\frac{1}{2l}})\delta}\Bigg)^{1/\alpha}-1\Bigg\rceil,
\end{equation}
then \eqref{bubu_main} holds.
}
\item[\textit{ii)}]{ Assume that the conditions in Theorem \ref{thm2} i) and ii) hold. For the wavelet case, let $a=\frac{5}{3}$ and $\delta=1$ \big(where $\delta$ corresponds to the radius of the spectral gap left by the wavelets $\{ g_{j}\}_{j\in \mathbb{Z}\backslash\{0\}}$\big). For the Weyl-Heisenberg case, let $a=\frac{3}{2}$ and $\delta=R$ \big(here, $\delta$ corresponds to the radius of the spectral gap left by the Weyl-Heisenberg filters $\{ g_{k}\}_{k\in \mathbb{Z}\backslash\{0\}}$\big). Moreover,  let $l>1$, $\eps \in (0,1)$, and  $f\in L^2(\R)$  $L$-band-limited. If
\begin{equation}\label{hihih2}
N\geq \Bigg\lceil \log_a\hspace{-0.1cm}\Bigg(\frac{L}{(1-(1-\varepsilon)^{\frac{1}{2l}})\delta}\Bigg)\Bigg\rceil,
\end{equation}
then \eqref{bubu_main} holds in both cases.
}
\end{itemize}
\end{corollary}
\begin{proof}
See Appendix \ref{aca}. 
\end{proof}
\begin{table}
\vspace{0.2cm}
\centering
\setlength{\tabcolsep}{5pt}
{\small
\begin{tabular}{ l | c c c c c c | }
 & \multicolumn{6}{ c | }{$(1-\varepsilon)$}\\
 & 0.25 & 0.5 & 0.75& 0.9 & 0.95 & 0.99 \\ \hline
wavelets & 2 & 3 & 4 &  6 & 8  & 11 \\
Weyl-Heisenberg filters& 2 & 4 & 5 & 8 & 10 & 14   \\
 general filters & 2 & 3 & 7 & 19 & 39  & 199\\
\hline
\end{tabular}}
\caption{\label{tab:svmclass}  Number $N$ of layers needed to ensure that $((1-\varepsilon)\cdot 100)\%$ of the input signal energy are contained in the features generated in the first $N$ network layers.}
\label{tab2}
\end{table}
Corollary \ref{cor1} nicely shows how the description complexity of the signal class under consideration, namely the bandwidth $L$ and the dimension $d$ through the decay exponent $\alpha$ defined in \eqref{defalp} determine the number $N$ of layers  needed. Specifically, \eqref{hihih} and \eqref{hihih2} show that larger bandwidths $L$ and larger dimension $d$ render the input signal $f$ more ``complex'', which requires deeper networks to capture most of the energy of $f$.  The dependence of the lower bounds in \eqref{hihih} and \eqref{hihih2} on the network properties, through the module-sequence $\Omega$, is through the decay factor $a>1$ and the radius $\delta$ of the spectral gap left by the filters $\{ g_{\lambda_n}\}_{\lambda_n \in \Lambda_n}$.  

The following numerical example provides quantitative insights on the influence of the parameter $\eps$ on \eqref{hihih} and \eqref{hihih2}. Specifically, we set $L=1$, $\delta=1$, $d=1$ (which implies $\alpha=1$, see \eqref{defalp}), $l=1.0001$, and show in Table  \ref{tab2} the number $N$ of layers needed according to \eqref{hihih} and \eqref{hihih2} for different values of $\varepsilon$. The results show that $95\%$ of the input signal energy are contained in the first $8$ layers in the wavelet case and in the first $10$ layers in the Weyl-Heisenberg case. We can therefore conclude that in practice a relatively small number of layers is needed to have most of the input signal energy be contained in the feature vector. In contrast, for general filters, where we can guarantee  polynomial energy decay only,  $N=39$ layers are needed to absorb $95\%$ of the input signal energy. We hasten to add, however, that \eqref{main_h:up1} simply \emph{guarantees} polynomial energy decay and does not preclude the energy from decaying faster in practice.

We proceed with the estimates  for Sobolev-class input signals.

\begin{corollary}\label{cor3}
\begin{itemize}
\item[]
\item[\textit{i)}]{Let $\Omega$ be the module-sequence \eqref{mods} with filters $\{ g_{\lambda_n}\}_{\lambda_n \in \Lambda_n}$ satisfying the conditions in Assumption \ref{ass}, and let the corresponding frame bounds be $A_n=B_n=1$,  $n\in \mathbb{N}$. Let $\delta>0$ be the radius of the spectral gap $B_{\delta}(0)$ left by the filters $\{ g_{\lambda_n}\}_{\lambda_n \in \Lambda_n}$ according to \eqref{eq:ass2}. Furthermore, let $l>\lfloor d/2 \rfloor+1$, $\eps \in (0,1)$, $\alpha$ as defined in \eqref{defalp}, and $f\in H^s(\Rd)\backslash\{ 0\}$, for  $s>0$. If
\begin{equation}\label{hihihs}
N\geq \Bigg\lceil \Bigg(\frac{2l\, \| f\|_{H^s}^{2/{\gamma}}}{\varepsilon^{1/\gamma} \,\delta \| f\|_2^{2/{\gamma}}} \Bigg)^{1/\alpha}-1\Bigg\rceil,
\end{equation}
where $\gamma:=\min\{1,2s\}$, 
then \eqref{bubu_main} holds.}
\item[\textit{ii)}]{	Assume that the conditions in Theorem \ref{thm2} i) and ii) hold. For the wavelet case, let $a=\frac{5}{3}$ and $\delta=1$ \big(where $\delta$ corresponds to the radius of the spectral gap left by the wavelets $\{ g_{j}\}_{j\in \mathbb{Z}\backslash\{0\}}$\big). For the Weyl-Heisenberg case, let $a=\frac{3}{2}$ and $\delta=R$ \big(here, $\delta$ corresponds to the radius of the spectral gap left by the Weyl-Heisenberg filters $\{ g_{k}\}_{k\in \mathbb{Z}\backslash\{0\}}$\big). Furthermore, let $l>1$, $\eps \in (0,1)$, and  $f\in H^s(\R)\backslash\{ 0\}$, for $s>0$. If 
\begin{equation}\label{hihihs2}
N\geq \Bigg\lceil \log_a\Bigg(\frac{2l\, \| f\|_{H^s}^{2/{\gamma}}}{\varepsilon^{1/\gamma} \,\delta \| f\|_2^{2/{\gamma}}} \Bigg)\Bigg\rceil,
\end{equation}
where $\gamma:=\min\{1,2s\}$, 
then \eqref{bubu_main} holds.}
\end{itemize}
\end{corollary}
\begin{proof}
See Appendix \ref{acas}. 
\end{proof}
As already mentioned in Section \ref{Sec:prop}, Sobolev spaces $H^s(\Rd)$ contain a wide range of practically relevant signal classes.  The results in Corollary \ref{cor3} therefore provide---for a wide variety of input signals---a  picture of how many layers are needed to have most of the input signal energy be contained in the feature vector.

The width of the networks considered throughout the paper is, in principle, infinite as the sets $\Lambda_n$ need to be countably infinite in order to guarantee that the frame property \eqref{PFP} is satisfied. For input signals that exhibit mild spectral decay, the number of ``operationally  significant nodes'' will, however, be finite in practice. For a treatment of this aspect as well as results on  depth-width tradeoffs, the interested reader is referred to  \cite{SPIEWiatowski}.

\appendices
\section{A feature extractor with a non-trivial kernel}\label{counterexmp}
We show, by way of example, that employing filters $\Psi_n$ which satisfy the frame property \eqref{PFP} alone does not guarantee a trivial null-set for the feature extractor $\Phi_\Omega$. Specifically, we construct a feature extractor $\Phi_\Omega$ based on filters satisfying \eqref{PFP} and a corresponding function $f\neq 0$ with $f\in \mathcal{N}(\Phi_\Omega)$.

Our example employs, in every network layer, filters $\Psi=\{ \chi\}\cup\{ g_k\}_{k\in \mathbb{Z}}$ that satisfy the Littlewood-Paley condition \eqref{eq:tight} with $A=B=1$,   and where $g_{0}$ is such that $\widehat{g_0}(\omega)=1$, for $\omega \in \overline{B_{1}(0)}$, and arbitrary else (of course, as long as the Littlewood-Paley condition \eqref{eq:tight} with $A=B=1$ is satisfied). We emphasize that no further restrictions are imposed on the filters $\{ \chi\}\cup\{ g_k\}_{k\in \mathbb{Z}}$, specifically $\chi$ need not be of low-pass nature and the filters $\{ g_k\}_{k\in \mathbb{Z}}$ may be structured (such as wavelets \cite[Appendix B]{Wiatowski_journal}) or unstructured (such as random filters \cite{Jarrett,hierachies}), as long as they satisfy the Littlewood-Paley condition \eqref{eq:tight} with $A=B=1$. Now, consider the input signal $f\in L^2(\Rd)$  according to $$\widehat{f}(\omega):=(1-|\omega|)^{l}_{+},\quad \omega \in \Rd,$$ with $l>\lfloor d/2 \rfloor+1$. Then $f\ast g_0=f,$ owing to $\text{supp}(\widehat{f}\,)= \overline{B_1(0)}$ and $\widehat{g_0}(\omega)=1$, for $\omega \in \overline{B_1(0)}$. Moreover, $\widehat{f}$  is a positive definite radial basis function  \cite[Theorem 6.20]{Wendland} and hence by \cite[Theorem 6.18]{Wendland} $f(x)\geq 0$,  $x \in \Rd$, which, in turn, implies $|f|=f$. This yields
 $$
 U[q^N_0]f= \big|\cdots \big||f\ast g_0|\ast g_0\big|\cdots\ast g_0\big|=f,\vspace{0.08cm}
  $$
for $q^N_{0}:=(0,0,\dots,0)\in \mathbb{Z}^N$ and $N\in \mathbb{N}$. Owing to the energy decomposition identity \eqref{eq:decomp}, together with $A_\Omega^N=B_\Omega^N=1$,  $N\in \mathbb{N}$, which, in turn, is by $A_n=B_n=1$,  $n\in \N$, we have
\begin{align*}
&\|f\|_2^2=\sum_{n=0}^{N-1}||| \Phi^n_\Omega(f)|||^2+W_N(f)\nonumber\\
=&\sum_{n=0}^{N-1}||| \Phi^n_\Omega(f)|||^2+\underbrace{\| U[q^N_{0}] f\|_2^2}_{=\,\|f\|_2^2}+\sum_{q\,\in\, \Z^N\backslash \{q_0^N\}}|| U[q] f||_2^2,\label{labum}
\end{align*}
for $N\in \mathbb{N}$. This implies 
\begin{equation}\label{kfjkfjjfkjss}
\sum_{n=0}^{N-1}||| \Phi^n_\Omega(f)|||^2+\sum_{q\,\in\, \Z^N\backslash \{q_0^N\}}|| U[q] f||_2^2=0.
\end{equation}
As both terms in \eqref{kfjkfjjfkjss} are positive, we can conclude that $\sum_{n=0}^{N-1}||| \Phi^n_\Omega(f)|||^2=0$,  $N\in \mathbb{N}$, and thus $$||| \Phi_\Omega(f)|||^2=\sum_{n=0}^\infty|||\Phi_\Omega^n(f)|||^2=0.$$ Since $||| \Phi_\Omega(f)|||^2=0$ implies $\Phi_\Omega(f)=0$, we have constructed a non-zero $f$, namely $$f(x)=\int_{\Rd}(1-|\omega|)^{l}_{+}e^{\,2\pi i \langle x,\,\omega\rangle}\mathrm d\omega,$$ that maps to the all-zeros feature vector, i.e., $f\in \mathcal{N}(\Phi_\Omega)$.

The point of this example is the following. Owing to the nature of $\widehat{g_0}(\omega)$ (namely, $\widehat{g_0}(\omega)=1$, for $\omega \in \overline{B_{1}(0)}$) and the Littlewood-Paley condition
$$ |\widehat{\chi}(\omega)|^2+\sum_{k \in \mathbb{Z}}|\widehat{g_{k}}(\omega)|^2=1,\hspace{0.25cm} \text{a.e. } \omega \in \Rd,$$
it follows that neither the output-generating filter $\chi$ nor any of the other filters $g_k$, $k\in \mathbb{Z}\backslash\{0\}$, can have spectral support in $\overline{B_1(0)}$. Consequently, the only non-zero contribution to the feature vector can  come from $$U[q^N_{0}]f=f,$$ which, however, thanks to $\text{supp}(\widehat{f}\,)= \overline{B_1(0)}$, is spectrally disjoint from the output-generating filter $\chi$. Therefore, $\Phi_\Omega(f)$ will be identically equal to $0$. Assumption \ref{ass} disallows this situation as it forces the filters $g_k$, $k\in \mathbb{Z}$, to be of high-pass nature which, in turn, implies that $\chi$ must  have low-pass characteristics. The punch-line of our general results on energy conservation, be it for finite $N$ or for $N\to\infty$, is that Assumption \ref{ass} in combination with the frame property  and  the modulus non-linearity prohibit a non-trivial null-set \emph{in general}.

\vspace{-0.1cm}
\section{Sobolev smoothness of cartoon functions}\label{sec:proofcart}
Cartoon functions, introduced in \cite{Cartoon}, satisfy mild decay properties and are piecewise continuously differentiable apart from curved discontinuities along smooth hypersurfaces. This function class has been widely adopted in the mathematical signal processing literature \cite{grohs_wiatowski,wiatowski2016discrete,Grohs_alpha,grohs2014parabolic,ShearletsIntro} as a standard model for natural images such as, e.g., images of handwritten digits \cite{MNIST} (see Figure  \ref{fig:data}). We will work with the following---relative to the definition in \cite{Cartoon}---slightly modified version of cartoon functions. 

\begin{definition}\label{def:cartoon}
The function $f:\R^d \to \mathbb{C}$ is referred to as a cartoon function if it can be written as $f = f_1 + \mathds{1}_Df_2$, where $D\subseteq \R^d$ is a compact domain whose boundary $\partial D$ is a compact topologically embedded $C^2$-hypersurface of $\Rd$ without boundary\footnote{We refer the reader to \cite[Chapter 0]{RiemannianBook} for a review on differentiable manifolds.}, and $f_i\in H^{1/2}(\Rd)\cap C^1(\R^d,\mathbb{C})$, $i=1,2$, satisfy the decay condition 
\begin{align*}
|\nabla f_i(x)|&\leq C \langle x\rangle^{-d},\quad i=1,2,
\end{align*}
for some $C>0$ (not depending on $f_1$,$f_2$). Furthermore, we denote by 
\begin{align*}
&\mathcal{C}^{K}_{\mathrm{CART}}:=\{f_1 + \mathds{1}_Df_2 \ | \  f_i\in H^{1/2}(\Rd)\cap C^1(\R^d,\mathbb{C}), \\ 
 &|\nabla f_i(x)|\le K \langle x\rangle^{-d},\ \emph{vol}^{d-1}(\partial D)\leq K, \ \|f_2\|_\infty \leq K\}
\end{align*} 
the class of cartoon functions of  ``size'' $K>0$.
\end{definition}
Even though cartoon functions are in general discontinuous, they still admit Sobolev smoothness. The following result formalizes this statement. 

\begin{lemma}
\label{remark_sob}
Let $K>0$. Then, $\mathcal{C}^{K}_{\mathrm{CART}} \subseteq H^s(\Rd)$, for all $s\in (0,1/2)$. 
\end{lemma}
\begin{proof}
Let $(f_1 + \mathds{1}_Df_2) \in \mathcal{C}^{K}_{\mathrm{CART}}$. We first establish $\mathds{1}_D\in H^{s}(\Rd)$, for all $s\in (0,1/2)$. To this end, we define the Sobolev-Slobodeckij semi-norm \cite[Section 2.1.2]{runst1996sobolev}
\begin{equation*}\label{eq:Soboldef}
	|f|_{H^s} :=
	\Big(\int_{\mathbb{R}^d}\int_{\mathbb{R}^d}
	\frac{|f(x)-f(y)|^2}{|x-y|^{2s+d}}\mathrm dx\, \mathrm dy \Big)^{1/s},
\end{equation*}
and note that, thanks to \cite[Section 2.1.2]{runst1996sobolev}, $\mathds{1}_D\in H^s(\Rd)$ if $|\mathds{1}_D|_{H^s}<\infty$. We have 
\begin{align*}
	|\mathds{1}_D|_{H^s}^s &=\int_{\mathbb{R}^d}\int_{\mathbb{R}^d} \frac{|\mathds{1}_D(x)-\mathds{1}_D(y)|^2}{|x-y|^{2s+d}}\mathrm dx \,\mathrm dy \\
	&=\int_{\mathbb{R}^d}\frac{1}{|t|^{2s+d}}\int_{\mathbb{R}^d} |\mathds{1}_D(x)-\mathds{1}_D(x-t)|^2\mathrm dx \, \mathrm dt,
\end{align*}
where we employed the change of variables $t=x-y$. Next, we note that, for fixed $t\in \Rd$, the function $$h_t(x):=|\mathds{1}_D(x) - \mathds{1}_D(x-t)|^2$$ satisfies $h_t(x)=1$, for $x\in S_t$, where
\begin{align}
S_t:=&\,\{ x\in \R^d \, |\, x\in D \  \text{ and }\  x-t\notin D \} \nonumber\\
 \cup& \,\{ x\in \R^d \, |\, x\notin D\  \text{ and }\  x-t\in D \}\nonumber\\
 =&D\Delta (D+t)\label{djhdskjffhkjh},
\end{align}
and $h_t(x)=0$, for $x\in \R^d\backslash S_t$. It follows from \eqref{djhdskjffhkjh}  that 
\begin{equation}\label{fhfjhfjhjfss}
\mbox{vol}^{d}(S_t) \leq 2\,\mbox{vol}^{d}(D), \quad \forall \, t\in \Rd.
\end{equation}
Moreover, owing to  $S_t\subseteq \big(\partial{D} + B_{|t|}(0)\big)$, where $(\partial{D} + B_{|t|}(0))$ is a tube of radius $|t|$ around the boundary $\partial{D}$ of $D$ (see Figure \ref{fig:set}), and Lemma \ref{lemma_tube}, stated below, there exists a  constant $C_{\partial D}>0$  such that 
\begin{equation}\label{fhfjhfjhjfss2}
\mbox{vol}^{d}(S_t)\leq \mbox{vol}^{d}(\partial{D} + B_{|t|}(0)) \leq C_{\partial D}|t|,
\end{equation}
for all $t\in \Rd$ with $|t| \leq 1$.  Next, fix $R$ such that $0<R<1$. Then, 
\begin{align}
	&|\mathds{1}_D|_{H^s}^s = \int_{\mathbb{R}^d}\frac{1}{|t|^{2s+d}}\int_{\mathbb{R}^d} |\mathds{1}_D(x)-\mathds{1}_D(x-t)|^2\mathrm dx \, \mathrm dt \nonumber\\
	&= \int_{\mathbb{R}^d}\frac{1}{|t|^{2s+d}}\int_{\mathbb{R}^d} h_t(x)\mathrm dx \, \mathrm dt \nonumber\\
	&= \int_{\mathbb{R}^d}\frac{1}{|t|^{2s+d}}\int_{S_t} 1\mathrm dx \, \mathrm dt = \int_{\mathbb{R}^d}\frac{\mbox{vol}^d(S_t)}{|t|^{2s+d}} \mathrm dt \nonumber\\
	&\leq \int_{\Rd\backslash B_R(0)}\frac{2\,\mbox{vol}^{d}(D)}{|t|^{2s+d}} \mathrm dt
	+\int_{B_R(0)}\frac{C_{\partial D}}{|t|^{2s+d-1}} \mathrm dt\label{kjhsjfjkdhfdjhjhjj}\\
	&=2\,\mbox{vol}^{d}(D)\,\mbox{vol}^{d-1}(\partial B_1(0)) \underbrace{\int_{R}^\infty r^{-(2s+1)}\mathrm dr}_{=:I_1}\nonumber\\
	&+C_{\partial D}\,\mbox{vol}^{d-1}(\partial B_1(0))\underbrace{\int_{0}^R r^{-2s}\mathrm dr}_{=:I_2}\label{Sobolcharest},
\end{align}
where in \eqref{kjhsjfjkdhfdjhjhjj} we employed \eqref{fhfjhfjhjfss} and \eqref{fhfjhfjhjfss2}, and in the last step we introduced polar coordinates. The  integral $I_1$ is finite for all $s>0$, while $I_2$ is finite for all $s<1/2$. Moreover, $\mbox{vol}^{d}(D)=\int_{D}1\mathrm dx$ is finite owing to $D$ being compact (and thus bounded). We can therefore conclude that \eqref{Sobolcharest} is finite for $s\in (0,1/2)$, and hence $\mathds{1}_D\in H^s(\Rd)$, for  $s \in (0,1/2)$. To see that $(f_1 + \mathds{1}_Df_2) \in H^s(\Rd)$, for $s\in (0,1/2)$, we first note that \begin{align}
|f_1 + \mathds{1}_Df_2|_{H^s}&\leq |f_1|_{H^s} + |\mathds{1}_Df_2|_{H^s},\label{dfjhdjhjdhgfj}
\end{align}
which is thanks to the sub-additivity of the semi-norm $|\cdot|_{H^s}$. Now, the first term on the RHS of \eqref{dfjhdjhjdhgfj} is finite owing to $f_1\in H^{1/2}(\Rd)\subseteq H^s(\Rd)$, for all $s\in (0,1/2)$. For the second term on the RHS, we start by noting that
\begin{align}
&|\mathds{1}_Df_2|_{H^s}^s = \nonumber\\
&\Big(\int_{\mathbb{R}^d}\int_{\mathbb{R}^d}
	\frac{|(\mathds{1}_Df_2)(x)-(\mathds{1}_Df_2)(y)|^2}{|x-y|^{2s+d}}\mathrm dx\, \mathrm dy \Big)\label{shfdksfhjhkss}
\end{align}
and
\begin{align}
&|(\mathds{1}_Df_2)(x)-(\mathds{1}_Df_2)(y)|^2\nonumber\\
&= |(\mathds{1}_D(x)- \mathds{1}_D(y))f_2(x) + (f_2(x) - f_2(y))\mathds{1}_D(y)|^2\nonumber\\
&\leq 2|(\mathds{1}_D(x)- \mathds{1}_D(y))|^2 |f_2(x)|^2 \label{dkjkjdkjd}\\
&+ 2|(f_2(x) - f_2(y))|^2|\mathds{1}_D(y)|^2\label{dkjkjdkjddldkldkdl},
\end{align}
where  \eqref{dkjkjdkjd} and \eqref{dkjkjdkjddldkldkdl} are thanks to $|a+b|^2\leq 2|a|^2 + 2|b|^2$,  for $a,b \in \mathbb{C}$. Substituting \eqref{dkjkjdkjd} and \eqref{dkjkjdkjddldkldkdl} into \eqref{shfdksfhjhkss} and noting that $|f_2(x)|^2\leq \|f_2 \|^2_\infty\leq K^2$, $x\in \Rd$, which is by assumption, and $\mathds{1}_D(y)\leq 1$, $y\in \Rd$, implies
\begin{align}
|\mathds{1}_Df_2|_{H^s}^s \leq 2K^2 |\mathds{1}_D|_{H^s}^s + 2|f_2|_{H^s}^s<\infty,
\end{align}
where in the last step we used $\mathds{1}_D\in H^s(\Rd)$, established above, and $f_2\in H^{1/2}(\Rd)\subseteq H^s(\Rd)$, both for all $s\in (0,1/2)$. This completes the proof.
\end{proof}
It remains to establish the second  inequality in \eqref{fhfjhfjhjfss2}. 

\tikzset{
    circ/.pic={ 
    \fill (0,0) circle (2pt);
    }
}
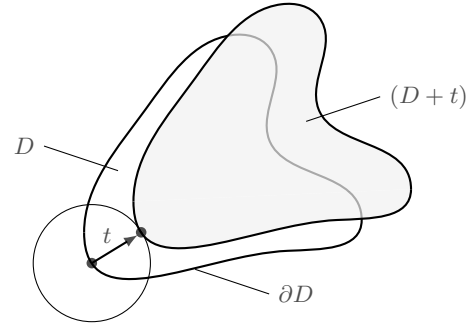
\begin{figure}
\begin{center}
\vspace{0.05cm}
\begin{tikzpicture}
	\begin{scope}[scale=.82]
	\pgfsetfillopacity{.7}

\coordinate (a0) at (0.125,-.5);
\coordinate (a1) at (0,0);
\coordinate (a2) at (3.3,-0.3);
\coordinate (a3) at (4.5,0.2);
\coordinate (a4) at (3,1.5);
\coordinate (a5) at (3,3);
\coordinate (a6) at (2.4,3.2);
\coordinate (a7) at (1,2);

\coordinate (b0) at (0.925,0);
\coordinate (b1) at (0.8,0.5);
\coordinate (b2) at (4.1,.2);
\coordinate (b3) at (5.3,0.7);
\coordinate (b4) at (3.8,2);
\coordinate (b5) at (3.8,3.5);
\coordinate (b6) at (3.2,3.7);
\coordinate (b7) at (1.8,2.5);

\filldraw[fill=black!0!white, draw=black,use Hobby shortcut,thick]([out angle=-90]a1)..(a2)..([in angle=-90]a3); 
\filldraw[fill=black!0!white, draw=black,use Hobby shortcut,thick] ([out angle=90]a3)..(a4)..(a5)..(a6)..(a7)..([in angle=90]a1);

\filldraw[fill=black!5!white, draw=black,use Hobby shortcut,thick]([out angle=-90]b1)..(b2)..([in angle=-90]b3); 
\filldraw[fill=black!5!white, draw=black,use Hobby shortcut,thick] ([out angle=90]b3)..(b4)..(b5)..(b6)..(b7)..([in angle=90]b1);
\draw (3.6cm,1.8cm) -- (4.8 cm,2.2cm) node[anchor= west] {\small$(D+t)$};
\draw (.55 cm,1.cm)--(-.65cm,1.4cm)  node[anchor= east] {\small$D$};
\draw (1.8 cm,-.6cm) -- (3cm,-1.cm) node[anchor= west] {\small$\partial D$};

\draw [->,>=latex,thick] (0.125,-.5) -- (0.87,-.045);

\draw (.15,-.05) node[anchor= west] {\small$t$};

\draw (a0) circle (.95cm);

\foreach \i in  {0} {\draw (a\i) pic {circ};} 
\foreach \i in  {0} {\draw (b\i) pic {circ};} 

	\end{scope}
\end{tikzpicture}
\end{center}
\vspace{0.3cm}
\caption{Illustration in dimension $d=2$. The set $(D+t)$ (grey) is obtained by translating the set $D$ (white) by $t\in \R^2$. The symmetric difference $D\Delta (D+t)$ is contained in $(\partial{D} + B_{|t|}(0))$, the tube  of radius $|t|$ around the boundary $\partial{D}$ of $D$.}
\label{fig:set}
\end{figure}
\begin{lemma}\label{lemma_tube}
Let $M$ be a compact topologically embedded $C^2$-hypersurface of $\Rd$ without boundary and let 
$$
T(M,r):= \big\{ x \in \Rd \ \big| \ \inf_{y\in M}| x - y| \leq r\big\}, \quad r>0,
$$ 
be the tube of radius $r$ around $M$. Then, there exists a  constant $C_M>0$ (that does not depend on $r$) such that for all $r \leq 1$ it holds that
 \begin{equation}\label{dssndjkfkjs}
\mbox{vol}^{d}(T(M,r))\leq C_Mr.
\end{equation}
\end{lemma} 
\begin{proof}The proof is based on Weyl's tube formula \cite{WeylTubes}.
	Let 
	$$\kappa:= \max_{i \in \{1,\dots,d-1\}} \kappa_i,$$
	where $\kappa_i$ is the $i$-th principal curvature of the hypersurface $M$ (see \cite[Section 3.1]{TubeBook} for a formal definition). 
 It follows from \cite[Theorem 8.4 (i)]{TubeBook} that
$$
\mbox{vol}^{d}(T(M,r)) = \sum_{i=0}^{\lfloor \frac{d-1}{2} \rfloor} \frac{2r^{2i+1}k_{2i}(M)}{\prod_{j=0}^i(1+2j)},
$$
for all $r \leq \kappa^{-1}$, where $k_{2i}(M) = \int_{M} H_{2i}(x)\mathrm dx$, $i \in \{0, \dots, \lfloor \frac{d-1}{2} \rfloor\}$, with $H_{2i}$ denoting the so-called $(2i)$-th  curvature of $M$, see \cite[Section 4.1]{TubeBook} for a formal definition. Now, thanks to $M$ being a $C^2$-hypersurface, we have that $H_{2i}$, $i \in \{0, \dots, \lfloor \frac{d-1}{2} \rfloor\}$, is bounded (see \cite[Section 4.1]{TubeBook}), which together with $M$ compact (and thus bounded) implies  $|k_{2i}(M)|<\infty$, for all $i \in \{0, \dots, \lfloor \frac{d-1}{2} \rfloor\}$. Moreover, by definition, $k_{2i}(M)$, $i \in \{0, \dots, \lfloor \frac{d-1}{2} \rfloor\}$, is independent of the tube radius $r$. Therefore, setting 
$$C_M:=\Big(\Big\lfloor \frac{d-1}{2} \Big\rfloor + 1\Big)\cdot \max_{i}\frac{2|k_{2i}(M)|}{\prod_{j=0}^i(1+2j)}$$ 
 establishes \eqref{dssndjkfkjs} for $0<r\leq \min\{1,\kappa^{-1}\}$. It remains to prove \eqref{dssndjkfkjs} for  $\min\{1,\kappa^{-1}\}<r\le 1$.
 Let $$R^\ast:=\inf\{R>0:\ M\subseteq B_R(0)\}$$ and $D_{R^\ast}:=\mbox{vol}^d(B_{R^\ast+1}(0))$. Since
 $$
	 \mbox{vol}^d(T(M,r))\le D_{R^\ast},\quad \forall \, 0<r\le 1,
 $$
it follows that
 $$
	 \mbox{vol}^d(T(M,r))< D_{R^\ast}\cdot \,\max\{1,\kappa\}\cdot r 
$$
for all $\min\{1,\kappa^{-1}\}<r\le 1$, which establishes \eqref{dssndjkfkjs} for $\min\{1,\kappa^{-1}\}<r\le 1$ and thereby concludes the proof. 
\end{proof}

\section{Proof of statement  i) in Theorem \ref{thm3}}\label{app:prop2}
We start by establishing \eqref{main_h:up1} with  $\alpha=\log_2(\sqrt{d/(d-1/2)})$, for  $d\geq1$. Then, we sharpen our result in the $1$-D case by proving that \eqref{main_h:up1} holds for $d=1$ with $\alpha=1$. This leads to a significant improvement, in the $1$-D case, of the decay exponent from $\log_2(\sqrt{d/(d-1/2)})=\frac{1}{2}$ to $1$. The idea for the proof of \eqref{main_h:up1} for  $\alpha=\log_2(\sqrt{d/(d-1/2)})$,   $d\geq1$, is to establish that\footnote{We prove the more general result \eqref{main_h:up} for technical reasons, concretely in order to be able to argue by induction over path lengths with flexible starting index $n$.}
\begin{align}
&\sum_{q\,\in\, \Lambda_n \times\Lambda_{n+1}\times\cdots\times \Lambda_{n+N-1}} \|U[q]f\|_2^2\nonumber\\
&\leq 
C_{n}^{n+N-1} \int_{\Rd}\big|\widehat{f}(\omega)\big|^2\Big(1-\Big|\widehat{r_l}\Big(\frac{\omega}{N^{\alpha}\delta}\Big)\Big|^2\Big)\mathrm d\omega, \label{main_h:up}
\end{align}
for  $N\in \mathbb{N}$, where 
$$
C_{n}^{n+N-1}:=\prod_{k=n}^{n+N-1}\max\{1,B_k \}.
$$
Setting $n=1$ in  \eqref{main_h:up} and noting that $C_{1}^N=B_\Omega^N$ yields the desired result \eqref{main_h:up1}. We proceed by induction over the path length $\ell(q):=N$, for $q=(\lambda_{n},\lambda_{n+1},$\mydots$, \lambda_{n+N-1})\in \Lambda_n\times \Lambda_{n+1}\times \dots \times \Lambda_{n+N-1}$. Starting with the base case $N=1$, we have 
\begin{align}
&\sum_{q\,\in\, \Lambda_n}\|U[q]f\|^2_2=\sum_{\lambda_n\in \Lambda_n}\|f\ast g_{\lambda_n}\|^2_2\nonumber\\
&=\int_{\Rd}\sum_{\lambda_n\in \Lambda_n}|\widehat{g_{\lambda_n}}(\omega)|^2|\widehat{f}(\omega)|^2\mathrm d\omega\label{y1}\\
&\leq B_n\int_{\Rd\backslash B_{\delta}(0)}|\widehat{f}(\omega)|^2\mathrm d\omega\label{y2}\\
&\leq \underbrace{\max\{1,B_n\}}_{=C_{n}^{^n}} \int_{\Rd}\big|\widehat{f}(\omega)\big|^2\Big(1-\Big|\widehat{r_l}\Big(\frac{\omega}{\delta}\Big)\Big|^2\Big)\mathrm d\omega\label{y3}, 
\end{align}
for  $N\in \mathbb{N}$, where \eqref{y1} is by  Parseval's formula, \eqref{y2} is thanks to \eqref{eq:ass2} and \eqref{eq:tight}, and \eqref{y3} is due to $\text{supp}(\widehat{r_l})\subseteq B_1(0)$ and $0\leq \widehat{r_l}(\omega)\leq 1$, for $\omega\in \Rd$. The inductive step is established as follows.  Let $N>1$  and suppose that \eqref{main_h:up} holds for all paths $q$ of length $\ell(q)=N-1$, i.e., 
\begin{align}
&\sum_{q\,\in\, \Lambda_n \times\Lambda_{n+1}\times\cdots\times \Lambda_{n+N-2}} \|U[q]f\|_2^2\nonumber\\
&\leq 
C_{n}^{n+N-2} \int_{\Rd}\big|\widehat{f}(\omega)\big|^2\Big(1-\Big|\widehat{r_l}\Big(\frac{\omega}{(N-1)^{\alpha}\delta}\Big)\Big|^2\Big)\mathrm d\omega,\label{indii}
\end{align}
for  $n\in \mathbb{N}$. We start by noting that  every path $\tilde{q} \in \Lambda_n\times\Lambda_{n+1}\times\text{\mydots}\times\Lambda_{n+N-1}$ 
of length $\ell(\tilde{q})=N$,  with arbitrary starting index $n$, can be decomposed into a path $q \in \Lambda_{n+1}\times\text{\mydots}\times\Lambda_{n+N-1}$ of length $\ell(q)=N-1$ and an index $\lambda_{n} \in \Lambda_{n}$ according to $\tilde{q}=(\lambda_{n},q)$. Thanks to \eqref{aaaab} we have  $$U[\tilde{q}]=U[(\lambda_{n},q)]=U[q]U_n[\lambda_n],$$ which yields 
\begin{align}
&\sum_{q\,\in\, \Lambda_n \times\Lambda_{n+1}\times\cdots\times \Lambda_{n+N-1}} \|U[q]f\|_2^2\nonumber\\
&=\sum_{\lambda_n\in\Lambda_n} \sum_{q\,\in\, \Lambda_{n+1}\times \cdots\times \Lambda_{n+N-1}}\|U[q]\big(U_n[\lambda_n]f\big)\|_2^2\label{yf1455}, 
\end{align}
for  $n\in \mathbb{N}$. We proceed by examining the inner sum on the RHS of \eqref{yf1455}. Invoking the induction hypothesis \eqref{indii} with $n$ replaced by $(n+1)$ and employing Parseval's formula, we get 
\begin{align}
&\hspace{-0.2cm}\sum_{q\,\in\, \Lambda_{n+1}\times\cdots\times \Lambda_{n+N-1}}\|U[q]\big(U_n[\lambda_n]f\big)\|_2^2\nonumber\\
&\hspace{-0.2cm} \leq C_{n+1}^{n+N-1}\int_{\Rd}\big|\reallywidehat{U_n[\lambda_n]f}(\omega)\big|^2\Big(1-\Big|\widehat{r_l}\Big(\frac{\omega}{(N-1)^{\alpha}\delta}\Big)\Big|^2\Big)\mathrm d\omega\nonumber\\
&\hspace{-0.2cm}=C_{n+1}^{n+N-1}\big( \| U_n[\lambda_n]f\|_2^2-\|(U_n[\lambda_n]f)\ast r_{l,N-1,\alpha,\delta}\|_2^2\big)\nonumber\\
&\hspace{-0.2cm}=C_{n+1}^{n+N-1}\big( \| f\ast g_{\lambda_n}\|_2^2-\||f\ast g_{\lambda_n}|\ast r_{l,N-1,\alpha,\delta}\|_2^2\big)\label{eq:f898},
\end{align}
for $n\in \N$, where $r_{l,N-1,\alpha,\delta}$ is the inverse Fourier transform of $\widehat{r_l}\big( \frac{\omega}{(N-1)^{\alpha}\delta}\big)$.  Next, we note that $\widehat{r_l}\big( \frac{\omega}{(N-1)^{\alpha}\delta}\big)$ is a positive definite radial basis function  \cite[Theorem 6.20]{Wendland} and hence by \cite[Theorem 6.18]{Wendland} $r_{l,N-1,\alpha,\delta}(x)\geq 0$, for $x \in \Rd$. Furthermore, it follows from Lemma \ref{lemMallat}, stated below, that for   $\{\nu_{\lambda_n}\}_{\lambda_n \in \Lambda_n}\subseteq \Rd$, we have
\begin{align}
&\||f\ast g_{\lambda_n}|\ast r_{l,N-1,\alpha,\delta}\|_2^2 \nonumber\\&\geq \|f\ast g_{\lambda_n}\ast (M_{\nu_{\lambda_n}}r_{l,N-1,\alpha,\delta})\|_2^2.\label{yf6}
\end{align}
Here, we note that choosing the modulation factors  $\{\nu_{\lambda_n}\}_{\lambda_n \in \Lambda_n}\subseteq \Rd$ appropriately (see \eqref{modwei} below) will be  key in establishing  the inductive step.
\begin{lemma}\cite[Lemma 2.7]{MallatS}:\label{lemMallat}
Let $f,g\in L^2(\Rd)$ with $g(x)\geq 0$, for $x\in \Rd$. Then, $
\| |f|\ast g\|_2^2\geq \| f\ast (M_\omega g)\|_2^2, 
$ for $\omega \in \Rd.$
\end{lemma}
Inserting  \eqref{eq:f898} and \eqref{yf6} into the inner sum on the RHS of \eqref{yf1455} yields
\begin{align}
&\sum_{q\,\in\, \Lambda_n \times\Lambda_{n+1}\times\cdots\times \Lambda_{n+N-1}} \|U[q]f\|_2^2\nonumber\\
&\leq C_{n+1}^{n+N-1}\sum_{\lambda_n \in \Lambda_n}\Big( \| f\ast g_{\lambda_n}\|_2^2\nonumber\\
&-\|f\ast g_{\lambda_n}\ast (M_{\nu_{\lambda_n}}r_{l,N-1,\alpha,\delta})\|_2^2\Big)\nonumber\\
&= C_{n+1}^{n+N-1}\int_{\Rd}|\widehat{f}(\omega)|^2  h_{n,N,\alpha,\delta}(\omega)\mathrm d\omega,\quad\forall  N\in \mathbb{N},\label{finald1}
\end{align}
 where we applied Parseval's formula together with $\widehat{M_\omega f}=T_{\omega} \widehat{f}$, for $f\in L^2(\Rd)$, and $\omega\in\Rd$, and set 
 \begin{align}
&h_{n,N,\alpha,\delta}(\omega)\nonumber\\
&:=\sum_{\lambda_n \in \Lambda_n}|\widehat{g_{\lambda_n}}(\omega)|^2\Big(1-\Big|\widehat{r_l}\Big(\frac{\omega-\nu_{\lambda_n}}{(N-1)^{\alpha}\delta}\Big)\Big|^2\Big).\label{ekhfkjehfkrhsssskh}
\end{align}
 The key step is now to establish---by judiciously choosing $\{\nu_{\lambda_n}\}_{\lambda_n \in \Lambda_n}\subseteq \Rd$---the upper bound 
\begin{equation}\label{toshoww}
h_{n,N,\alpha,\delta}(\omega)\leq\max\{1,B_n\}\Big(1-\Big|\widehat{r_l}\Big(\frac{\omega}{N^{\alpha}\delta}\Big)\Big|^2\Big), 
\end{equation}
for $\omega \in \Rd,$ which upon noting that $C_n^{n+N-1}=\max\{ 1,B_n\}\,C_{n+1}^{n+N-1}$ yields \eqref{main_h:up} and thereby completes the proof. We start by defining $H_{A_{\lambda_n}}$,  for $\lambda_n\in \Lambda_n$, to be the orthant supporting $\widehat{g_{\lambda_n}}$, i.e., $\text{supp}(\widehat{g_{\lambda_n}}) \subseteq H_{A_{\lambda_n}}$, where $A_{\lambda_n}\in O(d)$, for $\lambda_n \in \Lambda_n$  (see Assumption \ref{ass}). Furthermore, for $\lambda_n \in \Lambda_n$, we choose the modulation factors according to
\begin{equation}\label{modwei}
\nu_{\lambda_n}:=A_{\lambda_n}\nu\in \Rd, 
\end{equation}
where the components  of $\nu \in \Rd$  are given by $\nu_k:=(1+2^{-1/2})\frac{\delta}{d}$, for $k\in \{1,\dots,d \}$. Invoking \eqref{eq:ass1} and \eqref{eq:ass2}, we get
\begin{align}
&\hspace{-0.25cm}h_{n,N,\alpha,\delta}(\omega)= \sum_{\lambda_n \in \Lambda_n}|\widehat{g_{\lambda_n}}(\omega)|^2\Big(1-\Big|\widehat{r_l}\Big(\frac{\omega-\nu_{\lambda_n}}{(N-1)^{\alpha}\delta}\Big)\Big|^2\Big)\nonumber\\
&\hspace{-0.3cm}=\sum_{\lambda_n \in \Lambda_n}|\widehat{g_{\lambda_n}}(\omega)|^2\mathds{1}_{S_{\lambda_n,\delta}}(\omega)\Big(1-\Big|\widehat{r_l}\Big(\frac{\omega-\nu_{\lambda_n}}{(N-1)^{\alpha}\delta}\Big)\Big|^2\Big),\label{ax1}
\end{align}
for $\omega \in \Rd,$ where $S_{\lambda_n,\delta}:=H_{A_{\lambda_n}}\hspace{-0.1cm}\backslash B_{\delta}(0)$. For the first canonical orthant $H=\{ x\in \Rd \ | \ x_k\geq 0,  \ k=1,$\mydots$,d\}$, we show in Lemma \ref{help} below that
\begin{equation}\label{help1}
\Big|\widehat{r_l}\Big(\frac{\omega-\nu}{(N-1)^{\alpha}\delta}\Big)\Big|\geq \Big|\widehat{r_l}\Big(\frac{\omega}{N^{\alpha}\delta}\Big)\Big|, 
\end{equation}
for  $\omega \in H\backslash B_{\delta}(0)$ and  $N\geq 2$. This will allow us to deduce 
\begin{equation}\label{help2}
\Big|\widehat{r_l}\Big(\frac{\omega-\nu_{\lambda_n}}{(N-1)^{\alpha}\delta}\Big)\Big|\geq\Big|\widehat{r_l}\Big(\frac{\omega}{N^{\alpha}\delta}\Big)\Big|,
\end{equation}
for $\omega \in S_{\lambda_n,\delta}$, $\lambda_{n}\in \Lambda_n$, and $N\geq 2$, where $S_{\lambda_n,\delta}=H_{A_{\lambda_n}}\hspace{-0.1cm}\backslash B_{\delta}(0)$, simply by noting that 
\begin{align}
&\Big|\widehat{r_l}\Big(\frac{\omega-\nu_{\lambda_n}}{(N-1)^{\alpha}\delta}\Big)\Big|=\bigg(1-\bigg| \frac{A_{\lambda_n}(\omega'-\nu)}{(N-1)^{\alpha}\delta}\bigg|\bigg)_{+}^l\nonumber\\
&=\Big(1-\Big| \frac{\omega'-\nu}{(N-1)^{\alpha}\delta}\Big|\Big)_{+}^l=\Big|\widehat{r_l}\Big(\frac{\omega'-\nu}{(N-1)^{\alpha}\delta}\Big)\Big|\label{ttt1}\\
&\geq\Big|\widehat{r_l}\Big(\frac{\omega'}{N^{\alpha}\delta}\Big)\Big|=\Big(1-\Big| \frac{\omega'}{N^{\alpha}\delta}\Big|\Big)_{+}^l\label{ttt3}\\
&=\bigg(1-\bigg| \frac{A_{\lambda_n}\omega'}{N^{\alpha}\delta}\bigg|\bigg)_{+}^l=\Big|\widehat{r_l}\Big(\frac{\omega}{N^{\alpha}\delta}\Big)\Big|,\label{ttt2}
\end{align}
 for $\omega=A_{\lambda_n}\omega'\in H_{A_{\lambda_n}}$\hspace{-0.1cm}$\backslash B_{\delta}(0)$, where $\omega'\in H\backslash B_{\delta}(0)$. Here, \eqref{ttt1} and 
 \eqref{ttt2} are thanks to $|\omega|=|A_{\lambda_n}\omega|$, which is by $A_{\lambda_n}\in O(d)$, and the inequality in \eqref{ttt3} is due to \eqref{help1}. Insertion of \eqref{help2} into \eqref{ax1} then yields
\begin{align}
h_{n,N,\alpha,\delta}(\omega)&\leq\sum_{\lambda_n \in \Lambda_n}|\widehat{g_{\lambda_n}}(\omega)|^2\mathds{1}_{S_{\lambda_n,\delta}}(\omega)\Big(1-\Big|\widehat{r_l}\Big(\frac{\omega}{N^{\alpha}\delta}\Big)\Big|^2\Big)\nonumber\\
&=\sum_{\lambda_n \in \Lambda_n}|\widehat{g_{\lambda_n}}(\omega)|^2\Big(1-\Big|\widehat{r_l}\Big(\frac{\omega}{N^{\alpha}\delta}\Big)\Big|^2\Big)\label{mfgsdjhf}\\
&\leq \max\{1,B_n \}\Big(1-\Big|\widehat{r_l}\Big(\frac{\omega}{N^{\alpha}\delta}\Big)\Big|^2\Big),\label{fgkkgkn}
\end{align}
for $\omega \in \Rd$, where in \eqref{mfgsdjhf} we employed  Assumption \ref{ass}, and \eqref{fgkkgkn} is thanks to \eqref{eq:tight}. This establishes \eqref{toshoww} and completes the proof of \eqref{main_h:up1} for  $\alpha=\log_2(\sqrt{d/(d-1/2)})$,   $d\geq1$.

It remains to show \eqref{help1}, which is accomplished through the following lemma.
\begin{lemma}\label{help}
Let $\alpha:=\log_2\big(\sqrt{{d/(d-1/2)}}\big)$, $\widehat{r_l}:\Rd\to \R$, $\widehat{r_l}(\omega):=(1-|\omega|)^{l}_{+}$, with $l>\lfloor d/2 \rfloor+1$, and define $\nu \in \mathbb{R}^d$ to have components $\nu_k=(1+2^{-1/2})\frac{\delta}{d}$, for $k\in \{1,\dots,d \}$. Then, 
\begin{equation}\label{now}
\Big|\widehat{r_l}\Big(\frac{\omega-\nu}{(N-1)^{\alpha}\delta}\Big)\Big|\geq \Big|\widehat{r_l}\Big(\frac{\omega}{N^{\alpha}\delta}\Big)\Big|, 
\end{equation}
for  $\omega \in H\backslash B_{\delta}(0)$ and $N\geq 2$.
\end{lemma}

\begin{proof}
The key idea of the proof is to employ a monotonicity argument. Specifically, thanks to $\widehat{r_l}$   monotonically decreasing in $|\omega|$, i.e., $\widehat{r_l}(\omega_1)\geq\widehat{r_l}(\omega_2)$, for $\omega_1,\omega_2\in \Rd$ with $ |\omega_2|\geq |\omega_1|$, \eqref{now} can be established simply by showing that  
\begin{equation}\label{toshow}
\kappa_N(\omega):=|\omega|^2\bigg|\frac{N-1}{N}\bigg|^{2\alpha}-\big|\omega-\nu|^2\geq 0,
\end{equation}
for   $\omega \in H\backslash B_{\delta}(0)$ and  $N\geq 2$. We first note that for $\omega\in H\backslash B_{\delta}(0)$ with $|\omega|> N^{\alpha}\delta$,  \eqref{now} is trivially satisfied as the RHS of \eqref{now} equals zero (owing to  $\big|\frac{\omega}{N^{\alpha}\delta}\big|> 1$   together with $\text{supp}(\widehat{r_l})\subseteq B_1(0)$). It hence suffices to prove \eqref{toshow} for $\omega \in H$ with $\delta\leq |\omega|\leq N^{\alpha}\delta$. To this end, fix $\tau \in [\delta,N^{\alpha}\delta]$, and define the spherical segment $$\Xi_{\tau}:=\{ \omega \in H \ | \ |\omega|=\tau\}.$$ We then have
\begin{align}
\kappa_N(\omega)&=\tau^2\bigg|\frac{N-1}{N}\bigg|^{2\alpha}-\big|\omega-\nu|^2\nonumber\\
&\geq\tau^2\bigg|\frac{N-1}{N}\bigg|^{2\alpha}-\big|\omega^\ast-\nu|^2\label{eq101},
\end{align}
for  $\omega \in \Xi_\tau$ and $N\geq 2$, where $$\omega^\ast=(\tau,0,\dots,0)\in \Xi_\tau.$$ The  inequality in \eqref{eq101} holds thanks to the mapping $$\omega\mapsto |\omega-\nu|^2, \quad \omega\in \Xi_\tau,$$  attaining its maxima along   the coordinate axes (see Figure \ref{fig:-1}). 
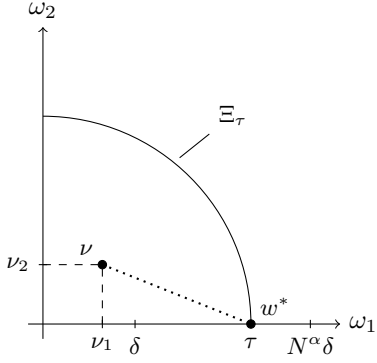
\begin{figure}
\begin{center}
\begin{tikzpicture}
	\begin{scope}[scale=.79]
	\draw[->] (-.25,0) -- (5,0) node[right] {$\omega_1$};
	\draw[->] (0,-.25) -- (0,5) node[above] {$\omega_2$};
	\draw (3.5cm,0cm) arc (0:90:3.5cm);
	\draw[dashed] (0,1) --  (1,1) -- (1,0);
		\draw (3.5 cm,2pt) -- (3.5 cm,-2pt) node[anchor=north] {\small$\tau$};
		\draw (1 cm,2pt) -- (1 cm,-2pt) node[anchor=north] {\small$\nu_1$};
		\draw (2pt,1cm) -- (-2pt,1cm) node[anchor=east] {\small$\nu_2$};		
		\draw (1.55 cm,2pt) -- (1.55 cm,-2pt) node[anchor=north] {\small$\delta$};
		\draw (4.5 cm,2pt) -- (4.5 cm,-2pt) node[anchor=north] {\small$N^\alpha \delta$};
		\draw (2.3cm,2.8cm) -- (2.8 cm,3.2cm) node[anchor=south west] {\small$\Xi_\tau$};
		\filldraw [black] (1,1) circle (2pt) node[anchor=south east] {\small$\nu$};
		\filldraw [black] (3.5,0) circle (2pt) node[anchor=south west] {\small$w^\ast$};
		\draw[thick,dotted] (1,1) -- (3.5,0);
	\end{scope}
\end{tikzpicture}
\end{center}
\vspace{0.3cm}
\caption{Illustration of \eqref{eq101} in dimension $d=2$. The mapping $\omega \mapsto |\omega-\nu|^2$, $\omega \in \Xi_\tau=\{ \omega=(\omega_1,\omega_2) \in \R^2 \ | \ |\omega|=\tau,\, \omega_1\geq 0,\, \omega_2 \geq 0\}$, computes the squared Euclidean distance between an element $\omega$ of the spherical segment $\Xi_\tau$ and the vector $\nu=(\nu_1,\nu_2)$ with components $\nu_k=(1+2^{-1/2})\frac{\delta}{2}$, $k\in \{ 1,2\}$. The mapping attains its maxima along  the coordinate axes, e.g., for $\omega^\ast=(\tau,0)\in \Xi_\tau$. }
\label{fig:-1}
\end{figure}
Inserting
\begin{align*}
 |\omega^\ast-\nu|^2&=\Big(\tau-\frac{\delta(1+2^{-1/2})}{d}\Big)^2+\frac{(d-1)\delta^2(1+2^{-1/2})^2}{d^2}\nonumber\\
&=\tau^2-\frac{\tau\delta(2+2^{1/2})}{d} + \frac{\delta^2(1+2^{-1/2})^2}{d}
\end{align*}
into \eqref{eq101} and rearranging terms yields 
\begin{align*}
&\kappa_N(\omega)\nonumber\\
&\geq\underbrace{\tau^2\bigg(\bigg|\frac{N-1}{N}\bigg|^{2\alpha}-1\bigg) +\frac{\tau\delta(2+2^{1/2})}{d} - \frac{\delta^2(1+2^{-1/2})^2}{d}}_{=:p_N(\tau)},
\end{align*}
for   $\omega \in \Xi_\tau$ and  $N\geq 2$. This inequality  shows that $\kappa_N(\omega)$ is lower-bounded---for $\omega \in \Xi_\tau$---by the $1$-D function $p_N(\tau)$. Now, $p_N(\tau)$ is quadratic in $\tau$, with the highest-degree coefficient $$\Big(\Big|\frac{N-1}{N}\Big|^{2\alpha}-1\Big)$$ negative owing to $$\alpha=\log_2\big(\sqrt{{d/(d-1/2)}}\big)>0, \quad\forall\, d \geq 1.$$
Therefore, thanks to $p_N$, $N\geq 2$, being concave, establishing $p_N(\delta)\geq0$ and $p_N(N^{\alpha}\delta)\geq0$, for $N\geq 2$, implies $p_N(\tau)\geq0$,  for $\tau\in [\delta,N^{\alpha}\delta]$ and $N\geq 2$ (see Figure \ref{fig:h0}), and thus  \eqref{toshow}, which completes the proof. It remains to show  that  $p_N(\delta)\geq0$ and $p_N(N^{\alpha}\delta)\geq0$, both for $N\geq 2$. We have
\begin{align}
p_N(\delta)&=\delta^2\bigg(\bigg|\frac{N-1}{N}\bigg|^{2\alpha}-1 +\frac{2+2^{1/2}}{d} - \frac{(1+2^{-1/2})^2}{d}\bigg)\nonumber\\
&\geq\delta^2\bigg(2^{-2\alpha} -\frac{d-1/2}{d}\bigg)=0\label{eq909},
\end{align}
where the inequality in \eqref{eq909} is by $$N\mapsto \Big|\frac{N-1}{N}\Big|^{2\alpha}, \quad N\geq 2,$$  monotonically increasing in $N$, and the equality  is thanks to $\alpha=\log_2\big(\sqrt{{d/(d-1/2)}}\big)$. Next, we have 
\begin{align}
&\frac{p_N(N^{\alpha}\delta)}{\delta^2}\nonumber\\
&=\big|N-1\big|^{2\alpha}-N^{2\alpha} +\frac{N^\alpha(2+2^{1/2})}{d} - \frac{(1+2^{-1/2})^2}{d}\nonumber\\
&\geq1-2^{2\alpha} +\frac{2^{\alpha}(2+2^{1/2})}{d} - \frac{(1+2^{-1/2})^2}{d}\label{eq907}\\
&=1-\frac{d}{d-1/2}+\frac{\sqrt{d}(2+2^{1/2})}{d\sqrt{d-1/2}} - \frac{(1+2^{-1/2})^2}{d}\geq 0\label{eq906}, 
\end{align}
for   $d\geq 1$ and  $N\geq 2$, where \eqref{eq907} is by $N\mapsto (N-1)^{2\alpha}-N^{2\alpha}+d^{-1}N^\alpha(2+2^{1/2})$, for $N\geq 2$,  monotonically increasing in $N$ (owing to $\alpha=\log_2\big(\sqrt{{d/(d-1/2)}}\big)>0$, for  $d\geq 1$), and the equality in \eqref{eq906} is thanks to $\alpha=\log_2\big(\sqrt{{d/(d-1/2)}}\big)$. The inequality in  \eqref{eq906} is established in Lemma \ref{helpaaa} below. This completes the proof.
\end{proof}
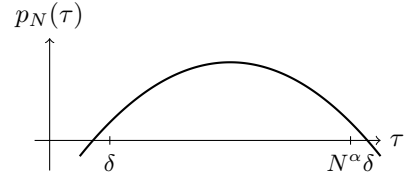
\begin{figure}
\begin{center}
\begin{tikzpicture}
	\begin{scope}[scale=.8]
	\draw[->] (-.25,0) -- (5.5,0) node[right] {$\tau$};
	\draw[->] (0,-.5) -- (0,1.7) node[above] {$p_N(\tau)$};

		\draw (1 cm,2pt) -- (1 cm,-2pt) node[anchor=north] {\small$\delta$};
		\draw (5 cm,2pt) -- (5 cm,-2pt) node[anchor=north] {\small$N^{\alpha}\delta$};
	\draw[thick,smooth, domain=.5:5.5,samples=100] plot (\x,{-(\x/2-1.5)*(\x/2-1.5)+1.3)});
	\end{scope}
\end{tikzpicture}
\end{center}
\vspace{0.2cm}
\caption{The function $p_N(\tau)$ is quadratic in $\tau$, with the coefficient of the highest-degree term  negative. Establishing $p_N(\delta)\geq 0$ and $p_N(N^{\alpha}\delta)\geq0$ therefore implies $p_N(\tau)\geq 0$,  $\tau \in [\delta,N^\alpha \delta]$. }
\label{fig:h0}
\end{figure}

\begin{lemma}\label{helpaaa}
For every $d\geq 1$ it holds that 
$$
1-\frac{d}{d-1/2}+\frac{\sqrt{d}(2+2^{1/2})}{d\sqrt{d-1/2}} - \frac{(1+2^{-1/2})^2}{d}\geq 0.
$$
\end{lemma}
\begin{proof}
We start by multiplying the inequality by $d(d-1/2)$, which (after rearranging terms) yields 
\begin{equation}\label{sfklhjkwehkweh}
\sqrt{d(d-1/2)}\alpha\geq (d-1/2)\beta+d/2,\quad d\geq 1,
\end{equation}
where 
$$\alpha:=(2+2^{1/2}),\hspace{1cm} \beta:=(1+2^{-1/2})^2.$$ Squaring \eqref{sfklhjkwehkweh} yields (again, after rearranging terms) 
\begin{equation*}
d^2\underbrace{\Big(\alpha^2-\beta^2-\beta-\frac{1}{4}\Big)}_{=\,0}+d\underbrace{\Big(-\frac{\alpha^2}{2}+\beta^2+\frac{\beta}{2}\Big)}_{\geq \,4}-\underbrace{\frac{\beta^2}{4}}_{\geq\, 3}\geq 0,
\end{equation*}
for $d\geq 1,$ which completes the proof.
\end{proof}

We proceed to sharpen, for $d=1$, the exponent $\alpha=\log_2(\sqrt{d/(d-1/2)})=1/2$ to $\alpha=1$. The structure of the corresponding  proof is similar to that of the proof  of the general result $\alpha=\log_2\big(\sqrt{{d/(d-1/2)}}\big)$, for $d\geq1$. Specifically, we start by employing  the arguments leading to \eqref{finald1} with $N^\alpha$ replaced by $N$. With this replacement $h_{n,N,\alpha,\delta}$ in \eqref{ekhfkjehfkrhsssskh} becomes
$$
\hspace{-0.02cm}h_{n,N,\alpha,\delta}(\omega):=\sum_{\lambda_n \in \Lambda_n}|\widehat{g_{\lambda_n}}(\omega)|^2\Big(1-\Big|\widehat{r_l}\Big(\frac{\omega-\nu_{\lambda_n}}{(N-1)\delta}\Big)\Big|^2\Big),
$$
where, again, appropriate choice of the modulation factors $\{\nu_{\lambda_n}\}_{\lambda_n \in \Lambda_n}\subseteq \Rd$ will be key in establishing the inductive step.  We start by defining $\Lambda_n^+$ to be the set of  indices $\lambda_n\in \Lambda_n$ such that $\text{supp}(\widehat{g_{\lambda_n}}) \subseteq [\delta,\infty)$, and take $\Lambda_n^-$ to be the set of  indices $\lambda_n\in \Lambda_n$ so that $\text{supp}(\widehat{g_{\lambda_n}}) \subseteq (-\infty,-\delta]$ (see Assumption \ref{ass}). Clearly, $\Lambda_n=\Lambda_n^+\cup\Lambda_n^-$. Moreover, we define the modulation factors according to $\nu_{\lambda_n}:=\delta$, for  $\lambda_n \in \Lambda_n^+$, and  $\nu_{\lambda_n}:=-\delta$, for  $\lambda_n \in \Lambda_n^-$. We then get
\begin{align}
&h_{n,N,\alpha,\delta}(\omega)= \sum_{\lambda_n \in \Lambda_n}|\widehat{g_{\lambda_n}}(\omega)|^2\Big(1-\Big|\widehat{r_l}\Big(\frac{\omega-\nu_{\lambda_n}}{(N-1)\delta}\Big)\Big|^2\Big)\nonumber\\
&=\sum_{\lambda_n \in \Lambda_n^+}|\widehat{g_{\lambda_n}}(\omega)|^2\,\mathds{1}_{[\delta,\,\infty)}(\omega)\Big(1-\Big|\widehat{r_l}\Big(\frac{\omega-\delta}{(N-1)\delta}\Big)\Big|^2\Big)\label{fffsssf1}\\
&+\sum_{\lambda_n \in \Lambda_n^-}|\widehat{g_{\lambda_n}}(\omega)|^2\,\mathds{1}_{(-\infty,-\delta\,]}(\omega)\Big(1-\Big|\widehat{r_l}\Big(\frac{\omega+\delta}{(N-1)\delta}\Big)\Big|^2\Big)\nonumber\\
&\leq\max\{1,B_n \}\,\mathds{1}_{[\delta,\,\infty)}(\omega)\Big(1-\Big|\widehat{r_l}\Big(\frac{\omega-\delta}{(N-1)\delta}\Big)\Big|^2\Big)\label{bbb1}\\
&+\max\{1,B_n \}\, \mathds{1}_{(-\infty,-\delta\,]}(\omega)\Big(1-\Big|\widehat{r_l}\Big(\frac{\omega+\delta}{(N-1)\delta}\Big)\Big|^2\Big)\label{bbb2},
\end{align}
where  \eqref{fffsssf1} is thanks to Assumption \ref{ass}, and for the last step we employed \eqref{eq:tight}. We show in Lemma \ref{helpi} below that
\begin{equation}\label{help12}
\Big|\widehat{r_l}\Big(\frac{\omega-\delta}{(N-1)\delta}\Big)\Big|\geq \Big|\widehat{r_l}\Big(\frac{\omega}{N\delta}\Big)\Big|, 
\end{equation}
for  $\omega \in [\delta,\infty)$ and  $N\geq 2$. This will allow us to deduce 
\begin{equation}\label{help22}
\Big|\widehat{r_l}\Big(\frac{\omega+\delta}{(N-1)\delta}\Big)\Big|\geq\Big|\widehat{r_l}\Big(\frac{\omega}{N\delta}\Big)\Big|,
\end{equation}
for  $\omega \in (-\infty,-\delta]$ and $N\geq 2$, simply by noting that 
\begin{align}
&\Big|\widehat{r_l}\Big(\frac{\omega+\delta}{(N-1)\delta}\Big)\Big|=\bigg(1-\bigg| \frac{\omega+\delta}{(N-1)\delta}\bigg|\bigg)_{+}^l\nonumber\\
&=\Big(1-\Big| \frac{-(-\omega-\delta)}{(N-1)\delta}\Big|\Big)_{+}^l=\Big|\widehat{r_l}\Big(\frac{-\omega-\delta}{(N-1)\delta}\Big)\Big|\nonumber\\
&\geq\Big|\widehat{r_l}\Big(\frac{-\omega}{N\delta}\Big)\Big|=\Big(1-\Big| \frac{-\omega}{N\delta}\Big|\Big)_{+}^l\label{ttt32}=\Big|\widehat{r_l}\Big(\frac{\omega}{N\delta}\Big)\Big|,
\end{align}
 for $\omega \in (-\infty,-\delta]$. Here, the inequality in \eqref{ttt32} is due to \eqref{help12}. Insertion of \eqref{help12} into \eqref{bbb1} and of  \eqref{help22} into \eqref{bbb2} then yields
\begin{align}
&h_{n,N,\alpha,\delta}(\omega)\nonumber\\
&\leq\max\{1,B_n \}\,\mathds{1}_{(-\infty,\delta\,]\,\cup\,[\delta,\,\infty)}(\omega)\Big(1-\Big|\widehat{r_l}\Big(\frac{\omega}{N\delta}\Big)\Big|^2\Big)\nonumber\\
&\leq \max\{1,B_n \}\Big(1-\Big|\widehat{r_l}\Big(\frac{\omega}{N\delta}\Big)\Big|^2\Big),\nonumber
\end{align}
for $\omega \in \R$, where the last inequality is thanks to $0\leq \widehat{r_l}(\omega)\leq 1$, for $\omega \in \R$. This establishes \eqref{main_h:up1}---in the $1$-D case---for $\alpha=1$ and completes the proof of statement i) in Theorem \ref{thm3}.

It remains to prove \eqref{help12}, which is done through the following lemma.
\begin{lemma}\label{helpi}
Let $\widehat{r_l}:\R\to \R$, $\widehat{r_l}(\omega):=(1-|\omega|)^{l}_{+}$, with $l>1$. Then, 
\begin{equation}\label{now2}
\Big|\widehat{r_l}\Big(\frac{\omega-\delta}{(N-1)\delta}\Big)\Big|\geq \Big|\widehat{r_l}\Big(\frac{\omega}{N\delta}\Big)\Big|,
\end{equation}
for   $\omega \in [\delta,\infty)$ and  $N\geq 2$.
\end{lemma}

\begin{proof}
We first note that for  $\omega> N\delta$,  \eqref{now2} is trivially satisfied as the RHS of \eqref{now2} equals zero (owing to  $\big|\frac{\omega}{N\delta}\big|> 1$   together with $\text{supp}(\widehat{r_l})\subseteq B_1(0)$). It hence suffices to prove \eqref{now2} for $\delta\leq \omega \leq N\delta$. The key idea of the proof is to employ a monotonicity argument. Specifically, thanks to $\widehat{r_l}$   monotonically decreasing in $|\omega|$, i.e., $\widehat{r_l}(\omega_1)\geq\widehat{r_l}(\omega_2)$, for $\omega_1,\omega_2\in \R$ with $ |\omega_2|\geq |\omega_1|$, \eqref{now2} can be established simply by showing that  
\begin{equation*}
\Big|\frac{\omega-\delta}{(N-1)\delta}\Big|\leq \Big|\frac{\omega}{N\delta}\Big| ,\hspace{0.75cm} \forall \, \omega \in [\delta,N\delta\,],\ \forall  N\geq 2,
\end{equation*}
which, by $\omega \in [\delta,N\delta]$, is equivalent to 
\begin{equation}\label{toshow2}
\frac{\omega-\delta}{(N-1)\delta}\leq \frac{\omega}{N\delta} ,\hspace{0.75cm} \forall \, \omega \in[\delta,N\delta\,],\ \forall  N\geq 2.
\end{equation}
Rearranging terms in \eqref{toshow2}, we get $\omega\leq N\delta$, for  $\omega \in[\delta,N\delta\,]$ and $N\geq 2$, which completes the proof.
\end{proof}

\begin{remark}
What makes the improved exponent $\alpha$ possible in the $1$-D case is the absence of rotated orthants. Specifically, for $d=1$, the filters $\{g_{\lambda_n}\}_{\lambda_n \in \Lambda_n}$ satisfy either  $\text{supp}(\widehat{g_{\lambda_n}})\subseteq (-\infty,-\delta]$ or $\text{supp}(\widehat{g_{\lambda_n}})\subseteq [\delta,\infty)$, i.e., the support sets $\text{supp}(\widehat{g_{\lambda_n}})$ are located in one of the two half-spaces. 
\end{remark}

\section{Proof of statement  ii) in Theorem \ref{thm3}}\label{Sob1Ap}
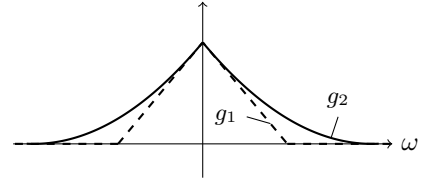
\begin{figure}
\begin{center}
\begin{tikzpicture}
	\begin{scope}[scale=.9]
	\draw[->] (-2.8,0) -- (2.8,0) node[right] {$\omega$};
	\draw[->] (0,-.5) -- (0,2.1) node[above] {};

			
			\draw (1cm,.3cm)--(.65 cm,.4cm)  node[anchor=east] {\small$g_1$};
			\draw (1.9 cm,.1cm) -- (2 cm,.4cm) node[anchor=south] {\small$g_2$};			
			
        \draw[thick,dashed] (1.25,0) -- (2.8,0);
        \draw[thick,dashed] (-1.25,0) -- (-2.8,0);
	\draw[thick,dashed, domain=-1.25:1.25,samples=100] plot (\x,{1.5*(1-2*1*abs(\x/2.5))});
	\draw[thick,smooth, domain=0:2.6,samples=100] plot (\x,{1.5*((1-abs(\x/2.5))^(2))});
	\draw[thick,smooth, domain=-2.6:0,samples=100] plot (\x,{1.5*((-1+abs(\x/2.5))^(2))});
	\end{scope}
\end{tikzpicture}
\end{center}
\caption{Illustration of \eqref{sdfknfldknfknff} in dimension $d=1$. The functions $g_1(\omega):=\max\{0, 1-2l|\omega|\}$ (dashed  line) and  $g_2(\omega):=(1-|\omega|)^{2l}_{+}$ (solid line) satisfy $g_1(\omega) \leq g_2(\omega)$, for $\omega\in \R$. Note that  $l>\lfloor d/2 \rfloor+1$.}
\label{fig:-so}
\end{figure}
We need to show that there exist constants $C_{1,s},C_{2,s}>0$ (that are independent of $N$)  such that
\begin{equation}\label{here1}
W_N(f)\leq C_{1,s}B_{\Omega}^N N^{-2s\alpha}, \quad \forall s\in (0, 1/2), \ \forall N\geq 1,
\end{equation}
and
\begin{equation}\label{here2}
W_N(f)\leq C_{2,s}B_{\Omega}^N N^{-\alpha}, \quad \forall s\in [1/2,\infty), \ \forall N\geq 1.
\end{equation}
Let us start by noting that 
\begin{equation}\label{sdfknfldknfknff}
\max\{0, 1-2l|\omega|\} \leq (1-|\omega|)^{2l}_{+},\quad \omega\in \Rd,
\end{equation}
where $l>\lfloor d/2 \rfloor+1$, see Figure \ref{fig:-so}. This implies
\begin{align}
1-\Big|\widehat{r_l}\Big(\frac{\omega}{N^{\alpha}\delta}\Big)\Big|^2&=1-\Big(1-\Big|\frac{\omega}{N^{\alpha}\delta}\Big|\Big)^{2l}_{+}\nonumber\\
&\leq 1- \max\bigg\{0,1- \frac{2l\,|\omega|}{N^{\alpha}\delta}\bigg\}\nonumber\\
&= 1 + \min\bigg\{0, \frac{2l\,|\omega|}{N^{\alpha}\delta} - 1\bigg\}\nonumber\\
&= \min\bigg\{1, \frac{2l\,|\omega|}{N^{\alpha}\delta} \bigg\}, \quad \forall \omega \in \Rd.\label{skfkhdjkwhejkhfhjf}
\end{align}
The key idea of the proof of  \eqref{here1} is to upper-bound the integral on the RHS of \eqref{main_h:up1} according to 
\begin{align}
&\int_{\Rd}|\widehat{f}(\omega)|^2 \Big(1-\Big|\widehat{r_l}\Big(\frac{\omega}{N^{\alpha}\delta}\Big)\Big|^2\Big)\mathrm d\omega\nonumber\\
&\leq \int_{\Rd}|\widehat{f}(\omega)|^2  \min\bigg\{1, \frac{2l\,|\omega|}{N^{\alpha}\delta} \bigg\}\mathrm d\omega\label{flksdjfdkljkfdsj111}\\
&= \int_{B_{\tau}(0)}\hspace{-0.25cm}|\widehat{f}(\omega)|^2 \frac{2l\,|\omega|}{N^{\alpha}\delta}\mathrm d\omega+ \int_{\Rd \backslash B_{\tau}(0)}\hspace{-0.5cm}|\widehat{f}(\omega)|^2  \mathrm d\omega\label{flksdjfdkljkfddddssj111},
\end{align}
where $\tau:=\frac{N^\alpha \delta}{2l}$. Here, the inequality in \eqref{flksdjfdkljkfdsj111} follows from \eqref{skfkhdjkwhejkhfhjf}, and \eqref{flksdjfdkljkfddddssj111} is owing to 
$$
\min\bigg\{1, \frac{2l\,|\omega|}{N^{\alpha}\delta} \bigg\} = \begin{cases} 
\frac{2l \, |\omega|}{N^\alpha \delta}, &|\omega| \leq \tau,\\
1, &|\omega| > \tau.
\end{cases}
$$
Now, the first integral in \eqref{flksdjfdkljkfddddssj111} satisfies
\begin{align}
&\int_{B_{\tau}(0)}\hspace{-0.25cm}|\widehat{f}(\omega)|^2 \frac{2l\,|\omega|}{N^{\alpha}\delta}\mathrm d\omega = \frac{2l}{N^{\alpha}\delta}\int_{B_{\tau}(0)}\hspace{-0.25cm}|\widehat{f}(\omega)|^2|\omega|^{1-2s}|\omega|^{2s}\mathrm d\omega\nonumber\\
&\leq \frac{2l \,\tau^{1-2s}}{N^{\alpha}\delta}\int_{B_{\tau}(0)}\hspace{-0.25cm}|\widehat{f}(\omega)|^2|\omega|^{2s}\mathrm d\omega\label{skjshdfkjjkhfdhfjsk12234}\\
&\leq \frac{2l \,\tau^{1-2s}}{N^{\alpha}\delta}\int_{B_{\tau}(0)}\hspace{-0.25cm}|\widehat{f}(\omega)|^2(1+|\omega|^2)^s\mathrm d\omega\nonumber\\
&  \leq \Big(\frac{2l}{N^\alpha \delta} \Big)^{2s} \int_{B_{\tau}(0)}|\widehat{f}(\omega)|^2(1+|\omega|^2)^s\,\mathrm d\omega \label{djkahdajkshjksahjshjsh},
\end{align}
where \eqref{skjshdfkjjkhfdhfjsk12234} is owing to $|\omega| \mapsto (1+|\omega|)^{1-2s}$ monotonically increasing in $|\omega|$ for $s\in (0,1/2)$. For the second integral in \eqref{flksdjfdkljkfddddssj111}, we have 
\begin{align}
&\int_{\Rd\backslash B_{\tau}(0)}\hspace{-0.25cm}|\widehat{f}(\omega)|^2 \mathrm d\omega = \int_{\Rd\backslash B_{\tau}(0)}\hspace{-0.25cm}|\widehat{f}(\omega)|^2|\omega|^{-2s}|\omega|^{2s}\mathrm d\omega\nonumber\\
&\leq\tau^{-2s} \int_{\Rd\backslash B_{\tau}(0)}\hspace{-0.25cm}|\widehat{f}(\omega)|^2 \underbrace{|\omega|^{2s}}_{\leq(1+|\omega|^2)^s}\mathrm d\omega\label{skffrkejshdfkjjkhfdhfjsk12234}\\
&\leq\tau^{-2s} \int_{\Rd\backslash B_{\tau}(0)}\hspace{-0.25cm}|\widehat{f}(\omega)|^2 (1+|\omega|^2)^s\mathrm d\omega\nonumber\\
&\leq \Big( \frac{2l}{N^\alpha \delta}\Big)^{2s} \int_{\Rd\backslash B_{\tau}(0)}|\widehat{f}(\omega)|^2(1+|\omega|^2)^s \,\mathrm d\omega\label{aslkjdslkjkdjkjl},
\end{align}
where \eqref{skffrkejshdfkjjkhfdhfjsk12234} is thanks to $$|\omega| \mapsto |\omega|^{-2s}, \quad \omega \in \Rd,$$ monotonically decreasing in $|\omega|$ for $s\in (0,1/2)$. Inserting \eqref{djkahdajkshjksahjshjsh} and \eqref{aslkjdslkjkdjkjl} into \eqref{flksdjfdkljkfddddssj111} establishes \eqref{here1} with $$C_{1,s}:= (2l)^{2s} \delta^{-2s}\| f\|_{H^s}^2.$$ Next, we show \eqref{here2} by noting that 
\begin{align}
&\int_{\Rd}|\widehat{f}(\omega)|^2 \Big(1-\Big|\widehat{r_l}\Big(\frac{\omega}{N^{\alpha}\delta}\Big)\Big|^2\Big)\mathrm d\omega\nonumber\\
&\leq \int_{\Rd}|\widehat{f}(\omega)|^2  \min\bigg\{1, \frac{2l\,|\omega|}{N^{\alpha}\delta} \bigg\}\mathrm d\omega\label{flksdjfdkljkfdsj13345511}\\
&\leq  \frac{2l}{N^{\alpha}\delta} \int_{\Rd}|\widehat{f}(\omega)|^2  |\omega| \,\mathrm d\omega \nonumber\\
&\leq  \frac{2l}{N^{\alpha}\delta} \int_{\Rd}|\widehat{f}(\omega)|^2  (1+|\omega|^2)^s \,\mathrm d\omega\nonumber = \frac{2l}{N^{\alpha}\delta} \| f\|_{H^s}^2 \nonumber,
\end{align}
where \eqref{flksdjfdkljkfdsj13345511} is by \eqref{skfkhdjkwhejkhfhjf}, and  the last inequality follows from $|\omega| \leq (1+|\omega|^2)^s$, for  $\omega \in \Rd$ and $s\in [1/2,\infty)$. This establishes \eqref{here2} with $$C_{2,s}:= (2l)  \,\delta^{-1}\| f\|_{H^s}^2$$ and thereby completes the proof.

\section{Proposition 1}\label{app:prop1}
\begin{proposition}\label{prop1}
Let $\Omega$ be the module-sequence \eqref{mods}. Then,
\begin{equation}\label{main_eq:decomp}
A_\Omega^N\|f\|_2^2 \leq\sum_{n=0}^{N-1}||| \Phi^n_\Omega(f)|||^2+W_N(f)\leq B_\Omega^N\|f\|_2^2,
\end{equation}
for all $f\in L^2(\Rd)$ and all $N\geq 1$, where $A_\Omega^N=\prod_{k=1}^N \min\{1,A_k\}$ and $B_\Omega^N=\prod_{k=1}^N \max\{1,B_k\}.
$
\end{proposition}

\begin{proof}
We proceed by induction over $N$ and  start with  the base case $N=1$ which follows directly from the frame property \eqref{PFP} according to 
\begin{align}
A_\Omega^1\| f\|_2^2&=\min\{1,A_1\}\|f\|_2^2\leq A_1\|f\|_2^2\nonumber\\
&\leq \underbrace{\| f\ast \chi_1\|_2^2+\sum_{\lambda_1\in \Lambda_1}\|f\ast g_{\lambda_1}\|_2^2}_{=\,||| \Phi^0_\Omega(f)|||^2+W_1(f)}\leq B_1\|f\|_2^2\nonumber\\
&\leq \max\{1,B_1\}\|f\|_2^2= B_\Omega^1\| f \|_2^2\nonumber,\hspace{0.75cm}\forall \, f\in L^2(\Rd). 
\end{align}
The inductive step is obtained as follows. Let $N>1$  and suppose that \eqref{main_eq:decomp} holds for $N-1$, i.e., 
\begin{align}
A_\Omega^{N-1}\|f\|_2^2 &\leq\sum_{n=0}^{N-2}||| \Phi^n_\Omega(f)|||^2+W_{N-1}(f)\nonumber\\
&\leq B_\Omega^{N-1}\|f\|_2^2, \quad \forall f\in L^2(\Rd).\label{indh2}
\end{align}
 We start by noting that
\begin{align}
&\sum_{n=0}^{N-1}||| \Phi^n_\Omega(f)|||^2+W_N(f)=\sum_{n=0}^{N-2}||| \Phi^n_\Omega(f)|||^2\nonumber\\
&+\sum_{q\,\in\, \Lambda^{N-1}}\|(U[q]f)\ast \chi_N \|_2^2+\sum_{q\,\in\, \Lambda^N}\|U[q]f\|_2^2\label{yf},
\end{align}
and proceed by examining the third term on the RHS of \eqref{yf}. Every path 
\begin{equation*}\label{eq:ee1}
\tilde{q} \in \Lambda^{N}=\underbrace{\Lambda_1\times\text{\mydots}\times\Lambda_{N-1}}_{=\Lambda^{N-1}}\times\Lambda_{N}
\end{equation*}
of length $N$ can  be decomposed into a path $q \in \Lambda^{N-1}$ of length $N-1$ and an index $\lambda_N \in \Lambda_N$ according to $\tilde{q}=(q,\lambda_N)$. Thanks to \eqref{aaaab} we have $U[\tilde{q}]=U[(q,\lambda_N)]=U_N[\lambda_N]U[q]$, which yields 
\begin{align}
&\sum_{q\,\in\, \Lambda^N} \|U[q]f\|_2^2=\sum_{q\,\in\, \Lambda^{N-1}}\sum_{\lambda_N\in\Lambda_N}\|(U[q]f)\ast g_{\lambda_N}\|_2^2\label{yf1}.
\end{align}
Substituting the third term on the RHS of \eqref{yf} by \eqref{yf1} and rearranging terms, we obtain
\begin{align}
&\sum_{n=0}^{N-1}||| \Phi^n_\Omega(f)|||^2+W_N(f)=\sum_{n=0}^{N-2}||| \Phi^n_\Omega(f)|||^2\nonumber\\
&+\sum_{q\,\in\, \Lambda^{N-1}}\Big(\underbrace{\|(U[q]f)\ast \chi_N \|_2^2+\sum_{\lambda_N\in \Lambda_N}\|(U[q]f)\ast g_{\lambda_N}\|_2^2}_{=:\rho_N(U[q]f)}\Big)\nonumber.
\end{align}
Thanks to the frame property \eqref{PFP} and $U[q]f \in L^2(\Rd)$, which is by \cite[Eq. 16]{Wiatowski_journal}, we have
$
A_N\|U[q]f\|_2^2\leq\rho_N(U[q]f)\leq B_N\| U[q]f\|_2^2,$
and thus
\begin{align}
&\min\{1,A_N\}\Big(\sum_{n=0}^{N-2}||| \Phi^n_\Omega(f)|||^2+W_{N-1}(f)\Big)\label{a1a1}\\
&\leq\sum_{n=0}^{N-1}||| \Phi^n_\Omega(f)|||^2+W_N(f)\nonumber\\
&\leq \max\{1,B_N\}\Big(\sum_{n=0}^{N-2}||| \Phi^n_\Omega(f)|||^2+W_{N-1}(f)\Big)\label{b1b1},
\end{align}
where we employed the identity $\sum_{q\,\in\, \Lambda^{N-1}}\| U[q]f\|_2^2=W_{N-1}(f)$. Invoking the induction hypothesis \eqref{indh2} in \eqref{a1a1} and \eqref{b1b1} and noting that $A_{\Omega}^{N}=\min\{1,A_N\}A_{\Omega}^{N-1}$ and $B_{\Omega}^{N}=\max\{1,B_N\}B_{\Omega}^{N-1}$ 
completes the proof. 
\end{proof}

\section{Proposition 2}\label{bum}
\begin{proposition}\label{prop2}
Let $\widehat{r_l}:\Rd\to \R$, $\widehat{r_l}(\omega):=(1-|\omega|)^{l}_{+}$, with $l>\lfloor d/2 \rfloor+1$, and let $\alpha:=\log_2(\sqrt{d/(d-1/2)})$. Then, we have 
\begin{equation}\label{elfknf}
\lim\limits_{N\to \infty} \int_{\Rd}\big|\widehat{f}(\omega)\big|^2\Big(1-\Big|\widehat{r_l}\Big(\frac{\omega}{N^{\alpha}\delta}\Big)\Big|^2\Big)\mathrm d\omega=0,
\end{equation}
for all  $f\in L^2(\Rd)$.
\end{proposition}
\begin{proof}
We start by setting
$$
d_{N,\alpha,\delta}(\omega):=\Big(1-\Big|\widehat{r_l}\Big(\frac{\omega}{N^{\alpha}\delta}\Big)\Big|^2\Big), \quad \omega \in \Rd, \, N\in \mathbb{N}.
$$
Let $f\in L^2(\Rd)$. For every $\varepsilon>0$ there exists $R>0$ such that 
$$
\int_{\Rd\backslash \overline{B_{R}(0)}} |\widehat{f}(\omega)|^2\mathrm d\omega\leq\varepsilon/2,
$$
where $\overline{B_{R}(0)}$ denotes the closed ball of radius $R$ centered at the origin. Next, we employ Dini's theorem  \cite[Theorem 7.3]{Fubini} to show that $(d_{N,\alpha,\delta})_{N\in \mathbb{N}}$ converges to zero uniformly on $\overline{B_{R}(0)}$. To this end, we note that (i)  $d_{N,\alpha,\delta}$ is continuous as a composition of continuous functions, (ii) $z_0(\omega)=0$, for $\omega\in \Rd$, is clearly continuous, (iii) $d_{N,\alpha,\delta}(\omega)\geq d_{N+1,\alpha,\delta}(\omega)$, for $\omega \in \Rd$ and $N\in \mathbb{N}$, and (iv) $d_{N,\alpha,\delta}$ converges to $z_0$  pointwise on  $\overline{B_{R}(0)}$, i.e., $$\lim\limits_{N\to \infty} d_{N,\alpha,\delta}(\omega)=z_0(\omega)=0, \quad \forall \omega \in \Rd.$$ This allows us to conclude that there exists $N_0\in \mathbb{N}$ (that depends on $\eps$) such that $d_{N,\alpha,\delta}(\omega)\leq \frac{\varepsilon}{2\| f\|_2^2},$ for $\omega \in \overline{B_{R}(0)}$ and  $N\geq N_0,$ and we therefore get
\begin{align*}
\int_{\Rd}\big|\widehat{f}(\omega)\big|^2d_{N,\alpha,\delta}(\omega)\mathrm d\omega=&\int_{\Rd\backslash \overline{B_R(0)}}\big|\widehat{f}(\omega)\big|^2\underbrace{d_{N,\alpha,\delta}(\omega)}_{\leq 1}\mathrm d\omega\\
+&\int_{ \overline{B_R(0)}}\big|\widehat{f}(\omega)\big|^2\underbrace{d_{N,\alpha,\delta}(\omega)}_{\leq \frac{\varepsilon}{2\| f\|_2^2}}\mathrm d\omega\\
\leq& \frac{\varepsilon}{2}+\frac{\varepsilon}{2\| f\|_2^2}\| \widehat{f}\|_2^2= \varepsilon,
\end{align*}
where in the last step we employed Parseval's formula. Since $\varepsilon>0$ was arbitrary, we have  \eqref{elfknf}, which completes the proof.
\end{proof}

\section{Proof of Theorem \ref{thm2}}\label{majabum}
We start by establishing \eqref{main_h:wave} in statement i). The structure of the proof is  similar to that of the proof of statement i) in Theorem \ref{thm3}, specifically we perform induction over $N$. Starting with the base case $N=1$, we first note that $\text{supp}(\widehat{\psi}\,)\subseteq[1/2,2],$  $\widehat{g_j}(\omega)=\widehat{\psi}(2^{-j}\omega),$ $j\geq 1,$ and $\widehat{g_j}(\omega)=\widehat{\psi}(-2^{-|j|}\omega),$ $j\leq -1,$ all by assumption, imply 
\begin{equation}\label{ee2}
 \text{supp}(\widehat{g_j})= \text{supp}(\widehat{\psi}(2^{-j}\cdot))\subseteq[2^{j-1},2^{j+1}], \quad j\geq 1,
\end{equation}
and 
\begin{equation}\label{ee3}
 \text{supp}(\widehat{g_j})=\text{supp}(\widehat{\psi}(-2^{-|j|}\cdot))\subseteq[-2^{|j|+1},-2^{|j|-1}],
\end{equation}
for $j\leq -1.$ We then get 
\begin{align}
W_1(f)&=\sum_{j\in \mathbb{Z}\backslash\{ 0\}}\|f\ast g_{j}\|^2_2\nonumber\\
&=\int_{\R}\sum_{j\in \mathbb{Z}\backslash\{ 0\}}|\widehat{g_{j}}(\omega)|^2|\widehat{f}(\omega)|^2\mathrm d\omega\label{ay1}\\
&=\int_{\R}\sum_{j\geq 1}|\widehat{\psi}(2^{-j}\omega)|^2|\widehat{f}(\omega)|^2\mathrm d\omega\nonumber\\&+\int_{\R}\sum_{j\leq -1}|\widehat{\psi}(-2^{-|j|}\omega)|^2|\widehat{f}(\omega)|^2\mathrm d\omega\nonumber\\
&=\int_{1}^\infty\sum_{j\geq 1}|\widehat{\psi}(2^{-j}\omega)|^2|\widehat{f}(\omega)|^2\mathrm d\omega\nonumber\\&+\int_{-\infty}^{-1}\sum_{j\leq -1}|\widehat{\psi}(-2^{-|j|}\omega)|^2|\widehat{f}(\omega)|^2\mathrm d\omega\label{my1}\\
&\leq \int_{\R\backslash [-1,1]}|\widehat{f}(\omega)|^2\mathrm d\omega\label{ay3aaa}\\
&\leq  \int_{\R}\big|\widehat{f}(\omega)\big|^2(1-|\widehat{r_l}(\omega)|^2)\mathrm d\omega\label{ay3},
\end{align}
where \eqref{ay1} is by Parseval's formula, and \eqref{my1} is thanks to \eqref{ee2} and \eqref{ee3}. The inequality in \eqref{ay3aaa} is owing to \eqref{eq:Waveletcondi}, and \eqref{ay3} is due to $\text{supp}(\widehat{r_l})\subseteq [-1,1]$ and $0\leq \widehat{r_l}(\omega)\leq 1$, for $\omega\in \R$.  The inductive step is obtained as follows.  Let $N>1$  and suppose that \eqref{main_h:wave} holds for $N-1$, i.e., 
\begin{equation}\label{ind_h:wave}
W_{N-1}(f)\leq \int_{\R}\big|\widehat{f}(\omega)\big|^2\Big(1-\Big|\widehat{r_l}\Big(\frac{\omega}{(5/3)^{N-2}}\Big)\Big|^2\Big)\mathrm d\omega,
\end{equation}
for   $f\in L^2(\R)$.  We start by noting that every path $\tilde{q} \in (\mathbb{Z}\backslash\{ 0\})^N$ of length $N$ can be decomposed into a path $q\,\in\,(\mathbb{Z}\backslash\{ 0\})^{N-1}$ of length $N-1$ and an index $j \in \mathbb{Z}\backslash\{ 0\}$ according to $\tilde{q}=(j, q)$. Thanks to \eqref{aaaab} we have  $U[\tilde{q}]=U[(j,q)]=U[q]U_1[j]$, which yields 
\begin{align}
W_N(f)&=\sum_{j\in\mathbb{Z}\backslash\{ 0\}}\,\sum_{q\,\in\, (\mathbb{Z}\backslash\{ 0\})^{N-1}}||U[q](U_1[j]f)||_2^2\nonumber\\
&=\sum_{j\in\mathbb{Z}\backslash\{ 0\}} W_{N-1}(U_1[j]f).\label{ayf1455}
\end{align}
We proceed by examining the term $W_{N-1}(U_1[j]f)$ inside the sum on the RHS of \eqref{ayf1455}. Invoking the induction hypothesis \eqref{ind_h:wave} and employing Parseval's formula, we get 
\begin{align}
 &W_{N-1}(U_1[j]f)\nonumber\\
 &\leq \int_{\R}\big|\widehat{U_1[j]f}(\omega)\big|^2\Big(1-\Big|\widehat{r_l}\Big(\frac{\omega}{(5/3)^{N-2}}\Big)\Big|^2\Big)\mathrm d\omega\nonumber\\
 &=\big( \| U_1[j]f\|_2^2-\|(U_1[j]f)\ast r_{l,N-2}\|_2^2\big)\nonumber\\
&=\big( \| f\ast g_{j}\|_2^2-\||f\ast g_{j}|\ast r_{l,N-2}\|_2^2\big),\label{aeq:f898}
\end{align}
where $r_{l,N-2}$ is the inverse Fourier transform of $\widehat{r_l}\big( \frac{\omega}{(5/3)^{N-2}}\big)$.  Next, we note that $\widehat{r_l}\big( \frac{\omega}{(5/3)^{N-2}}\big)$ is a positive definite radial basis function  \cite[Theorem 6.20]{Wendland} and hence by \cite[Theorem 6.18]{Wendland} $r_{l,N-2}(x)\geq 0$, for $x \in \R$. Furthermore, it follows from Lemma \ref{lemMallat} in Appendix \ref{app:prop2} that \begin{align}
&\||f\ast g_{j}|\ast r_{l,N-2}\|_2^2\geq \,\|f\ast g_{j}\ast (M_{\nu_{j}}r_{l,N-2})\|_2^2,\label{ayf6}
\end{align}
for  $\{\nu_{j}\}_{j \in \mathbb{Z}\backslash\{ 0\}}\subseteq \R$. Choosing the modulation factors  $\{\nu_{j}\}_{j \in \mathbb{Z}\backslash\{ 0\}}\subseteq \R$ appropriately (see \eqref{amodwei} below) will be  key in establishing the inductive step. 
Using  \eqref{aeq:f898} and \eqref{ayf6} to upper-bound the term $W_{N-1}(U_1[j]f)$ inside the sum on the RHS of \eqref{ayf1455} yields
\begin{align}
W_{N}(f)\leq& \sum_{j \in \Z\backslash\{0\}}\Big( \| f\ast g_{j}\|_2^2-\|f\ast g_{j}\ast (M_{\nu_{j}}r_{l,N-2})\|_2^2\Big)\nonumber\\
=&\ \int_{\R}|\widehat{f}(\omega)|^2  h_{l,N-2}(\omega)\mathrm d\omega,\label{aeq:f8998}
\end{align}
where 
\begin{align}
&h_{l,N-2}(\omega)\nonumber\\
&:=\sum_{j \in \mathbb{Z}\backslash\{0\}}|\widehat{g_{j}}(\omega)|^2\Big(1-\Big|\widehat{r_l}\Big(\frac{\omega-\nu_{j}}{(5/3)^{N-2}}\Big)\Big|^2\Big).\label{ekhfkjehfkrhkh}
\end{align}
In \eqref{aeq:f8998} we employed Parseval's formula together with $\widehat{M_\omega f}=T_{\omega} \widehat{f}$, for $f\in L^2(\R)$ and $\omega\in\R$. The key step is now to establish---by judiciously choosing $\{\nu_{j}\}_{j \in \Z\backslash\{ 0\}}\subseteq \R$---the upper bound 
\begin{equation}\label{atoshoww}
h_{l,N-2}(\omega)\leq\Big(1-\Big|\widehat{r_l}\Big(\frac{\omega}{(5/3)^{N-1}}\Big)\Big|^2\Big),\hspace{0.75cm} \forall \, \omega \in \R,
\end{equation}
which then yields \eqref{main_h:wave} and thereby completes the proof. To this end, we set $\eta:=\frac{4}{5}$, 
\begin{equation}\label{amodwei}
\nu_j:=2^{j}\eta, \quad j\geq 1, \hspace{.75cm} \nu_j:=-2^{|j|}\eta, \quad j\leq-1,
\end{equation}
and note that it suffices to prove \eqref{atoshoww} for $\omega\geq 0$, as
\begin{align}
&h_{l,N-2}(-\omega)=\sum_{j \in \mathbb{Z}\backslash\{0\}}|\widehat{g_{j}}(-\omega)|^2\Big(1-\Big|\widehat{r_l}\Big(\frac{-\omega-\nu_{j}}{(5/3)^{N-2}}\Big)\Big|^2\Big)\nonumber\\
&=\sum_{j\leq -1}|\widehat{g_{j}}(-\omega)|^2\Big(1-\Big|\widehat{r_l}\Big(\frac{\omega+\nu_{j}}{(5/3)^{N-2}}\Big)\Big|^2\Big)\label{i1}\\
&=\sum_{j\geq 1}|\widehat{g_{-j}}(-\omega)|^2\Big(1-\Big|\widehat{r_l}\Big(\frac{\omega+\nu_{-j}}{(5/3)^{N-2}}\Big)\Big|^2\Big)\nonumber\\
&=\sum_{j\geq 1}|\widehat{g_{j}}(\omega)|^2\Big(1-\Big|\widehat{r_l}\Big(\frac{\omega-\nu_{j}}{(5/3)^{N-2}}\Big)\Big|^2\Big)\label{i3}\\
&=h_{l,N-2}(\omega), \quad \forall \omega \geq0\label{ii3}.
\end{align}
Here, \eqref{i1} is thanks to $\widehat{g_j}(-\omega)=0$, for $j\geq1$ and $\omega\geq0$, which is by \eqref{ee2}, and \eqref{ii3} is owing to $\widehat{g_j}(\omega)=0$, for $j\leq-1$ and $\omega\geq0$, which is by \eqref{ee3}. Moreover, in \eqref{i1} we used  $\widehat{r_l}(-\omega)=\widehat{r_l}(\omega)$, for $\omega \in \R$, and \eqref{i3} is thanks to 
\begin{align*}
\widehat{g_{-j}}(-\omega)&=\widehat{\psi}(2^{-|-j|}\omega)=\widehat{\psi}(2^{-j}\omega)=\widehat{g_j}(\omega), 
\end{align*}
for $\omega \in \R$ and $j\geq 1$, as well as $$\nu_{-j}=-2^{j}\eta=-\nu_j, \quad j\geq 1.$$ Now, let $\omega \in [0,1]$, and note that 
\begin{align}\label{one}
h_{l,N-2}(\omega)&=\sum_{j \in \mathbb{Z}\backslash\{0\}}|\widehat{g_{j}}(\omega)|^2\Big(1-\Big|\widehat{r_l}\Big(\frac{\omega-\nu_{j}}{(5/3)^{N-2}}\Big)\Big|^2\Big)=0\nonumber\\
&\leq 1-\Big|\widehat{r_l}\Big(\frac{\omega}{(5/3)^{N-1}}\Big)\Big|^2, \quad \ \forall N\geq 2,
\end{align}
where the second equality in \eqref{one} is simply a consequence of $\widehat{g_j}(\omega)=0$, for $j\in \Z\backslash\{ 0\}$ and $\omega \in [0,1]$, which, in turn, is by  \eqref{ee2} and \eqref{ee3}. The inequality in \eqref{one}  is thanks to $0\leq \widehat{r_l}(\omega)\leq 1$, for $\omega\in \R$. Next, let $\omega \in [1,2]$. Then, we have
\begin{align}
&h_{l,N-2}(\omega)=|\widehat{g_1}(\omega)|^2\Big(1-\Big|\widehat{r_l}\Big(\frac{\omega-2\eta}{(5/3)^{N-2}}\Big)\Big|^2\Big)\label{flfkf}\\
&\leq |\widehat{g_1}(\omega)|^2\Big(1-\Big|\widehat{r_l}\Big(\frac{\omega-2\eta}{(5/3)^{N-2}}\Big)\Big|^2\Big)\nonumber\\
&+\underbrace{\big(1-|\widehat{g_1}(\omega)|^2\big)}_{\geq 0}\underbrace{\Big(1-\Big|\widehat{r_l}\Big(\frac{\omega-\eta}{(5/3)^{N-2}}\Big)\Big|^2\Big)}_{\geq 0}\label{slkls}\\
&= 1-\Big|\widehat{r_l}\Big(\frac{\omega-\eta}{(5/3)^{N-2}}\Big)\Big|^2+|\widehat{g_1}(\omega)|^2\Big(\Big|\widehat{r_l}\Big(\frac{\omega-\eta}{(5/3)^{N-2}}\Big)\Big|^2\nonumber\\
&-\Big|\widehat{r_l}\Big(\frac{\omega-2\eta}{(5/3)^{N-2}}\Big)\Big|^2\Big),\label{abs1}
\end{align}
where \eqref{flfkf} is thanks to $\widehat{g_j}(\omega)=0$, for $j\in \Z\backslash\{ 0,1\}$ and $\omega\in [1,2]$, which, in turn, is by \eqref{ee2} and \eqref{ee3}. Moreover, \eqref{slkls} is owing to $$|\widehat{g_1}(\omega)|^2\in [0,1],$$ which, in turn, is by \eqref{eq:Waveletcondi} and $0\leq\widehat{r_l}(\omega)\leq1$, for $\omega\in\R$. Next, fix $j\geq 2$ and  let $\omega \in [2^{j-1},2^{j}]$. Then, we have
\begin{align}
&h_{l,N-2}(\omega)=|\widehat{g_j}(\omega)|^2\Big(1-\Big|\widehat{r_l}\Big(\frac{\omega-2^{j}\eta}{(5/3)^{N-2}}\Big)\Big|^2\Big)\nonumber\\
&+\underbrace{|\widehat{g_{j-1}}(\omega)|^2}_{=(1-|\widehat{g_j}(\omega)|^2-|\widehat{\phi}(\omega)|^2)}\underbrace{\Big(1-\Big|\widehat{r_l}\Big(\frac{\omega-2^{j-1}\eta}{(5/3)^{N-2}}\Big)\Big|^2\Big)}_{\geq 0}\label{kerkjhferkh}\\
&\leq 1-\Big|\widehat{r_l}\Big(\frac{\omega-2^{j-1}\eta}{(5/3)^{N-2}}\Big)\Big|^2+|\widehat{g_j}(\omega)|^2\Big(\Big|\widehat{r_l}\Big(\frac{\omega-2^{j-1}\eta}{(5/3)^{N-2}}\Big)\Big|^2\nonumber\\
&-\Big|\widehat{r_l}\Big(\frac{\omega-2^{j}\eta}{(5/3)^{N-2}}\Big)\Big|^2\Big),\label{abs2}
\end{align}
where \eqref{kerkjhferkh} is thanks to i) $\widehat{g_{j'}}(\omega)=0$, for $j'\in \Z\backslash\{ 0,j,j-1\}$ and $\omega\in [2^{j-1},2^j]$, which, in turn, is by \eqref{ee2} and \eqref{ee3}, and ii) 
\begin{equation}\label{lsnkln}
|\widehat{\phi}(\omega)|^2+|\widehat{g_{j-1}}(\omega)|^2+|\widehat{g_j}(\omega)|^2=1,
\end{equation}
for $\omega\in [2^{j-1},2^{j}],$ which is a consequence of the Littlewood-Paley condition \eqref{eq:Waveletcondi} and of  \eqref{ee2} and \eqref{ee3}. It follows from \eqref{abs1} and \eqref{abs2} that for every $j\geq1$, we have
\begin{align}
&h_{l,N-2}(\omega)\leq 1-\Big|\widehat{r_l}\Big(\frac{\omega-2^{j-1}\eta}{(5/3)^{N-2}}\Big)\Big|^2\nonumber\\
&+|\widehat{g_j}(\omega)|^2\underbrace{\Big(\Big|\widehat{r_l}\Big(\frac{\omega-2^{j-1}\eta}{(5/3)^{N-2}}\Big)\Big|^2-\Big|\widehat{r_l}\Big(\frac{\omega-2^{j}\eta}{(5/3)^{N-2}}\Big)\Big|^2\Big)}_{=:s(\omega)},\nonumber
\end{align}
for $\omega\in [2^{j-1},2^j]$. Next, we divide the interval $[2^{j-1},2^{j}]$ into two intervals, namely $I_L:=[2^{j-1},\frac{3}{5}2^{j}]$ and $I_R:=[\frac{3}{5}2^{j},2^j]$, and note that $s(\omega)\geq0$, for $\omega \in I_L$, and $s(\omega)\leq0$, for $\omega \in I_R$, as $\widehat{r_l}$ is monotonically decreasing in $|\omega|$ and $|\omega-2^j\eta|\geq |\omega-2^{j-1}\eta|$, for  $\omega \in I_L$, and $|\omega-2^j\eta|\leq |\omega-2^{j-1}\eta|$, for   $\omega \in I_R$,
respectively (see Figure \ref{fig:h1}).
\begin{figure}
\begin{center}
\begin{tikzpicture}
	\begin{scope}[scale=1]
	\draw[->] (-.25,0) -- (5.5,0) node[right] {$\omega$};
	\draw[->] (0,-.5) -- (0,2.1) node[above] {};

		\draw (1 cm,2pt) -- (1 cm,-2pt) node[anchor=north] {\small$2^{j-1}$};
		\draw (1.8 cm,2pt) -- (1.8 cm,-2pt) node[anchor=north] {\small$\frac{3}{5}2^{j}$};
		\draw (5 cm,2pt) -- (5 cm,-2pt) node[anchor=north] {\small$2^{j}$};
			
			\draw (1.6 cm,1.1cm) -- (1.5 cm,1.4cm) node[anchor=south] {\small$h_3$};
			\draw (3.15 cm,1.2cm) -- (3.25 cm,.9cm) node[anchor=north] {\small$h_2$};			
			\draw (4.75 cm,.55cm)--(4.65 cm,.85cm)  node[anchor=south] {\small$h_1$};

        \draw (1,0) -- (5,0);
	\draw[thick,dashed, domain=.5:5.5,samples=100] plot (\x,{.4*abs(\x-.2)});
	\draw[thick,smooth, domain=.5:5.5,samples=100] plot (\x,{.4*abs(\x-3.4)});
	\draw[thick,dotted, domain=.5:5.5,samples=100] plot (\x,{.4*(1.8+0.6*\x)});
	\end{scope}
\end{tikzpicture}
\vspace{0.2cm}
\end{center}
\caption{The functions $h_1(\omega):=|\omega-2^j\eta|$ (solid line), $h_2(\omega):=|\omega-2^{j-1}\eta|$ (dashed line), and $h_3(\omega):=\frac{3}{5}\omega$ (dotted line) satisfy $h_2\leq h_1\leq h_3$ on $I_L=[2^{j-1},\frac{3}{5}2^j]$ and $h_1\leq h_2\leq h_3$ on $I_R=[\frac{3}{5}2^j,2^j]$.}
\label{fig:h1}
\end{figure}
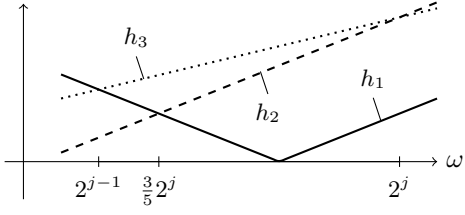
 For  $\omega\in I_L$, we  therefore have
\begin{align*}
h_{l,N-2}(\omega)&\leq1-\Big|\widehat{r_l}\Big(\frac{\omega-2^{j-1}\eta}{(5/3)^{N-2}}\Big)\Big|^2\nonumber+\underbrace{|\widehat{g_j}(\omega)|^2}_{\in \,[0,1]}\underbrace{s(\omega)}_{\geq 0}\\
&\leq 1-\Big|\widehat{r_l}\Big(\frac{\omega-2^{j-1}\eta}{(5/3)^{N-2}}\Big)\Big|^2\nonumber+s(\omega)\\
&= 1-\Big|\widehat{r_l}\Big(\frac{\omega-2^{j}\eta}{(5/3)^{N-2}}\Big)\Big|^2 \leq1-\Big|\widehat{r_l}\Big(\frac{\omega}{(5/3)^{N-1}}\Big)\Big|^2,
\end{align*}
where $|\widehat{g_j}(\omega)|^2\in [0,1]$ follows from \eqref{lsnkln}, and the last inequality is a consequence of $
|\omega-2^j\eta|\leq \frac{3\omega}{5}$, for  $\omega \in I_L$, see Figure \ref{fig:h1}. For $\omega\in I_R$, we have
\begin{align*}
&h_{l,N-2}(\omega)\leq1-\Big|\widehat{r_l}\Big(\frac{\omega-2^{j-1}\eta}{(5/3)^{N-2}}\Big)\Big|^2\nonumber+\underbrace{|\widehat{g_j}(\omega)|^2}_{\in \,[0,1]}\underbrace{s(\omega)}_{\leq 0}\\
&\leq 1-\Big|\widehat{r_l}\Big(\frac{\omega-2^{j-1}\eta}{(5/3)^{N-2}}\Big)\Big|^2\nonumber \leq 1-\Big|\widehat{r_l}\Big(\frac{\omega}{(5/3)^{N-1}}\Big)\Big|^2,
\end{align*}
where the last inequality now follows from 
$
|\omega-2^{j-1}\eta|\leq \frac{3\omega}{5}$, for  $\omega \in I_R$, see Figure \ref{fig:h1}. This completes the proof of  \eqref{main_h:wave}.

Next, we establish \eqref{main_sow1}. The proof is very similar to that of statement  ii) in Theorem \ref{thm2} in Section \ref{Sob1Ap}. We start by noting that \eqref{main_sow1} amounts to the existence of constants $C_{1,s}, C_{2,s}>0$ (that are independent of $N$) such that
\begin{equation}\label{here1a}
W_N(f)\leq  C_{1,s}(5/3)^{-2sN}, \quad \forall s\in (0, 1/2), \ \forall N\geq 1,
\end{equation}
and
\begin{equation}\label{here2a}
W_N(f)\leq  C_{2,s} (5/3)^{-N}, \quad \forall s\in [1/2,\infty), \ \forall N\geq 1.
\end{equation}
The key idea of the proof of  \eqref{here1a} is to upper-bound the integral on the RHS of \eqref{main_h:wave} according to 
\begin{align}
&\int_{\R}|\widehat{f}(\omega)|^2 \Big(1-\Big|\widehat{r_l}\Big(\frac{\omega}{(5/3)^{N-1}}\Big)\Big|^2\Big)\mathrm d\omega\nonumber\\
&\leq \int_{\R}|\widehat{f}(\omega)|^2  \min\bigg\{1, \frac{2l\,|\omega|}{(5/3)^{N-1}} \bigg\}\mathrm d\omega\label{flksdjfdkljkfdsj111a}\\
&= \int_{B_{\tau}(0)}\hspace{-0.25cm}|\widehat{f}(\omega)|^2 \frac{2l\,|\omega|}{(5/3)^{N-1}}\mathrm d\omega+ \int_{\R \backslash B_{\tau}(0)}\hspace{-0.5cm}|\widehat{f}(\omega)|^2  \mathrm d\omega\label{flksdjfdkljkfddddssj111a},
\end{align}
where $\tau:=\frac{(5/3)^{N-1} }{2l}$. Here, the inequality in \eqref{flksdjfdkljkfdsj111a} follows from \eqref{skfkhdjkwhejkhfhjf}, and \eqref{flksdjfdkljkfddddssj111a} is owing to 
$$
\min\bigg\{1, \frac{2l\,|\omega|}{(5/3)^{N-1}} \bigg\} = \begin{cases} 
\frac{2l \, |\omega|}{(5/3)^{N-1}}, &|\omega| \leq \tau,\\
1, &|\omega| > \tau.
\end{cases}
$$
Now, the first integral in \eqref{flksdjfdkljkfddddssj111a} satisfies
\begin{align}
&\int_{B_{\tau}(0)}\hspace{-0.25cm}|\widehat{f}(\omega)|^2 \frac{2l\,|\omega|}{(5/3)^{N-1}}\mathrm d\omega \nonumber\\
&= \frac{2l}{(5/3)^{N-1}}\int_{B_{\tau}(0)}\hspace{-0.25cm}|\widehat{f}(\omega)|^2|\omega|^{1-2s}|\omega|^{2s}\mathrm d\omega\nonumber\\
&\leq \frac{2l \,\tau^{1-2s}}{(5/3)^{N-1}}\int_{B_{\tau}(0)}\hspace{-0.25cm}|\widehat{f}(\omega)|^2\underbrace{|\omega|^{2s}}_{\leq (1+|\omega|^2)^s}\mathrm d\omega\label{skjshdfkjjkhfdhfjsk12234a}\\
&  \leq \Big(\frac{2l}{(5/3)^{N-1} } \Big)^{2s} \int_{B_{\tau}(0)}|\widehat{f}(\omega)|^2(1+|\omega|^2)^s\,\mathrm d\omega \label{djkahdajkshjksahjshjsha},
\end{align}
where \eqref{skjshdfkjjkhfdhfjsk12234a} is owing to $|\omega| \mapsto (1+|\omega|)^{1-2s}$ monotonically increasing in $|\omega|$ for $s\in (0,1/2)$. For the second integral in \eqref{flksdjfdkljkfddddssj111a}, we have 
\begin{align}
&\int_{\R\backslash B_{\tau}(0)}\hspace{-0.25cm}|\widehat{f}(\omega)|^2 \mathrm d\omega = \int_{\R\backslash B_{\tau}(0)}\hspace{-0.25cm}|\widehat{f}(\omega)|^2|\omega|^{-2s}|\omega|^{2s}\mathrm d\omega\nonumber\\
&\leq\tau^{-2s} \int_{\R\backslash B_{\tau}(0)}\hspace{-0.25cm}|\widehat{f}(\omega)|^2 \underbrace{|\omega|^{2s}}_{\leq (1+|\omega|^2)^s}\mathrm d\omega\label{skffrkejshdfkjjkhfdhfjsk12234a}\\
&\leq \Big( \frac{2l}{(5/3)^{N-1} }\Big)^{2s} \int_{\R\backslash B_{\tau}(0)}|\widehat{f}(\omega)|^2(1+|\omega|^2)^s \,\mathrm d\omega\label{aslkjdslkjkdjkjla},
\end{align}
where \eqref{skffrkejshdfkjjkhfdhfjsk12234a} is thanks to $|\omega| \mapsto |\omega|^{-2s}$ monotonically decreasing in $|\omega|$ for $s\in (0,1/2)$. Inserting \eqref{djkahdajkshjksahjshjsha} and \eqref{aslkjdslkjkdjkjla} into \eqref{flksdjfdkljkfddddssj111a} establishes \eqref{here1a} with $$C_{1,s}:= (2l)^{2s}(5/3)^{2s} \| f\|_{H^s}^2.$$ Next, we show \eqref{here2a} by noting that 
\begin{align}
&\int_{\R}|\widehat{f}(\omega)|^2 \Big(1-\Big|\widehat{r_l}\Big(\frac{\omega}{(5/3)^{N-1}}\Big)\Big|^2\Big)\mathrm d\omega\nonumber\\
&\leq \int_{\R}|\widehat{f}(\omega)|^2  \min\bigg\{1, \frac{2l\,|\omega|}{(5/3)^{N-1}} \bigg\}\mathrm d\omega\label{flksdjfdkljkfdsj13345511a}\\
&\leq  \frac{2l}{(5/3)^{N-1}} \int_{\R}|\widehat{f}(\omega)|^2  |\omega| \,\mathrm d\omega \nonumber\\
&\leq  \frac{2l}{(5/3)^{N-1}} \int_{\R}|\widehat{f}(\omega)|^2  (1+|\omega|^2)^s \,\mathrm d\omega\nonumber = \frac{2l}{(5/3)^{N-1}} \| f\|_{H^s}^2 \nonumber,
\end{align}
where \eqref{flksdjfdkljkfdsj13345511a} is by \eqref{skfkhdjkwhejkhfhjf}, and  the last inequality follows from $$|\omega| \leq (1+|\omega|^2)^s, \quad \forall \,\omega \in \R, \ \forall \, s\in [1/2,\infty).$$ This establishes \eqref{here2a} with $$C_{2,s}:= 2l\,(5/3) \| f\|_{H^s}^2$$ and thereby completes the proof of statement i).
We proceed to the proof of statement ii), again, effected by induction over $N$\hspace{-0.05cm}. Specifically, we first establish \eqref{main_g:wave}  by employing the same arguments as those leading to \eqref{aeq:f8998} with $(5/3)^{N-2}$ replaced by $(3/2)^{N-2}R$. With this replacement $h_{l,N-2}$ in \eqref{ekhfkjehfkrhkh} becomes
\begin{align}
&h_{l,N-2}(\omega)\nonumber\\
&:=\sum_{k \in \mathbb{Z}\backslash\{0\}}|\widehat{g_{k}}(\omega)|^2\Big(1-\Big|\widehat{r_l}\Big(\frac{\omega-\nu_{k}}{(3/2)^{N-2}R}\Big)\Big|^2\Big),\label{ajgkejrgkjhr}
\end{align}
where, again, judicious choice of the modulation factors $\{ \nu_k\}_{k\in \mathbb{Z}\backslash\{ 0\}}\subseteq \R$ (see \eqref{lregjhekjh} below) will be key in establishing the inductive step. Here, we note that the functions $\widehat{g_k}$ in \eqref{ajgkejrgkjhr} satisfy  $\widehat{g_k}(\omega)=\widehat{g}(\omega-R(k+1)),$ for $k\geq 1$, $\widehat{g_k}(\omega)=\widehat{g}(\omega+R(|k|+1)),$ for $k\leq -1$,  both by assumption, as well as   
\begin{align}
\text{supp}(\widehat{g_k})&=\text{supp}(\widehat{g}(\cdot-R(k+1)))\nonumber\\
&\subseteq[Rk,R(k+2)], \quad k\geq1,\label{ee4}
\end{align}
and 
\begin{align}
 \text{supp}(\widehat{g_k})&= \text{supp}(\widehat{g}(\cdot+R(|k|+1)))\nonumber\\
 &\subseteq[-R(|k|+2),-R|k|],\quad k\leq-1,\label{ee5}
\end{align}
where \eqref{ee4} and \eqref{ee5} follow from $$\text{supp}(\widehat{g})\subseteq[-R,R],$$ which again is by assumption. It remains to establish the equivalent of \eqref{atoshoww}, namely 
\begin{equation}\label{qqtoshoww4}
h_{l,N-2}(\omega)\leq1-\Big|\widehat{r_l}\Big(\frac{\omega}{(3/2)^{N-1}R}\Big)\Big|^2, \quad\forall \omega\in \R.
\end{equation}
To this end, we set $\eta:=\frac{2}{3}R$, 
\begin{equation}\label{lregjhekjh}
\nu_k:=Rk+\eta, \quad \forall k\geq 1, \hspace{0.3cm}  \nu_k:=-\nu_{|k|}, \quad \forall k\leq-1,
\end{equation}
and note that it suffices to establish \eqref{qqtoshoww4} for $\omega\geq 0$, thanks to 
\begin{align}
&h_{l,N-2}(-\omega)=\sum_{k \in \mathbb{Z}\backslash\{0\}}|\widehat{g_{k}}(-\omega)|^2\Big(1-\Big|\widehat{r_l}\Big(\frac{-\omega-\nu_{k}}{(3/2)^{N-2}R}\Big)\Big|^2\Big)\nonumber\\
&=\sum_{k\leq -1}|\widehat{g_{k}}(-\omega)|^2\Big(1-\Big|\widehat{r_l}\Big(\frac{\omega+\nu_{k}}{(3/2)^{N-2}R}\Big)\Big|^2\Big)\label{qi1}\\
&=\sum_{k\geq 1}|\widehat{g_{-k}}(-\omega)|^2\Big(1-\Big|\widehat{r_l}\Big(\frac{\omega+\nu_{-k}}{(3/2)^{N-2}R}\Big)\Big|^2\Big)\nonumber\\
&=\sum_{k\geq 1}|\widehat{g_{k}}(\omega)|^2\Big(1-\Big|\widehat{r_l}\Big(\frac{\omega-\nu_{k}}{(3/2)^{N-2}R}\Big)\Big|^2\Big)\label{qi3}\\
&=h_{l,N-2}(\omega), \quad \forall \omega \geq0\label{qqi3}.
\end{align}
Here, \eqref{qi1} follows from $\widehat{g_k}(-\omega)=0$, for $k\geq1$ and $\omega\geq0$, which, in turn, is by \eqref{ee4}, and \eqref{qqi3} is owing to $\widehat{g_k}(\omega)=0$, for $k\leq-1$ and $\omega\geq0$, which is by \eqref{ee5}. Moreover, in \eqref{qi1} we used $\widehat{r_l}(-\omega)=\widehat{r_l}(\omega)$, for $\omega \in \R$, and \eqref{qi3} is thanks to  $\nu_{-k}=-\nu_k$, for $k\geq 1$, and 
\begin{align*}
\widehat{g_{-k}}(-\omega)&=\widehat{g}(-\omega+R(|\!-\!k|+1))=\widehat{g}(-(\omega-R(k+1)))\nonumber\\
&=\widehat{g}(\omega-R(k+1))=\widehat{g_k}(\omega), \quad\forall \, \omega \in \R, \ \forall \, k\geq 1,
\end{align*}
where we used $\widehat{g}(-\omega)=\widehat{g}(\omega)$, for $\omega \in \R$, which is by assumption. Now, let $\omega \in [0,R]$, and note that 
\begin{equation}\label{qone}
h_{l,N-2}(\omega)=0\leq 1-\Big|\widehat{r_l}\Big(\frac{\omega}{(3/2)^{N-1}R}\Big)\Big|^2, 
\end{equation}
for $N\geq 2,$ where the equality in \eqref{qone} is  a consequence of \eqref{ee4} and \eqref{ee5}, and the inequality  is thanks to $0\leq \widehat{r_l}(\omega)\leq 1$, for $\omega\in \R$. Next, let $\omega \in [R,2R]$. Then, we have
\begin{align}
&h_{l,N-2}(\omega)=|\widehat{g_1}(\omega)|^2\Big(1-\Big|\widehat{r_l}\Big(\frac{\omega-\nu_1}{(3/2)^{N-2}R}\Big)\Big|^2\Big)\label{qhwehjheh}\\
&\leq|\widehat{g_1}(\omega)|^2\Big(1-\Big|\widehat{r_l}\Big(\frac{\omega-\nu_1}{(3/2)^{N-2}R}\Big)\Big|^2\Big)\nonumber\\
&+\underbrace{(1-|\widehat{g_{1}}(\omega)|^2)}_{\geq 0}\underbrace{\Big(1-\Big|\widehat{r_l}\Big(\frac{\omega-\eta}{(3/2)^{N-2}R}\Big)\Big|^2\Big)}_{\geq 0}\label{aqkerkjhferkh}\\
&= 1-\Big|\widehat{r_l}\Big(\frac{\omega-\eta}{(3/2)^{N-2}R}\Big)\Big|^2+|\widehat{g_1}(\omega)|^2\Big(\Big|\widehat{r_l}\Big(\frac{\omega-\eta}{(3/2)^{N-2}R}\Big)\Big|^2\nonumber\\
&-\Big|\widehat{r_l}\Big(\frac{\omega-\nu_1}{(3/2)^{N-2}R}\Big)\Big|^2\Big),\label{bibi1}
\end{align}
where \eqref{qhwehjheh} is thanks to $\widehat{g_k}(\omega)=0$, for $k\in \Z\backslash\{ 0,1\}$ and $\omega\in [0,R]$, which, in turn, is by \eqref{ee4} and \eqref{ee5}. Moreover, \eqref{aqkerkjhferkh} is owing to $|\widehat{g_1}(\omega)|^2\in [0,1]$, which, in turn, is by \eqref{eq:Gaborcondi}, and $0\leq \widehat{r_l}(\omega)\leq 1$, for $\omega \in \R$. Next, fix $k\geq 2$, and  let $\omega \in [Rk,R(k+1)]$. Then, we have\vspace{0.09cm}
\begin{align}
&h_{l,N-2}(\omega)=|\widehat{g_k}(\omega)|^2\Big(1-\Big|\widehat{r_l}\Big(\frac{\omega-\nu_k}{(3/2)^{N-2}R}\Big)\Big|^2\Big)\nonumber\\
&+\underbrace{|\widehat{g_{k-1}}(\omega)|^2}_{=(1-|\widehat{g_k}(\omega)|^2-|\widehat{\phi}(\omega)|^2)}\underbrace{\Big(1-\Big|\widehat{r_l}\Big(\frac{\omega-\nu_{k-1}}{(3/2)^{N-2}R}\Big)\Big|^2\Big)}_{\geq 0}\label{qkerkjhferkh}\\
&\leq 1-\Big|\widehat{r_l}\Big(\frac{\omega-\nu_{k-1}}{(3/2)^{N-2}R}\Big)\Big|^2+|\widehat{g_k}(\omega)|^2\Big(\Big|\widehat{r_l}\Big(\frac{\omega-\nu_{k-1}}{(3/2)^{N-2}R}\Big)\Big|^2\nonumber\\
&-\Big|\widehat{r_l}\Big(\frac{\omega-\nu_k}{(3/2)^{N-2}R}\Big)\Big|^2\Big),\label{bibi2}
\end{align}
where \eqref{qkerkjhferkh} is thanks to i) $\widehat{g_{k'}}(\omega)=0$, for $k'\in \Z\backslash\{ 0,k,k-1\}$ and $\omega\in [Rk,R(k+1)]$, which, in turn, is by \eqref{ee4} and \eqref{ee5}, and ii) 
\begin{equation}\label{qlsnkln}
|\widehat{\chi}(\omega)|^2+|\widehat{g_{k-1}}(\omega)|^2+|\widehat{g_k}(\omega)|^2=1,
\end{equation}
for $\omega\in [Rk,R(k+1)],$ which is a consequence of the Littlewood-Paley condition \eqref{eq:Gaborcondi} and of \eqref{ee4} and \eqref{ee5}. It follows from \eqref{bibi1} and \eqref{bibi2} that for $k\geq 1$, 
\begin{align}
&h_{l,N-2}(\omega)\leq 1-\Big|\widehat{r_l}\Big(\frac{\omega-\nu_{k-1}}{(3/2)^{N-2}R}\Big)\Big|^2\nonumber\\
&+|\widehat{g_k}(\omega)|^2\underbrace{\Big(\Big|\widehat{r_l}\Big(\frac{\omega-\nu_{k-1}}{(3/2)^{N-2}R}\Big)\Big|^2-\Big|\widehat{r_l}\Big(\frac{\omega-\nu_k}{(3/2)^{N-2}R}\Big)\Big|^2\Big)}_{=:s(\omega)},\nonumber
\end{align}
for $\omega \in [Rk,R(k+1)]$, where $\nu_0:=\eta$. Next, we divide the interval $[Rk,R(k+1)]$ into two intervals, namely $$I_L:=[Rk,R(k+1/6)]$$ and $$I_R:=[R(k+1/6),R(k+1)],$$ and note that $s(\omega)\geq0$, for $\omega \in I_L$, and $s(\omega)\leq0$, for $\omega \in I_R$, as $\widehat{r_l}$ is monotonically decreasing in $|\omega|$ and $|\omega-\nu_{k}|\geq |\omega-\nu_{k-1}|$, for $\omega \in I_L$, and $|\omega-\nu_k|\leq |\omega-\nu_{k-1}|$, for  $\omega \in I_R$,
respectively (see Figure \ref{fig:h2}). 
\begin{figure}
\begin{center}
\begin{tikzpicture}
	\begin{scope}[scale=1]
	\draw[->] (-.25,0) -- (5.5,0) node[right] {$\omega$};
	\draw[->] (0,-.5) -- (0,2.1) node[above] {};

		\draw (1 cm,2pt) -- (1 cm,-2pt) node[anchor=north] {\small$Rk$};
		\draw (2 cm,2pt) -- (2 cm,-.5pt) node[anchor=north] { \ \ \ \ \small$R(k+\frac{1}{6})$};
		\draw (5 cm,2pt) -- (5 cm,-2pt) node[anchor=north] {\small$R(k+1)$};
		

			\draw (1.6 cm,1.175cm) -- (1.5 cm,1.475cm) node[anchor=south] {\small$h_3$};
			\draw (3.15 cm,1.125cm) -- (3.25 cm,.825cm) node[anchor=north] {\small$h_2$};			
			\draw (4.75 cm,.45cm)--(4.65 cm,.75cm)  node[anchor=south] {\small$h_1$};

	\draw[thick,dashed, domain=.5:5.5,samples=100] plot (\x,{.4*abs(\x-.33)});
	\draw[thick,smooth, domain=.5:5.5,samples=100] plot (\x,{.4*abs(\x-3.67)});
	\draw[thick,dotted, domain=.5:5.5,samples=100] plot (\x,{.4*(2.16+0.5*\x)});
	\end{scope}
\end{tikzpicture}
\vspace{0.2cm}
\end{center}
\caption{The functions $h_1(\omega):=|\omega-\nu_k|$ (solid line), $h_2(\omega):=|\omega-\nu_{k-1}|$ (dashed line), and $h_3(\omega)=\frac{2}{3}\omega$ (dotted line) satisfy 
$h_2\leq h_1\leq h_3$ on $I_L=[Rk,R(k+\frac{1}{6})]$ and $h_1\leq h_2\leq h_3$ on $I_R=[R(k+\frac{1}{6}),R(k+1)]$.}
\label{fig:h2}
\end{figure}
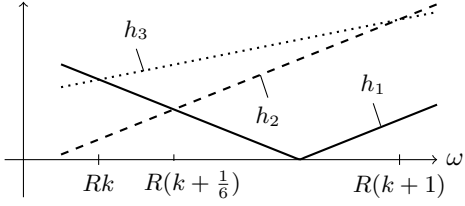
For  $\omega\in I_L$, we therefore have
\begin{align*}
&h_{l,N-2}(\omega)\leq1-\Big|\widehat{r_l}\Big(\frac{\omega-\nu_{k-1}}{(3/2)^{N-2}R}\Big)\Big|^2\nonumber+\underbrace{|\widehat{g_k}(\omega)|^2}_{\in \,[0,1]}\underbrace{s(\omega)}_{\geq 0}\nonumber\\
&\leq 1-\Big|\widehat{r_l}\Big(\frac{\omega-\nu_{k-1}}{(3/2)^{N-2}R}\Big)\Big|^2\nonumber+s(\omega)\\
&=1-\Big|\widehat{r_l}\Big(\frac{\omega-\nu_{k}}{(3/2)^{N-2}R}\Big)\Big|^2 \leq1-\Big|\widehat{r_l}\Big(\frac{\omega}{(3/2)^{N-1}R}\Big)\Big|^2,
\end{align*}
where $|\widehat{g_k}(\omega)|^2\in [0,1]$ follows  from \eqref{qlsnkln}, and the last inequality is by 
$
|\omega-\nu_k|\leq \frac{2\omega}{3}$, for  $\omega \in I_L$ (see Figure \ref{fig:h2}). For the interval $\omega\in I_R$, we have
\begin{align*}
&h_{l,N-2}(\omega)\leq1-\Big|\widehat{r_l}\Big(\frac{\omega-\nu_{k-1}}{(3/2)^{N-2}R}\Big)\Big|^2\nonumber+\underbrace{|\widehat{g_k}(\omega)|^2}_{\in \,[0,1]}\underbrace{s(\omega)}_{\leq 0}\\
&\leq 1-\Big|\widehat{r_l}\Big(\frac{\omega-\nu_{k-1}}{(3/2)^{N-2}R}\Big)\Big|^2\nonumber \leq 1-\Big|\widehat{r_l}\Big(\frac{\omega}{(3/2)^{N-1}R}\Big)\Big|^2,
\end{align*}
where the last inequality is by 
$
|\omega-\nu_{k-1}|\leq \frac{2\omega}{3}$, for  $\omega \in I_R$ (see Figure \ref{fig:h2}). This completes the proof of \eqref{main_g:wave}.

Next, we establish \eqref{main_sog1}. The proof is very similar to that of statement  ii) in Theorem \ref{thm3} in Appendix \ref{Sob1Ap}. We start by noting that \eqref{main_sog1} amounts to the existence of constants $C_{1,s},C_{2,s}>0$ (that are independent of $N$) such that
\begin{equation}\label{here1ab}
W_N(f)\leq C_{1,s} (3/2)^{-2sN}, \quad \forall s\in (0, 1/2), \ \forall N\geq 1,
\end{equation}
and
\begin{equation}\label{here2ab}
W_N(f)\leq C_{2,s} (3/2)^{-N}, \quad \forall s\in [1/2,\infty), \ \forall N\geq 1.
\end{equation}
The key idea of the proof of  \eqref{here1ab} is to upper-bound the integral on the RHS of \eqref{main_g:wave} according to 
\begin{align}
&\int_{\R}|\widehat{f}(\omega)|^2 \Big(1-\Big|\widehat{r_l}\Big(\frac{\omega}{(3/2)^{N-1}R}\Big)\Big|^2\Big)\mathrm d\omega\nonumber\\
&\leq \int_{\R}|\widehat{f}(\omega)|^2  \min\bigg\{1, \frac{2l\,|\omega|}{(3/2)^{N-1}R} \bigg\}\mathrm d\omega\label{flksdjfdkljkfdsj111ab}\\
&= \int_{B_{\tau}(0)}\hspace{-0.25cm}|\widehat{f}(\omega)|^2 \frac{2l\,|\omega|}{(3/2)^{N-1}R}\mathrm d\omega+ \int_{\R \backslash B_{\tau}(0)}\hspace{-0.5cm}|\widehat{f}(\omega)|^2  \mathrm d\omega\label{flksdjfdkljkfddddssj111ab},
\end{align}
where $$\tau:=\frac{(3/2)^{N-1}R }{2l}.$$ Here, the inequality in \eqref{flksdjfdkljkfdsj111ab} follows from \eqref{skfkhdjkwhejkhfhjf}, and \eqref{flksdjfdkljkfddddssj111ab} is owing to 
$$
\min\bigg\{1, \frac{2l\,|\omega|}{(3/2)^{N-1}R} \bigg\} = \begin{cases} 
\frac{2l \, |\omega|}{(3/2)^{N-1}R}, &|\omega| \leq \tau,\\
1, &|\omega| > \tau.
\end{cases}
$$
Now, the first integral in \eqref{flksdjfdkljkfddddssj111ab} satisfies
\begin{align}
&\int_{B_{\tau}(0)}\hspace{-0.25cm}|\widehat{f}(\omega)|^2 \frac{2l\,|\omega|}{(3/2)^{N-1}R}\mathrm d\omega \nonumber\\
&= \frac{2l}{(3/2)^{N-1}R}\int_{B_{\tau}(0)}\hspace{-0.25cm}|\widehat{f}(\omega)|^2|\omega|^{1-2s}|\omega|^{2s}\mathrm d\omega\nonumber\\
&\leq \frac{2l \,\tau^{1-2s}}{(3/2)^{N-1}R}\int_{B_{\tau}(0)}\hspace{-0.25cm}|\widehat{f}(\omega)|^2\underbrace{|\omega|^{2s}}_{\leq (1+|\omega|^2)^s}\mathrm d\omega\label{skjshdfkjjkhfdhfjsk12234ab}\\
&  \leq \Big(\frac{2l}{(3/2)^{N-1}R } \Big)^{2s} \int_{B_{\tau}(0)}|\widehat{f}(\omega)|^2(1+|\omega|^2)^s\,\mathrm d\omega \label{djkahdajkshjksahjshjshab},
\end{align}
where \eqref{skjshdfkjjkhfdhfjsk12234ab} is owing to $$|\omega| \mapsto (1+|\omega|)^{1-2s}$$ monotonically increasing in $|\omega|$ for $s\in (0,1/2)$. For the second integral in \eqref{flksdjfdkljkfddddssj111ab}, we have 
\begin{align}
&\int_{\R\backslash B_{\tau}(0)}\hspace{-0.25cm}|\widehat{f}(\omega)|^2 \mathrm d\omega = \int_{\R\backslash B_{\tau}(0)}\hspace{-0.25cm}|\widehat{f}(\omega)|^2|\omega|^{-2s}|\omega|^{2s}\mathrm d\omega\nonumber\\
&\leq\tau^{-2s} \int_{\R\backslash B_{\tau}(0)}\hspace{-0.25cm}|\widehat{f}(\omega)|^2 \underbrace{|\omega|^{2s}}_{\leq (1+|\omega|^2)^s}\mathrm d\omega\label{skffrkejshdfkjjkhfdhfjsk12234ab}\\
&\leq \Big( \frac{2l}{(3/2)^{N-1}R }\Big)^{2s} \int_{\R\backslash B_\tau(0)}|\widehat{f}(\omega)|^2(1+|\omega|^2)^s \,\mathrm d\omega\label{aslkjdslkjkdjkjlab},
\end{align}
where \eqref{skffrkejshdfkjjkhfdhfjsk12234ab} is thanks to $$|\omega| \mapsto |\omega|^{-2s}, \quad \omega \in \Rd,$$ monotonically decreasing in $|\omega|$ for $s\in (0,1/2)$. Inserting \eqref{djkahdajkshjksahjshjshab} and \eqref{aslkjdslkjkdjkjlab} into \eqref{flksdjfdkljkfddddssj111ab} establishes \eqref{here1ab} with $$C_{1,s}:= (2l)^{2s}(3/2)^{2s} R^{-2s} \| f\|_{H^s}^2.$$ Next, we show \eqref{here2ab} by noting that 
\begin{align}
&\int_{\R}|\widehat{f}(\omega)|^2 \Big(1-\Big|\widehat{r_l}\Big(\frac{\omega}{(3/2)^{N-1}R}\Big)\Big|^2\Big)\mathrm d\omega\nonumber\\
&\leq \int_{\R}|\widehat{f}(\omega)|^2  \min\bigg\{1, \frac{2l\,|\omega|}{(3/2)^{N-1}R} \bigg\}\mathrm d\omega\label{flksdjfdkljkfdsj13345511ab}\\
&\leq  \frac{2l}{(3/2)^{N-1}R} \int_{\R}|\widehat{f}(\omega)|^2  |\omega| \,\mathrm d\omega \nonumber\\
&\leq  \frac{2l}{(3/2)^{N-1}R} \int_{\R}|\widehat{f}(\omega)|^2  (1+|\omega|^2)^s \,\mathrm d\omega\nonumber\\
& = \frac{2l}{(3/2)^{N-1}R} \| f\|_{H^s}^2 \nonumber,
\end{align}
where \eqref{flksdjfdkljkfdsj13345511ab} is by \eqref{skfkhdjkwhejkhfhjf}, and  the last inequality follows from $|\omega| \leq (1+|\omega|^2)^s$, for  $\omega \in \R$ and $s\in [1/2,\infty)$. This establishes \eqref{here2ab} with $C_{2,s}:= 2l\,(3/2) R^{-1} \| f\|_{H^s}^2$ and thereby completes the proof of statement ii).

\section{Proof of Corollary \ref{cor1}}\label{aca}
We start with statement i) and note that $A_\Omega^N=B_\Omega^N=1$,  $N\in \mathbb{N}$, by assumption. Let $f\in L^2(\Rd)$ with $$\text{supp}(\widehat{f})\subseteq B_L(0).$$ Then, by Proposition \ref{prop1} in Appendix \ref{app:prop1} together with $$\lim\limits_{N\to \infty}W_N(f)=0,\quad \forall f\in L^2(\Rd),$$  which follows from Proposition \ref{prop2} in Appendix \ref{bum}, we have
\begin{align}\vspace{-0.5cm}
\| f \|_2^2&=|||\Phi_\Omega(f)|||^2=\sum_{n=0}^\infty||| \Phi_\Omega^n(f)|||^2\nonumber\\
&\geq\sum_{n=0}^N||| \Phi_\Omega^n(f)|||^2=\| f\|^2_2-W_{N+1}(f)\label{o2}  \\
&\geq \int_{\Rd}|\widehat{f}(\omega)|^2\Big|\widehat{r_l}\Big(\frac{\omega}{(N+1)^\alpha\delta}\Big)\Big|^2\mathrm d\omega\label{o3}\\
&=\int_{B_L(0)}|\widehat{f}(\omega)|^2\Big|\widehat{r_l}\Big(\frac{\omega}{(N+1)^\alpha\delta}\Big)\Big|^2\mathrm d\omega\label{o4},
\end{align}
 where \eqref{o2} is by the lower bound in \eqref{main_eq:decomp}, \eqref{o3} is thanks to Parseval's formula and \eqref{main_h:up1}, and \eqref{o4} follows from $f$ being $L$-band-limited. Next, thanks to $\widehat{r_l}$  monotonically decreasing in $|\omega|$, we get
\begin{equation}\label{o5}
\Big|\widehat{r_l}\Big(\frac{\omega}{(N+1)^\alpha\delta}\Big)\Big|^2\geq \Big|\widehat{r_l}\Big(\frac{L}{(N+1)^\alpha\delta}\Big)\Big|^2,
\end{equation}
for  $\omega \in B_L(0)$. Employing \eqref{o5} in \eqref{o4}, we obtain
\begin{align}
\| f\|_2^2&\geq \Big|\widehat{r_l}\Big(\frac{L}{(N+1)^\alpha\delta}\Big)\Big|^2 \| f\|_2^2\label{o6}\\
&=\Big( 1-\frac{L}{(N+1)^\alpha\delta}\Big)_{+}^{2l}\| f\|_2^2\nonumber\\
&=\Big( 1-\frac{L}{(N+1)^\alpha\delta}\Big)^{2l}\| f\|_2^2\geq (1-\varepsilon)\| f\|_2^2\label{o7},
\end{align}
where in \eqref{o6} we used Parseval's formula, the equality in \eqref{o7} is due to $L\leq (N+1)^\alpha\delta$, which, in turn, is by \eqref{hihih}, and the inequality in \eqref{o7} is also by \eqref{hihih} (upon rearranging terms). This establishes \eqref{bubu_main} and thereby completes the proof. 

The proof of statement ii) is very similar to that of statement i). Again, we start by noting that $A_\Omega^N=B_\Omega^N=1$,  $N\in \mathbb{N}$, by assumption. Let $f\in L^2(\R)$ with $\text{supp}(\widehat{f})\subseteq B_L(0)$. Then, by Proposition \ref{prop1} in Appendix \ref{app:prop1} together with $\lim\limits_{N\to \infty}W_N(f)=0$, for $f\in L^2(\R)$, we have
\begin{align}
\| f \|_2^2&=|||\Phi_\Omega(f)|||^2=\sum_{n=0}^\infty||| \Phi_\Omega^n(f)|||^2\nonumber\\
&\geq\sum_{n=0}^N||| \Phi_\Omega^n(f)|||^2=\| f\|^2_2-W_{N+1}(f)\label{ao2}  \\
&\geq \int_{\R}|\widehat{f}(\omega)|^2\Big|\widehat{r_l}\Big(\frac{\omega}{a^N\delta}\Big)\Big|^2\mathrm d\omega\label{ao3}\\
&=\int_{B_L(0)}|\widehat{f}(\omega)|^2\Big|\widehat{r_l}\Big(\frac{\omega}{a^N\delta}\Big)\Big|^2\mathrm d\omega\label{ao4},
\end{align}
 where \eqref{ao2} is by the lower bound in \eqref{main_eq:decomp}, \eqref{ao3} is thanks to Parseval's formula and \eqref{main_h:wave} and \eqref{main_g:wave}, and \eqref{ao4} follows from $f$ being $L$-band-limited. Next, thanks to $\widehat{r_l}$  monotonically decreasing in $|\omega|$, we get
\begin{equation}\label{ao5}
\Big|\widehat{r_l}\Big(\frac{\omega}{a^N\delta}\Big)\Big|^2\geq \Big|\widehat{r_l}\Big(\frac{L}{a^N\delta}\Big)\Big|^2,\hspace{0.75cm}\forall \, \omega \in B_L(0). 
\end{equation}
Employing \eqref{ao5} in \eqref{ao4} yields
\begin{align}
\| f\|_2^2&\geq \Big|\widehat{r_l}\Big(\frac{L}{a^N\delta}\Big)\Big|^2 \| f\|_2^2\label{ao6}=\Big( 1-\frac{L}{a^N\delta}\Big)_{+}^{2l}\| f\|_2^2\\
&=\Big( 1-\frac{L}{a^N\delta}\Big)^{2l}\| f\|_2^2\geq (1-\varepsilon)\| f\|_2^2\label{ao7},
\end{align}
where in \eqref{ao6} we used Parseval's formula, the equality in \eqref{ao7} is by $L\leq a^N\delta$, which, in turn, is by \eqref{hihih2}, and the inequality in \eqref{ao7} is also due to \eqref{hihih2} (upon rearranging terms). This establishes \eqref{bubu_main} and thereby completes the proof of ii).

\section{Proof of Corollary \ref{cor3}}\label{acas}
The proof is very similar to that of Corollary \ref{cor1} in Appendix \ref{aca}. We start with statement i). Let $f\in H^s(\Rd)\backslash\{ 0\}$ and $\eps \in (0,1)$ and note that, by \eqref{here1} and \eqref{here2} together with $B_\Omega^N=1$,  $N\in \mathbb{N}$, which is by assumption, we have 
\begin{equation}\label{sflhdshfdshfshfdsjk2}
W_N(f)\leq \frac{(2l)^{\gamma}\|f\|^2_{H^s}}{\delta^\gamma N^{\gamma \alpha}}, \quad \forall s >0,
\end{equation}
where $\gamma = \min\{1,2s\}$. By Proposition \ref{prop1} in Appendix \ref{app:prop1} with $A_\Omega^N = B_\Omega^N =1$, $N \in \mathbb{N}$,  and $\lim\limits_{N\to \infty}W_N(f)=0$, $f\in L^2(\Rd)$, which follows from Proposition \ref{prop2} in Appendix \ref{bum}, we have \vspace{0.1cm}
\begin{align}
\| f \|_2^2&=|||\Phi_\Omega(f)|||^2=\sum_{n=0}^\infty||| \Phi_\Omega^n(f)|||^2\nonumber\\
&\geq\sum_{n=0}^N||| \Phi_\Omega^n(f)|||^2=\| f\|^2_2-W_{N+1}(f)\label{wo2}  \\
&\geq \|f\|_2^2 - \frac{(2l)^{\gamma}\|f\|^2_{H^s}}{\delta^\gamma (N+1)^{\gamma \alpha}}\label{wo3}\\
&\geq\| f\|_2^2 - \varepsilon \| f\|_2^2 = (1-\varepsilon)\| f\|_2^2, \label{wo4}
\end{align}
 where \eqref{wo2} is by the lower bound in \eqref{main_eq:decomp}, \eqref{wo3} is thanks to \eqref{sflhdshfdshfshfdsjk2}, and \eqref{wo4} follows from \eqref{hihihs}. This establishes \eqref{bubu_main} and thereby completes the proof of i).

The proof of statement ii) is very similar to that of statement i). Let $f\in H^s(\R)\backslash\{ 0\}$ and $\eps \in (0,1)$ and note that, by \eqref{here1a}, \eqref{here2a}, \eqref{here1ab}, and \eqref{here2ab}, we have 
\begin{equation}\label{sflhdshfdshfshfdsjk2a}
W_N(f)\leq \frac{(2l)^{\gamma} \|f\|^2_{H^s}}{\delta^\gamma a^{\gamma (N-1)}}, \quad \forall s >0,
\end{equation}
where $\gamma = \min\{1,2s\}$. By Proposition \ref{prop1} in Appendix \ref{app:prop1} with $A_\Omega^N = B_\Omega^N =1$, $N \in \mathbb{N}$,  and $\lim\limits_{N\to \infty}W_N(f)=0$, $f\in L^2(\R)$, which follows from Proposition \ref{prop2} in Appendix \ref{bum}, we have \vspace{0.1cm}
\begin{align}
\| f \|_2^2&=|||\Phi_\Omega(f)|||^2=\sum_{n=0}^\infty||| \Phi_\Omega^n(f)|||^2\nonumber\\
&\geq\sum_{n=0}^N||| \Phi_\Omega^n(f)|||^2=\| f\|^2_2-W_{N+1}(f)\label{wo2a}  \\
&\geq \|f\|_2^2 - \frac{(2l)^{\gamma}\|f\|^2_{H^s}}{\delta^\gamma a^{\gamma N}}\label{wo3a}\\
&\geq\| f\|_2^2 - \varepsilon \| f\|_2^2 = (1-\varepsilon)\| f\|_2^2, \label{wo4a}
\end{align}
 where \eqref{wo2a} is by the lower bound in \eqref{main_eq:decomp}, \eqref{wo3a} is thanks to \eqref{sflhdshfdshfshfdsjk2a}, and \eqref{wo4a} follows from \eqref{hihihs2}. This establishes \eqref{bubu_main} and thereby completes the proof of ii). 

\section*{Acknowledgments}
\addcontentsline{toc}{section}{Acknowledgment}
The authors are thankful to Verner Vla\v{c}i{\'c} for pointing out errors in arguments on the spectral decay of Sobolev functions and on the volume of tubes in an earlier version of the manuscript.

\vspace{-1cm}
\begin{IEEEbiographynophoto}{Thomas Wiatowski} was born in Strzelce Opolskie, Poland, on December 20, 1987, and received the BSc in Mathematics and the MSc in Mathematics from Technical University of Munich, Germany, in 2010 and 2012, respectively. In 2012 he was a  researcher at the Institute of Computational Biology at the Helmholtz Zentrum in Munich, Germany. He joined ETH Zurich in 2013, where he graduated  with the Dr.\ sc.\ degree in 2017. His research interests are in deep machine learning, mathematical signal processing, and applied harmonic analysis.
\end{IEEEbiographynophoto}
\vspace{-1.cm}
\begin{IEEEbiographynophoto}{Philipp Grohs}  received the MSc and  PhD degrees in Applied Mathematics from 
Vienna University of Technology in 2006 and 2007, respectively.  He is currently an Associate Professor of Mathematics at the
University of Vienna. His research
interests are in signal processing, numerical analysis, computational
harmonic analysis, machine
learning, and computational geometry. He is the recipient of the 2014 ETH
Latsis Prize. Since 2017 he has been serving as an associate editor of the IEEE Transactions on Information Theory. He is a member of the board of directors of the Austrian Mathematical Society.
\end{IEEEbiographynophoto}

\vspace{-1.cm}
\begin{IEEEbiographynophoto}{Helmut B\"olcskei} was born in M\"odling, Austria, on May 29, 1970, and received the Dipl.-Ing.\ and Dr.\ techn.\ degrees in electrical engineering from Vienna University of Technology, Vienna, Austria, in 1994 and 1997, respectively. In 1998 he was with Vienna University of Technology. From 1999 to 2001 he was a postdoctoral researcher in the Information Systems Laboratory, Department of Electrical Engineering, and in the Department of Statistics, Stanford University, Stanford, CA. He was in the founding team of Iospan Wireless Inc., a Silicon Valley-based startup company (acquired by Intel Corporation in 2002) specialized in multiple-input multiple-output (MIMO) wireless systems for high-speed Internet access, and was a co-founder of Celestrius AG, Zurich, Switzerland. From 2001 to 2002 he was an Assistant Professor of Electrical Engineering at the University of Illinois at Urbana-Champaign. He has been with ETH Zurich since 2002, where he is a Professor of Electrical Engineering. He was a visiting researcher at Philips Research Laboratories Eindhoven, The Netherlands, ENST Paris, France, and the Heinrich Hertz Institute Berlin, Germany. His research interests are in information theory, mathematical signal processing, machine learning, and statistics.

He received the 2001 IEEE Signal Processing Society Young Author Best Paper Award, the 2006 IEEE Communications Society Leonard G. Abraham Best Paper Award, the 2010 Vodafone Innovations Award, the ETH ``Golden Owl'' Teaching Award, is a Fellow of the IEEE, a 2011 EURASIP Fellow, was a Distinguished Lecturer (2013-2014) of the IEEE Information Theory Society, an Erwin Schr\"odinger Fellow (1999-2001) of the Austrian National Science Foundation (FWF), was included in the 2014 Thomson Reuters List of Highly Cited Researchers in Computer Science, and is the 2016 Padovani Lecturer of the IEEE Information Theory Society. He served as an associate editor of the IEEE Transactions on Information Theory, the IEEE Transactions on Signal Processing, the IEEE Transactions on Wireless Communications, and the EURASIP Journal on Applied Signal Processing. He was editor-in-chief of the IEEE Transactions on Information Theory during the period 2010-2013. He served on the editorial board of the IEEE Signal Processing Magazine  and is currently on the editorial boards of ``Foundations and Trends in Networking'' and ``Foundations and Trends in Communications and Information Theory''.  He was TPC co-chair of the 2008 IEEE International Symposium on Information Theory and the 2016 IEEE Information Theory Workshop and currently serves as $2$nd vice president of the IEEE Information Theory Society. He has been a delegate of the president of ETH Zurich for faculty appointments since 2008.
\end{IEEEbiographynophoto}

\end{document}